\documentclass[journal]{IEEEtran}

\usepackage[utf8]{inputenc} 
\usepackage[T1]{fontenc}
\usepackage{url}
\usepackage{ifthen}
\usepackage{cite}
\usepackage[cmex10]{amsmath} 
\usepackage{calc}
\usepackage{graphicx}
\graphicspath{{figs/}{matlab/}{../figs/}}  
\usepackage{cite}
\usepackage{amsmath}
\usepackage{amsthm}
\usepackage{color}
\usepackage[usenames,dvipsnames]{pstricks}
\usepackage{epsfig}
\usepackage{pst-grad} 
\usepackage{pst-plot} 
\usepackage{upgreek} 
\usepackage{mathrsfs}
\usepackage{amssymb}
\usepackage{psfrag}
\usepackage{bbm}
\usepackage{xspace}
\usepackage{enumitem}
\usepackage{multicol}
\usepackage[font=footnotesize]{caption}
\usepackage{subcaption}
\usepackage{tikz}
\usepackage{bbm}
\usetikzlibrary{shapes.geometric, arrows}
\usepackage{stfloats}
\usepackage{mathtools}
\usepackage{cuted}

\usepackage{booktabs}
\newcommand{\tabitem}{~~\llap{\textbullet}~~}
        
\theoremstyle{plain}
\newtheorem{theorem}{Theorem}

\newtheorem{lemma}[theorem]{Lemma}

\theoremstyle{definition}

\newtheorem{remark}{Remark}[section]
\newtheorem{definition}{Definition}[section]

\newtheoremstyle{colon}%
{}
{}
{\itshape}
{}
{\bfseries}
{:}
{ }
{}
\theoremstyle{colon}
\newtheorem{lego}{Lego Brick}[section]

\newcommand{\cS}{\mathcal{S}}
\newcommand{\cY}{\mathcal{Y}}
\newcommand{\cX}{\mathcal{X}}

\newcommand{\cI}{\mathcal{I}}
\newcommand{\cC}{\mathcal{C}}

\newcommand{\cM}{\mathcal{M}}
\newcommand{\cN}{\mathcal{N}}

\newcommand{\cD}{\mathcal{D}}
\newcommand{\cK}{\mathcal{K}}
\newcommand{\cL}{\mathcal{L}}
\newcommand{\IID}{\overset{\mathrm{iid}}{\sim}}

\newcommand{\BERN}{\mathrm{Bern}}
\newcommand{\BSC}{\mathrm{BSC}}
\newcommand{\Unif}{\mathrm{Unif}}
\newcommand{\Binom}{\mathrm{Binom}}
\newcommand{\weight}{\mathrm{wt}}
\newcommand{\wtilde}[1]{\widetilde{#1}}
\newcommand{\td}[1]{\tilde{#1}}
\newcommand{\abs}[1]{{\left\lvert #1 \right\rvert}}

\tikzstyle{box} = [rectangle, rounded corners, minimum width=3.5cm, minimum height=0.9cm, text centered, text width=3.5cm, draw=black, fill=black!15]
\tikzstyle{smallbox} = [rectangle, rounded corners, minimum width=2.5cm, minimum height=0.9cm, text centered, text width=2.5cm, draw=black, fill=black!15]
\tikzstyle{p2p-box} = [rectangle, rounded corners, minimum width=4cm, minimum height=0.9cm, text centered, text width=4cm, draw=black, fill=black!15]
\tikzstyle{wyner-box} = [rectangle, rounded corners, minimum width=2.5cm, minimum height=0.9cm, text centered, text width=2.5cm, draw=black, fill=black!15]
\tikzstyle{mac-box} = [rectangle, rounded corners, minimum width=3.2cm, minimum height=0.9cm, text centered, text width=3.2cm, draw=black, fill=black!15]
\tikzstyle{marton-box} = [rectangle, rounded corners, minimum width=3.3cm, minimum height=0.9cm, text centered, text width=3.3cm, draw=black, fill=black!15]
\tikzstyle{mdc-box} = [rectangle, rounded corners, minimum width=3cm, minimum height=0.9cm, text centered, text width=3cm, draw=black, fill=black!15]
\tikzstyle{downlink-cran-box} = [rectangle, rounded corners, minimum width=3cm, minimum height=0.9cm, text centered, text width=3cm, draw=black, fill=black!15]
\tikzstyle{uplink-cran-box} = [rectangle, rounded corners, minimum width=2.7cm, minimum height=0.9cm, text centered, text width=2.7cm, draw=black, fill=black!15]
\tikzstyle{slepian-box} = [rectangle, rounded corners, minimum width=2.3cm, minimum height=0.9cm, text centered, text width=2.3cm, draw=black, fill=black!15]
\tikzstyle{lossless-box} = [rectangle, rounded corners, minimum width=2.3cm, minimum height=0.9cm, text centered, text width=2.3cm, draw=black, fill=black!15]
\tikzstyle{arrow} = [thick,->,>=stealth]

\makeatletter
\let\save@mathaccent\mathaccent
\newcommand*\if@single[3]{%
  \setbox0\hbox{${\mathaccent"0362{#1}}^H$}%
  \setbox2\hbox{${\mathaccent"0362{\kern0pt#1}}^H$}%
  \ifdim\ht0=\ht2 #3\else #2\fi
  }
\newcommand*\rel@kern[1]{\kern#1\dimexpr\macc@kerna}
\newcommand*\widebar[1]{\@ifnextchar^{{\wide@bar{#1}{0}}}{\wide@bar{#1}{1}}}
\newcommand*\wide@bar[2]{\if@single{#1}{\wide@bar@{#1}{#2}{1}}{\wide@bar@{#1}{#2}{2}}}
\newcommand*\wide@bar@[3]{%
  \begingroup
  \def\mathaccent##1##2{%
    \let\mathaccent\save@mathaccent
    \if#32 \let\macc@nucleus\first@char \fi
    \setbox\z@\hbox{$\macc@style{\macc@nucleus}_{}$}%
    \setbox\tw@\hbox{$\macc@style{\macc@nucleus}{}_{}$}%
    \dimen@\wd\tw@
    \advance\dimen@-\wd\z@
    \divide\dimen@ 3
    \@tempdima\wd\tw@
    \advance\@tempdima-\scriptspace
    \divide\@tempdima 10
    \advance\dimen@-\@tempdima
    \ifdim\dimen@>\z@ \dimen@0pt\fi
    \rel@kern{0.6}\kern-\dimen@
    \if#31
      \overline{\rel@kern{-0.6}\kern\dimen@\macc@nucleus\rel@kern{0.4}\kern\dimen@}%
      \advance\dimen@0.4\dimexpr\macc@kerna
      \let\final@kern#2%
      \ifdim\dimen@<\z@ \let\final@kern1\fi
      \if\final@kern1 \kern-\dimen@\fi
    \else
      \overline{\rel@kern{-0.6}\kern\dimen@#1}%
    \fi
  }%
  \macc@depth\@ne
  \let\math@bgroup\@empty \let\math@egroup\macc@set@skewchar
  \mathsurround\z@ \frozen@everymath{\mathgroup\macc@group\relax}%
  \macc@set@skewchar\relax
  \let\mathaccentV\macc@nested@a
  \if#31
    \macc@nested@a\relax111{#1}%
  \else
    \def\gobble@till@marker##1\endmarker{}%
    \futurelet\first@char\gobble@till@marker#1\endmarker
    \ifcat\noexpand\first@char A\else
      \def\first@char{}%
    \fi
    \macc@nested@a\relax111{\first@char}%
  \fi
  \endgroup
}
\makeatother

\makeatletter
\newcommand*{\rom}[1]{\expandafter\@slowromancap\romannumeral #1@}
\makeatother

\usepackage{xspace}
\usepackage{bbm}
\catcode`@=11
\chardef\mathlig@atcode\count255

\def\actively#1#2{\begingroup\uccode`\~=`#2\relax\uppercase{\endgroup#1~}}
\def\mathlig@gobble{\afterassignment\mathlig@next@cmd\let\mathlig@next= }
\def\mathlig@delim{\mathlig@delim}
\def\mathlig@defcs#1{\expandafter\def\csname#1\endcsname}
\def\mathlig@let@cs#1#2{\expandafter\let\expandafter#1\csname#2\endcsname}
\def\mathlig@appendcs#1#2{\expandafter\edef\csname#1\endcsname{\csname#1\endcsname#2}}

\def\mathlig#1#2{\mathlig@checklig#1\mathlig@end\mathlig@defcs{mathlig@back@#1}{#2}\ignorespaces}

\def\mathlig@checklig#1#2\mathlig@end{%
 \expandafter\ifx\csname mathlig@forw@#1\endcsname\relax
 \expandafter\mathchardef\csname mathlig@back@#1\endcsname=\mathcode`#1%
 \mathcode`#1"8000\actively\def#1{\csname mathlig@look@#1\endcsname}%
 \mathlig@dolig#1\mathlig@delim
\fi
\mathlig@checksuffix#1#2\mathlig@end
}

\def\mathlig@checksuffix#1#2\mathlig@end{%
\ifx\mathlig@delim#2\mathlig@delim\relax\else\mathlig@checksuffix@{#1}#2\mathlig@end\fi
}
\def\mathlig@checksuffix@#1#2#3\mathlig@end{%
\expandafter\ifx\csname mathlig@forw@#1#2\endcsname\relax\mathlig@dosuffix{#1}{#2}\fi
\mathlig@checksuffix{#1#2}#3\mathlig@end
}

\def\mathlig@dosuffix#1#2{%
\mathlig@appendcs{mathlig@toks@#1}{#2}%
\mathlig@dolig{#1}{#2}\mathlig@delim
}

\def\mathlig@dolig#1#2\mathlig@delim{%
 \mathlig@defcs{mathlig@look@#1#2}{%
 \mathlig@let@cs\mathlig@next{mathlig@forw@#1#2}\futurelet\mathlig@next@tok\mathlig@next}%
 \mathlig@defcs{mathlig@forw@#1#2}{%
  \mathlig@let@cs\mathlig@next{mathlig@back@#1#2}%
  \mathlig@let@cs\checker{mathlig@chck@#1#2}%
  \mathlig@let@cs\mathligtoks{mathlig@toks@#1#2}%
  \expandafter\ifx\expandafter\mathlig@delim\mathligtoks\mathlig@delim\relax\else
  \expandafter\checker\mathligtoks\mathlig@delim\fi
  \mathlig@next
 }%
 \mathlig@defcs{mathlig@toks@#1#2}{}%
 \mathlig@defcs{mathlig@chck@#1#2}##1##2\mathlig@delim{%
  \ifx\mathlig@next@tok##1%
   \mathlig@let@cs\mathlig@next@cmd{mathlig@look@#1#2##1}\let\mathlig@next\mathlig@gobble
  \fi 
  \ifx\mathlig@delim##2\mathlig@delim\relax\else
   \csname mathlig@chck@#1#2\endcsname##2\mathlig@delim
  \fi
 }%
 \ifx\mathlig@delim#2\mathlig@delim\else
  \mathlig@defcs{mathlig@back@#1#2}{\csname mathlig@back@#1\endcsname #2}%
 \fi
}%

\catcode`@\mathlig@atcode

\usepackage{mathrsfs}
\newcommand{\muspace}{\mspace{1mu}}

\DeclareRobustCommand{\scond}{\mathchoice{\muspace\vert\muspace}{\vert}{\vert}{\vert}}
\mathlig{|}{\scond}
\newcommand{\cond}{\mathchoice{\,\vert\,}{\mspace{2mu}\vert\mspace{2mu}}{\vert}{\vert}}

\DeclareRobustCommand{\discint}{\mathchoice{\mspace{-1.5mu}:\mspace{-1.5mu}}{\mspace{-1.5mu}:\mspace{-1.5mu}}{:}{:}}
\mathlig{::}{\discint}

\def\P{{\mathsf P}}

\def\eps{\epsilon}

\DeclareMathOperator\E{\textsf{E}}
\let\P\relax
\DeclareMathOperator\P{\textsf{P}}

\newcommand{\Bern}{\mathrm{Bern}}

\def\textiid{i.i.d.\@\xspace}
\newcommand\iid{\ifmmode\text{ i.i.d. } \else \textiid \fi}

\def\mathllap{\mathpalette\mathllapinternal}
\def\mathllapinternal#1#2{%
  \llap{$\mathsurround=0pt#1{#2}$}}

\def\clap#1{\hbox to 0pt{\hss#1\hss}}
\def\mathclap{\mathpalette\mathclapinternal}
\def\mathclapinternal#1#2{%
  \clap{$\mathsurround=0pt#1{#2}$}}

\let\oldstackrel\stackrel
\renewcommand{\stackrel}[2]{\oldstackrel{\mathclap{#1}}{#2}}

\renewcommand{\hbar}{h\mathllap{\overline{\vphantom{h}\hphantom{\rule{4.6pt}{0pt}}}\mspace{0.77mu}}}

\catcode`~=11 
\newcommand{\urltilde}{\kern -.06em\lower -.06em\hbox{~}\kern .02em}
\catcode`~=13 

\begin{document}
\allowdisplaybreaks

\title{A Lego-Brick Approach to Coding \\ for Network Communication}

\author{Nadim~Ghaddar,~\IEEEmembership{Student Member,~IEEE,}
        Shouvik~Ganguly,~\IEEEmembership{Member,~IEEE,}
        Lele~Wang,~\IEEEmembership{Member,~IEEE,}
        and~Young-Han~Kim,~\IEEEmembership{Fellow,~IEEE}
\thanks{This paper was presented, in part, at the Annual Allerton Conference on Communication, Control, and Computing 2015, the International Symposium on Information Theory (ISIT) 2021, and the Canadian Workshop on Information Theory (CWIT) 2022.}
\thanks{N. Ghaddar is with the Department of Electrical and Computer Engineering, University of California San Diego, La Jolla, CA 92093, USA (e-mail: nghaddar@ucsd.edu).}
\thanks{S. Ganguly is with Samsung Research America, Plano, TX 75023, USA (email: shouvik.g@samsung.com)}
\thanks{L. Wang is with the Department of Electrical and Computer Engineering, University of British Columbia, Vancouver, BC V6T1Z4, Canada (email: lelewang@ece.ubc.ca)}
\thanks{Y.-H. Kim is with the Department of Electrical and Computer Engineering, University of California San Diego, La Jolla, CA 92093, USA (email: yhk@ucsd.edu)}}

%

\maketitle

\begin{abstract}
Coding schemes for several problems in network information theory are constructed starting from point-to-point channel codes that are designed for symmetric channels. Given that the point-to-point codes satisfy certain properties pertaining to the rate, the error probability, and the distribution of decoded sequences, bounds on the performance of the coding schemes are derived and shown to hold irrespective of other properties of the codes. In particular, we consider the problems of lossless and lossy source coding, Slepian--Wolf coding, Wyner--Ziv coding, Berger--Tung coding, multiple description coding, asymmetric channel coding, Gelfand--Pinsker coding, coding for multiple access channels, Marton coding for broadcast channels, and coding for cloud radio access networks (C-RAN's). We show that the coding schemes can achieve the best known inner bounds for these problems, provided that the constituent point-to-point channel codes are rate-optimal. This would allow one to leverage commercial off-the-shelf codes for point-to-point symmetric channels in the practical implementation of codes over networks. Simulation results demonstrate the gain of the proposed coding schemes compared to existing practical solutions to these problems.
\end{abstract}


\IEEEpeerreviewmaketitle

\section{Introduction} \label{sec:intro}
\IEEEPARstart{T}{oday's} modern infrastructure is becoming increasingly interconnected through information networks. Emerging applications in transportation systems, power systems, smart cities, cloud computing and digital healthcare, call for exceptionally efficient coding schemes to process, store and communicate the massive amounts of network data. In practice, the common paradigm of designing coding schemes over networks is to decompose the network into separate point-to-point links, where a point-to-point channel code is used over each link. Despite the practical convenience, such an approach is known to be sub-optimal from an information theoretic perspective, even when each link is utilized at its full capacity.

Network information theory studies the fundamental limits of network communication and the optimal coding schemes that achieve those limits. At a conceptual level, this theory has been hugely successful, with several basic coding schemes that are applicable to a variety of network models, and some optimal in certain special cases. Except for a few simple use cases, however, the coding schemes developed in network information theory have barely had any impact on the design of communication systems over networks. Even basic coding schemes such as Gelfand--Pinsker coding~\cite{Gelfand1980}, Marton coding~\cite{Marton1979} and compress-and-forward relaying~\cite{Cover1979} have not been used in practice in any meaningful manner in the forty years since their inception. The main reason behind this noticeable gap between theory and practice is that most of these coding schemes, albeit being conceptually beautiful, are not in an easily-implementable form, which is exemplified by the ubiquitous use of high-complexity coding techniques such as joint typicality encoding and decoding, and maximum likelihood decoding.

This paper is an attempt to close the aforementioned gap between theory and practice, so that all the beautiful coding schemes in network information theory (e.g., all the ones described in~\cite{NIT}) can be implemented in real systems to their full potential. To achieve this goal, we take a modular approach to transform the conceptual coding schemes into practical implementations. More specifically, we seek to identify basic coding schemes that are designed for one (or more) communication setting and satisfy certain properties, and combine them together to build a more complex coding scheme for a different communication setting. Under such a framework, we ask: 
\begin{itemize}
    \item What are the most primitive properties that the basic coding schemes should satisfy while being versatile in building coding schemes for network communication?
    \item How can such coding blocks be assembled together in different network communication scenarios?
    \item How do the performance guarantees and the achievable rate regions translate in different communication settings?
\end{itemize}%
We refer to such an approach to coding as a ``Lego-brick approach'' and to the basic coding blocks as ``Lego bricks'', taking a literal analogy to building complex Lego objects starting from simple components. As we shall see, such a Lego-brick approach to coding allows one to leverage commercial off-the-shelf codes that are designed for single-user symmetric channels (e.g., turbo codes~\cite{Berrou1993}, low-density parity-check (LDPC) codes~\cite{Gallager1962, Urbanke2001, Urbanke2013}, polar codes~\cite{Arikan2009}, or even hypothetical codes to be invented in the future) to build practical coding schemes for multi-user communication. 

\subsection{Related Work}
Many attempts to design practical coding schemes for multi-terminal scenarios have closely followed the footsteps of point-to-point channel coding. Polar codes, for example, have been specialized to several problems in network information theory, including, but not limited to, the Slepian--Wolf problem~\cite{Arikan2012}, the lossy source coding problem~\cite{Korada2010}, Gelfand--Pinsker problem~\cite{Korada2009}, the multiple description coding problem~\cite{Shi2015, Ghaddar2017}, multiple-access channels~\cite{Sasoglu2010,Abbe2012}, broadcast channels~\cite{Mondelli2015}, interference channels~\cite{Wang2014}, and relay channels~\cite{Serrano2012,Wang2015_2}. Sparse graph codes with logarithmic check node degrees have also been shown to achieve the optimal rates for various coding problems under maximum likelihood decoding, including the lossy source coding problem~\cite{Matsunaga2003}, the Gelfand--Pinsker problem and the Wyner--Ziv problem~\cite{Muramatsu2010}. Alternatively, low-density generator-matrix (LDGM) codes were shown to achieve the rate-distortion bound for the lossy source coding problem when optimal encoding is employed~\cite{Wainwright2010}, and to approach that bound when a variant of the low-complexity belief propagation algorithm is used~\cite{Wainwright2005}. Spatial coupling of low-density graph codes was also shown to be useful for approaching the optimal rates of the lossy source coding problem~\cite{Aref2015} and the Wyner--Ziv and Gelfand--Pinsker coding problems~\cite{Kumar2014} using the belief propagation algorithm with guided decimation. For Gaussian channels with Gaussian state that is known noncausally at the encoder (i.e., the dirty paper coding problem~\cite{Costa1983}), lattice codes have been shown to achieve capacity~\cite{Erez2005}, and variants of these codes with practical decoders have been proposed in the literature (e.g.,~\cite{Sommer2008,Shalvi2011}). On a separate note, randomly-constructed point-to-point channel codes were shown to be sufficient for achieving optimal rates for network coding problems; for example, random nested linear codes were shown to be optimal for asymmetric channels, Gelfand--Pinkser channels, and general broadcast channels~\cite{Pradhan2011}, and random point-to-point codes were used in~\cite{Lele2020} along with rate splitting, block-Markov coding and successive cancellation decoding to achieve the best known achievable rate over interference channels. While these constructions either are specific to certain families of point-to-point channel codes or involve nonconstructive and asymptotic code design, the Lego-brick approach to network coding, which we espouse in this paper, aims to distinguish the sufficient properties that would allow point-to-point channel codes to be used as building blocks in an \emph{explicit} construction of a coding scheme for a network setting.

Such an approach to coding for network communication is not new in general. In~\cite{Wyner1974}, Wyner constructed a Slepian--Wolf code for a doubly symmetric binary source starting from a point-to-point channel code designed for the binary symmetric channel (BSC). The duality between general Slepian--Wolf problems and point-to-point channel coding problems was further explored in~\cite{Ancheta1976,Yang2009}, where a maximum-likelihood channel decoder was assumed. In~\cite{Lele2015}, linear Slepian--Wolf codes were constructed starting from linear point-to-point channel codes designed for symmetric channels, where the exact relation of the rates and probability of error between the two problems was established. Further, a general method for constructing codes for asymmetric point-to-point channels was described in~\cite[Section~\rom{5}]{Mondelli2018}; the coding scheme uses a lossless source code and a channel code designed for a symmetric channel as its constituent building blocks. In~\cite{Muramatsu2014,Muramatsu2019}, codes for general (i.e., possibly asymmetric) channels were constructed using Slepian--Wolf codes, where constrained random number generators were used to construct the encoders and the decoders. Moreover, Marton coding for broadcast channels and successive cancellation decoding of multiple access channels were constructed in~\cite{Ganguly2020} starting from point-to-point codes that satisfy some basic properties. In particular, in our recent works of~\cite{Ghaddar2022} and~\cite{Ghaddar2021}, coding schemes for lossy source coding, asymmetric channels and channels with state were constructed starting from point-to-point channel codes designed for symmetric channels. 


\subsection{Summary of Proposed Coding Schemes}
This paper is an attempt to establish a systematic framework for constructing coding schemes for several problems in network information theory starting from simple building blocks (i.e., Lego bricks). Instead of restricting our attention to a particular channel code, we identify basic properties that point-to-point channel codes should satisfy so that they can be used in the construction of a coding scheme for a given network information theory problem. Under such a framework, we view the channel codes that satisfy those properties as ``black boxes'' which, when assembled together according to the guidelines that we will present, can be used as building blocks in the implementations of practical low-complexity coding schemes for network communication. The performance of the coding schemes in the new setting is derived in terms of the properties of the constituent point-to-point channel codes. Moreover, the proposed coding schemes can achieve rate regions that are strictly larger compared to the naive approach of coding over the point-to-point links of a network. In particular, we consider the problems of Slepian--Wolf coding~\cite{Slepian1973}, lossless and lossy source coding~\cite{Shannon1948, Shannon1959}, Wyner--Ziv coding~\cite{Wyner1976}, asymmetric channel coding~\cite{Shannon1948}, Gelfand--Pinkser coding~\cite{Gelfand1980}, coding for multiple access channels~\cite{Shannon1961}, Marton coding for broadcast channels~\cite{Marton1979},  Berger--Tung coding for distributed lossy compression~\cite{Berger1978}, multiple description coding~\cite{Witsenhausen1980}, and coding for cloud radio access networks~\cite{Simeone2016}. 

We identify two primitive properties of codes designed for point-to-point symmetric channels which allows to translate the performance of these codes to network communication settings. These two properties are the \emph{probability of error} of the code and its \emph{shaping distance}. Whereas the probability of error is a measure of the error correction capability of the code when simulated over the symmetric channel, the shaping distance is a measure of the shaping capability of the decoding function\footnote{See Section~\ref{sec:properties} for a formal definition of the shaping distance.}. Bounds on the performance of our code constructions will be derived in terms of these two properties regardless of other properties of the code. Such flexibility allows us to be tightly coupled with the most recent development in coding theory for point-to-point communication (in terms of performance), but at the same time to be completely decoupled from it (in terms of architecture). 


Fig.~\ref{fig:summary} illustrates the different coding problems considered in this paper, along with their constituent Lego bricks. Arrows pointing from a set of coding problems towards another coding problem means that codes for these problems can be used as building blocks in the design of a code for the designated problem. More precisely, coding schemes for the following problems will be constructed starting from the following codes.

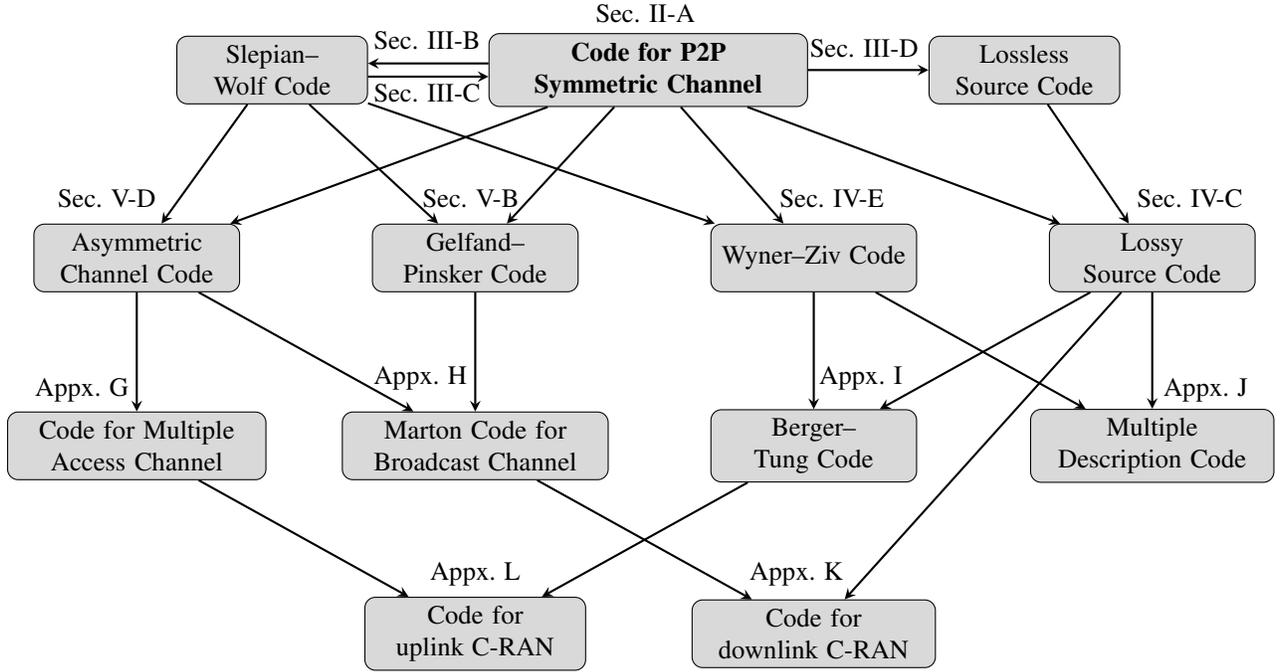
\begin{figure*}[t]
	\centering
    \hspace*{1em}
	\begin{tikzpicture}[node distance=2.5cm]
		\node (p2p) [p2p-box] {\textbf{Code for P2P Symmetric Channel}};
		\node (sw) [slepian-box, left of=p2p, xshift=-2.5cm] {Slepian--Wolf Code};
		\node (lossless) [lossless-box, right of=p2p, xshift=2.5cm] {Lossless Source Code};
		\node (gp) [smallbox, below of=p2p, xshift=-2.3cm] {Gelfand--Pinsker Code};
		\node (asym) [smallbox, left of=gp, xshift=-2cm] {Asymmetric Channel Code};
		\node (wz) [wyner-box, right of=gp, xshift=2cm] {Wyner--Ziv Code};
		\node (lossy) [smallbox, right of=wz, xshift=2cm] {Lossy Source Code};
		\node (mac) [mac-box, below of=asym] {Code for Multiple Access Channel};
		\node (marton) [marton-box, below of=gp] {Marton Code for Broadcast Channel};
		\node (bt) [wyner-box, below of=wz] {Berger--Tung Code};
		\node (mdc) [mdc-box, below of=lossy] {Multiple\\Description Code};
		\node (downlink-cran) [downlink-cran-box, below of=bt] {Code for\\ downlink C-RAN};
		\node (uplink-cran) [uplink-cran-box, below of=marton] {Code for\\ uplink C-RAN};
		\node[text width=2cm,fill=none] at (0.3,0.75) {Sec.~\ref{sec:p2p}};
		\node[text width=2cm,fill=none] at (-2.65,0.4) {Sec.~\ref{sec:p2p_slepian}};
		\node[text width=2cm,fill=none] at (-2.65,-0.3) {Sec.~\ref{sec:slepian_p2p}};
		\node[text width=2cm,fill=none] at (3.15,0.3) {Sec.~\ref{sec:lossless}};
		\node[text width=2cm,fill=none] at (-6.85,-1.7) {Sec.~\ref{sec:asym}};
		\node[text width=2cm,fill=none] at (-2,-1.7) {Sec.~\ref{sec:gelfand_coding_scheme}};
		\node[text width=2cm,fill=none] at (2.75,-1.7) {Sec.~\ref{sec:wyner}};
		\node[text width=2cm,fill=none] at (7.5,-1.7) {Sec.~\ref{sec:lossy_asym}};
		\node[text width=2cm,fill=none] at (-7.15,-4.25) {Appx.~\ref{sec:mac}};
		\node[text width=2cm,fill=none] at (-2.65,-4.1) {Appx.~\ref{sec:marton}};
		\node[text width=2cm,fill=none] at (3.275,-4.1) {Appx.~\ref{sec:berger}};
		\node[text width=2cm,fill=none] at (7.85,-4.25) {Appx.~\ref{sec:mdc}};
		\node[text width=2cm,fill=none] at (-1.9,-6.7) {Appx.~\ref{sec:uplink-cran}};
		\node[text width=2cm,fill=none] at (2.35,-6.7) {Appx.~\ref{sec:cran}};
		
		\draw[transform canvas={yshift=0.25em}] [arrow] (p2p) -- (sw);
		\draw[transform canvas={yshift=-0.25em}] [arrow] (sw) -- (p2p);
		\draw [arrow] (p2p) -- (lossless);
		\draw [arrow] (p2p) -- (gp);
		\draw [arrow] (sw) -- (gp);
		\draw [arrow] (p2p) -- (wz);
		\draw [arrow] (sw) -- (wz);
		\draw [arrow] (p2p) -- (asym);
		\draw [arrow] (sw) -- (asym);
		\draw [arrow] (p2p) -- (lossy);
		\draw [arrow] (lossless) -- (lossy);
		\draw [arrow] (lossy) -- (bt);
		\draw [arrow] (wz) -- (bt);
		\draw [arrow] (lossy) -- (mdc);
		\draw [arrow] (wz) -- (mdc);
		\draw [arrow] (asym) -- (mac);
		\draw [arrow] (asym) -- (marton);
		\draw [arrow] (gp) -- (marton);
		\draw [arrow] (marton) -- (downlink-cran);
		\draw [arrow] (lossy) -- (downlink-cran);
		\draw [arrow] (mac) -- (uplink-cran);
		\draw [arrow] (bt) -- (uplink-cran);
	\end{tikzpicture}
	\caption{Overview of the proposed coding schemes and their constituent building blocks, along with the section number in which each coding scheme is presented.}
	\label{fig:summary}
\end{figure*}

\begin{enumerate}
	\item Code for point-to-point symmetric channel $\to$ Lossless source code
	\item Code for point-to-point symmetric channel $\to$ Slepian--Wolf code
	\item Slepian--Wolf code $\to$ Code for point-to-point symmetric channel
	\item Code for point-to-point symmetric channel + Lossless source code $\to$ Lossy source code
	\item Code for point-to-point symmetric channel + Slepian--Wolf code $\to$ Wyner--Ziv code
	\item Code for point-to-point symmetric channel + Slepian--Wolf code $\to$ Gelfand--Pinsker code
	\item Code for point-to-point symmetric channel + Slepian--Wolf code $\to$ Asymmetric channel code
	\item Two asymmetric channel codes $\to$ Code for the two-user multiple access channel
	\item Asymmetric channel code + Gelfand--Pinsker code $\to$ Marton code for the two-user broadcast channel
	\item Lossy source code + Wyner--Ziv code $\to$ Berger--Tung code
	\item Two lossy source codes + Wyner--Ziv code $\to$ Multiple description code
	\item Two lossy source codes + Marton code for broadcast channel $\to$ Code for the downlink C-RAN problem
	\item Multiple access channel code + Berger--Tung code $\to$ Code for the uplink C-RAN problem
\end{enumerate}
It follows that all the coding schemes can be constructed starting from codes for point-to-point symmetric channels. As an example, since a Marton code for the two-user broadcast channel can be constructed starting from an asymmetric channel code and a Gelfand--Pinsker code, each of which can be constructed from a Slepian--Wolf code and a point-to-point channel code for a symmetric channel, and since the Slepian--Wolf code can be constructed from a point-to-point symmetric channel code, it follows that a Marton code can be constructed using four channel codes that are designed for point-to-point symmetric channels. Similarly, all other coding schemes follow from point-to-point channel coding. Such an approach enables one to leverage commercial off-the-shelf codes (such as those studied in~\cite{Berrou1993,Gallager1962, Urbanke2001, Urbanke2013, Arikan2009}) for point-to-point channels to build codes for the aforementioned coding problems in network information theory. Provided that the point-to-point channel codes are rate-optimal, the proposed constructions can achieve the best known inner bounds for these coding problems.


In several of the constructions that we will present, it will be assumed that two of the constituent linear point-to-point channel codes are \emph{nested}. Such a nested linear structure is helpful to construct low-complexity coding schemes for network communication, as previously pointed out in the literature (e.g.,~\cite{Pradhan2011,Zamir2002}). Nonetheless, it turns out that the nestedness requirement is not necessary, as we shall see later on in the paper (Section~\ref{sec:block_markov}), where, at the price of slightly larger implementation complexity, the constructions will be modified to account for the relaxed assumptions on the constituent point-to-point channel codes. One important implication of these modified constructions is the following: \emph{any} set of sequences of point-to-point channel codes that have a vanishing error probability over a symmetric channel can be combined with \emph{any} set of sequences of point-to-point channel codes that have a vanishing shaping distance to construct coding schemes that achieve the best known inner bounds for the coding problems considered in this paper. As a consequence of this observation, it follows that polar codes with successive cancellation decoding~\cite{Arikan2009} can achieve the best known inner bounds for the aforementioned coding problems using the proposed constructions.

We conclude this section with two interesting observations regarding the proposed constructions. First, one can see some form of duality between some of the constructions. For example, the encoder of the asymmetric channel coding scheme looks almost identical to the decoder of the lossy source coding scheme, and vice versa. Such a relation is well-known in the literature and has been pointed out in~\cite{Korada2010} and~\cite{Berger1998}. Nonetheless, note that, in our construction, the asymmetric channel encoder has a ``shaping role'', whereas the lossy source decoder has an ``error correction role''. Although not mathematically precise, this suggests some form of duality between error correction and shaping in our constructions, which can be seen as analogous to the well-understood packing-and-covering duality in network information theory. Similar observations can be made between the Gelfand--Pinkser coding scheme and the Wyner--Ziv coding scheme, between the Marton coding scheme and the Berger--Tung coding scheme, as well as between the downlink and uplink C-RAN code constructions. Secondly, as will be apparent from most of our constructions, randomization plays an important role in the design of coding schemes for network communication. The Lego bricks will often be used alongside a ``dithering'' brick, which generates a binary sequence uniformly at random, or an ``interleaver'' brick, which applies a permutation chosen uniformly at random to a binary sequence. Such randomness can shape signals into a desired structure, and will often be assumed to be shared between the encoders and the decoders. Note that sharing randomness between the transmitter and the receiver is a common practice in wireless communications, e.g., sharing the (pseudo)random spreading code in code-division multiple access (CDMA) systems~\cite{3gpp.cdma} or sharing the interleaving pattern in 5G NR systems~\cite{3gpp.38.212}.

\subsection{Paper Organization and Notation}
The rest of the paper is organized as follows. Section~\ref{sec:preliminaries} describes the problem of point-to-point channel coding, along with two primitive properties of point-to-point channel codes that we will focus on throughout the paper, namely, the error probability and the shaping distance. Section~\ref{sec:slepian} discusses the Slepian--Wolf coding problem, along with its specialization to the lossless source coding problem. More specifically, an ``equivalence'' is established between designing a code for the binary Slepian--Wolf problem and designing a point-to-point channel code for a symmetric binary-input channel. In Section~\ref{sec:lossy}, coding schemes for the symmetric and asymmetric lossy source coding problems are described, which can be extended to the case when side information is available at the decoder (i.e., the Wyner--Ziv problem). Section~\ref{sec:gelfand} describes a coding scheme for the binary-input Gelfand--Pinsker problem, along with its specialization to asymmetric point-to-point channel coding. Simulation results are presented for lossy source coding of an asymmetric source and for Gelfand--Pinsker coding in Sections~\ref{sec:lossy_simulation} and~\ref{sec:gelfand_simulation}, respectively. In Section~\ref{sec:block_markov}, modified constructions of the coding schemes are presented, which not only do not require any nested structure between the constituent Lego bricks, but also involve properties that are easily verifiable in practice. The modified constructions have a block-Markov structure, i.e., they are defined based upon the encoding and decoding of multiple message blocks, where the input to one coding block can depend on the outputs of previous/subsequent blocks. To complete the picture, code constructions for the remaining multiterminal source and channel coding problems are provided in the appendices. Appendix~\ref{sec:mac} considers coding for multiple access channels, and Appendix~\ref{sec:marton} describes a Marton coding scheme for broadcast channels. As for multiterminal lossy compression, Appendices~\ref{sec:berger} and~\ref{sec:mdc} present coding schemes for the Berger--Tung problem and the multiple description coding problem, respectively. Appendices~\ref{sec:cran} and~\ref{sec:uplink-cran} focus on the downlink and uplink C-RAN problems respectively, which are multihop networks entailing both source and channel coding counterparts. The code constructions in the appendices are largely inspired by their random coding counterparts as well as the previous constructions in the main sections of the paper, yet we choose to include them in the paper as they help to provide a comprehensive treatment of the Lego-brick approach to coding over networks. Simulation results are also provided in Appendices~\ref{sec:marton}, \ref{sec:cran} and \ref{sec:uplink-cran} for the Marton coding problem, and the downlink and uplink C-RAN problems, respectively, which can be of independent interest for both theoreticians and practitioners in the field. Section~\ref{sec:conclusion} concludes the paper and gives some insights towards potential future research directions. 

Notation: The following notation will be used. We write $x_1^n$ to denote a column vector $(x_1, \ldots, x_n)^T$, and $x_i^j$, for $i < j$, to denote the subvector $(x_i,\ldots,x_j)^T$. Sets are denoted in calligraphic letters, random variables in upper-case letters, and realizations of random variables in lower-case letters, i.e., $\cX$ and $\cY$ are two sets, $X$ and $Y$ are random variables, and $x$ and $y$ are their respective realizations. The distribution of a random variable $X$ will be denoted as $p_X = p_X(x)$, $x \in \cX$. When clear from the context, we will also often use the shorthand notation $p(x)$ to mean $p_X(x)$. The probability of an event $A$ is denoted as $\P\{A\}$. For a positive integer $n$, we use $[n]$ to denote the set $\{1,\ldots,n\}$. For a finite set $\cS$, we write $\abs{\cS}$ to denote the cardinality of $\cS$. $H(X)$ is used to denote the entropy of a random variable $X$ (in base 2). When $X\sim \BERN(\theta)$, we write $H(\theta)$ to denote $H(X)$. With a slight abuse of notation, we will also use $H$ to denote a parity-check matrix of a linear code. The mutual information of two random variables $X$ and $Y$ is written as $I(X;Y)$. The symmetric capacity of a binary-input channel $p(y \cond x)$ is written as $I(\BERN(1/2), p(y \cond x))$.

\section{Preliminaries} \label{sec:preliminaries}
\subsection{Basic Lego Brick: Code for a Point-to-Point Symmetric Channel} \label{sec:p2p}
Consider a binary-input discrete memoryless channel $p(y\cond x)$ with an input alphabet $\cX=\{0,1\}$, an output alphabet $\cY$ and a collection of conditional probability mass functions (pmf's)
$p(y\cond x)$ on $\cY$ for each $x\in \cX$. A $(k,n)$ point-to-point channel code $(f,\phi)$ for the channel $p(y\cond x)$ consists of
\begin{itemize}
	\item a codebook $\cC \subseteq \{0,1\}^n$ of size $|\cC| = 2^{k}$,
	\item an encoder $f: [2^{k}] \to \cC$ that maps each message $m \in [2^{k}]$ to a codeword $x^n = f(m)$,
	\item a decoder $\phi: \cY^n\to \cC$ that assigns a codeword estimate $\hat{x}^n = \phi(y^n)$ to each received sequence $y^n$.
\end{itemize}
The rate of the code is $R=k/n$. We say that the channel code is \emph{linear} if for any two codewords $c^n, \tilde{c}^n \in \cC$, we have $c^n \oplus \tilde{c}^n \in \cC$. A linear code can be alternatively defined by its parity-check matrix $H_{(n-k) \times n}$ and its decoding function $\phi$. In this case, the codebook can be written as $\cC = \{c^n: Hc^n = 0^{n-k}\}$. With a slight abuse of notation, when the code is linear and its parity-check matrix is $H$, we will refer to it as a $(k,n)$ linear point-to-point channel code $(H,\phi)$.

\begin{remark} \label{remark:H_tilde}
Given any $(n-k)\times n$ parity-check matrix $H$, we will assume throughout the paper that the last $n-k$ columns of $H$ are linearly independent, i.e., $H = \begin{bmatrix}A & B\end{bmatrix}$ for some nonsingular $(n-k)\times (n-k)$ matrix $B$. As it will be deemed useful in the sequel, we also introduce the following linear transformation of the parity-check matrix,
\begin{equation} \label{eqn:H_tilde}
	\wtilde{H} = \begin{bmatrix}
		\bf 0 \\
		B^{-1}H
	\end{bmatrix} = \begin{bmatrix} {\bf 0} & {\bf 0} \\
	B^{-1}A & I
	\end{bmatrix},
\end{equation}
where $\bf 0$ denotes the all-zero matrix of the appropriate dimension.
\end{remark}

The classical result of Shannon~\cite{Shannon1948} states that the maximum achievable rate for reliable communication over a discrete memoryless channel $p(y\cond x)$, known as the capacity $C$ of the channel, is given by the maximum mutual information, 
\begin{equation} \label{eqn:capacity_p2p}
    C = \underset{p(x)}{\max} \, I(X;Y).
\end{equation}
Shannon showed that this fundamental limit can be attained using a random coding scheme that generates the codebook i.i.d. according to the capacity-achieving distribution.

\begin{definition}[BMS Channel]
We say that a binary-input memoryless channel $p(y\cond x)$ is \emph{symmetric} (abbreviated, a BMS channel) if there exists a permutation $\pi:\cY\to\cY$ such that $\pi^{-1} = \pi$ and $p(y\cond x) = p(\pi(y)\cond x\oplus 1)$ for all $y\in \cY$ and $x\in\{0,1\}$. The channel $p(y\cond x)$ is \emph{asymmetric} if it is not symmetric.
\end{definition}

\begin{remark}
The capacity-achieving input distribution for a BMS channel is the uniform distribution~\cite[Theorem 4.5.2]{Gallager1968}. Note that this is the only input distribution that can be attained using linear code ensembles.
\end{remark}

\subsection{Symmetrized Channel} \label{sec:symmetrized}
Given any binary-input channel $p(y\cond x)$ (that is not necessarily symmetric) and any input distribution $p(x)$ (that is not necessarily uniform), a technique that will be crucial for us in our code constructions is the concept of a \emph{symmetrized channel} corresponding to the joint distribution $p(x,y) = p(x)p(y\cond x)$, defined as follows~\cite{Yang2009}.

\begin{definition}[Symmetrized Channel] \label{def:symmetrized}
	Given a binary-input channel $p(y \cond x)$ and an input distribution $p(x)$ (not necessarily uniform), the \emph{symmetrized channel} corresponding to $p(x,y) = p(x)p(y \cond x)$ is defined as the channel $\bar{p}$ with input alphabet $\cX = \{0,1\}$, output alphabet $\cY\times \{0,1\}$ and transition probabilities
	\begin{equation*}
		\bar{p}(y,v \cond x) = p_{X,Y}(x\oplus v, y).
	\end{equation*}
\end{definition}

\begin{remark} \label{remark:symmetrized}
	The following are immediate consequences of the definition of a symmetrized channel.
	\begin{enumerate}[label=(\roman*)]
		\item The channel $\bar{p}$ is symmetric under permutation $\pi\big((y,v)\big) = (y,v \oplus 1)$. In other words, for any $s \in \{0,1\}$,
		\[
		\bar{p}(y,v \cond x) = \bar{p}(y,v\oplus s \cond x \oplus s),
		\]
		which extends naturally when considering length-$n$ sequences.
		\item The channel $\bar{p}$ is, in particular, the conditional distribution $p_{\widebar{Y}\cond \widebar{X}}$ when $\widebar{X} = X\oplus V$, $\widebar{Y} = (Y,V)$, $(X,Y)$ is distributed according to $p(x,y)$, and $V\sim \BERN(1/2)$ is independent of $(X,Y)$.
		\item  Since $\bar{X} \sim \BERN(1/2)$ and the channel $\bar{p}$ is symmetric, it follows that the capacity of the symmetrized channel $\bar{p}$ is
        \begin{equation*}
            \begin{aligned}
                I(\widebar{X};\widebar{Y}) &= H(\widebar{X}) - H(\widebar{X}\cond \widebar{Y}) \\
                &= 1 - H(X\oplus V\cond Y,V) = 1-H(X\cond Y).
            \end{aligned}
        \end{equation*}
	\end{enumerate}
\end{remark}

Observations (i) and (ii) made above are illustrated in Figure~\ref{fig:symmetrized_channel}. If $(X^n,Y^n)$ are i.i.d. according to $p(x,y)$, then the symmetrized channel $\bar{p}$ is exactly the channel between $\widebar{X}^n = X^n\oplus V^n$ and $(Y^n,V^n)$, where $V^n$ is an i.i.d $\BERN(1/2)$ sequence that is independent of $(X^n,Y^n)$. Furthermore, the symmetric property of the channel $\bar{p}$ implies that, for any arbitrary sequence $S^n$, the channel between $\wtilde{X}^n = \widebar{X}^n \oplus S^n$ and $(Y^n,V^n\oplus S^n)$ is also described by $\bar{p}$. This property will turn out to be useful in several of our constructions in the coming chapters.

\begin{figure}[t]
	\centering
	\hspace*{2em}
	\def\svgscale{1.25}
	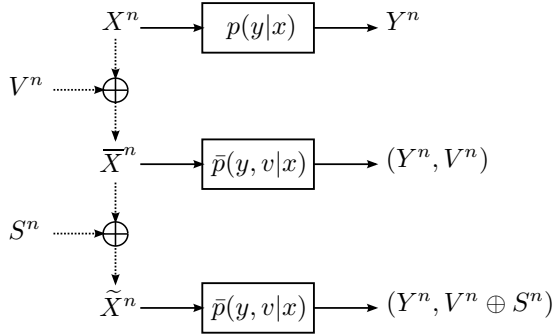
	\caption{The symmetrized channel.}
	\label{fig:symmetrized_channel}
\end{figure}

\subsection{Two Primitive Properties: Error Probability and Shaping Distance} \label{sec:properties}
Given a $(k,n)$ linear point-to-point channel code $(H,\phi)$ designed for a BMS channel $p(y \cond x)$, we will focus throughout this paper on two properties of the code, namely, the \emph{error probability} and the \emph{shaping distance}, which we define precisely below.
\begin{enumerate}[label=(\arabic*), leftmargin=1.75em]
    \item \textbf{Error probability:} Let $X^n \sim \Unif(\cC)$, where $\cC$ is the codebook corresponding to $H$, and let $Y^n$ be the output of the channel $p(y \cond x)$ when the input is $X^n$. The error probability $\eps$ of the code $(H,\phi)$ when used over the channel $p(y\cond x)$ is defined as the probability of block error of the decoding function, i.e., 
    \[
        \epsilon \triangleq \sum_{x^n \in \cC}\Big(2^{-k} \P\{\phi(Y^n)\neq X^n \,\cond \, X^n = x^n\}\Big).
    \]
    Shannon's point-to-point channel coding theorem~\cite{Shannon1948}, along with its achievability proof using linear codes~\cite{Elias1955}, state that a sequence of $(nR,n)$ linear codes having a vanishing error probability over the BMS channel $p(y \cond x)$ exists if and only if the rate $R < I\big(\BERN(1/2), p(y \cond x)\big)$, where $I\big(\Bern(1/2),p(y \cond x)\big)$ means the mutual information $I(X;Y)$ of the channel $p(y|x)$ when the input $X \sim \Bern(1/2)$. Since the channel is symmetric, $I\big(\Bern(1/2),p(y \cond x)\big)$ is, in fact, the capacity of the channel.
    \item \textbf{Shaping distance:} Let $Y^n$ be an i.i.d. $p(y)$ sequence, where $p(y) \triangleq \sum_{x}\frac{1}{2}p(y \cond x)$ is the marginal output distribution when the input distribution is $\BERN(1/2)$. Define $X^n = \phi(Y^n)$. The shaping distance $\delta$ of the code $(H,\phi)$ with respect to the channel $p(y \cond x)$ is defined as the total variation distance between the distribution of $(X^n,Y^n)$ (let's denote that by $q(x^n,y^n)$) and the i.i.d. $\frac{1}{2}p(y \cond x)$ distribution, i.e.,
    \[
        \delta \triangleq \frac{1}{2}\sum_{x^n,y^n}\left|q(x^n,y^n) - \frac{1}{2^n}\prod_{i=1}^{n}p(y_i \cond x_i)\right|.
    \]
    The shaping distance can be understood as a measure of the stochastic behavior of the decoding function $\phi$ in comparison to the ``backward'' channel $p(x\cond y)$ corresponding to the capacity-achieving input distribution. A sequence of $(nR,n)$ linear codes with a vanishing shaping distance exists if and only if the rate $R > I\big(\BERN(1/2), p(y \cond x)\big)$. To see this, one can refer to results on the distributed channel synthesis problem, introduced by Bennett \emph{et al.} in~\cite{Bennett2002} and further characterized by Cuff in~\cite{Cuff2013}. In this problem, an i.i.d. source $\widebar{Y}^n$ distributed according to $p(\bar{y})$ is encoded by an index $M \in [2^{nR_{\mathrm{s}}}]$ to a decoder that wishes to produce an output $\widebar{X}^n$ such that the joint distribution of $(\widebar{X}^n,\widebar{Y}^n)$ is indistinguishable (in total variation distance) from a given joint i.i.d distribution $p(\bar{x},\bar{y})$. This is referred to as ``synthesizing'' the channel $p(\bar{x} \cond \bar{y})$. The main result in the distributed channel synthesis literature~\cite{Cuff2013} is that there exists a construction of an encoder-decoder pair to synthesize $p(\bar{x} \cond \bar{y})$ if and only if the channel synthesis rate $R_{\mathrm{s}}$ is larger than $I(\widebar{X};\widebar{Y})$. We point out that, when $\widebar{X} \sim \BERN(1/2)$, the construction in~\cite{Cuff2013} can be generalized to a \emph{linear} construction\footnote{This holds because Theorem VII.1 in~\cite{Cuff2013} only uses the \emph{pairwise} independence of the codewords.}. Therefore, one can see that the condition of a small shaping distance on the point-to-point channel code $(H,\phi)$ is similar in nature to the condition imposed in the distributed channel synthesis problem, where the information bits corresponding to the output of $\phi$ can be seen as the index shared to the decoder in channel synthesis. Note that the encoder of the construction for the distributed channel synthesis problem would be the decoder of our point-to-point channel code, while the decoder of the channel synthesis construction is the encoder of the point-to-point code.
\end{enumerate}
In what follows, we will construct coding schemes for various problems in network information theory starting from one (or more) point-to-point channel codes. The properties of the point-to-point channel codes (either error probability or shaping distance) will be instrumental in translating the performance guarantees from one communication setting to another. For reference, Table~\ref{table:codes} shows a (possibly non-exhaustive) list of families of point-to-point channel codes (along with their decoding algorithms) that can achieve a vanishing error probability and a vanishing shaping distance asymptotically over any BMS channel.

\begin{table*}[t]
    \centering
    \begin{tabular}{ l|l } 
        \toprule
        \multicolumn{1}{c}{\textbf{Codes with a vanishing error probability}} & \multicolumn{1}{c}{\textbf{Codes with a vanishing shaping distance}} \\[.25\normalbaselineskip]
        \midrule
        \tabitem Random linear codes with joint typicality decoding~\cite{Elias1955} & \tabitem Random linear codes with likelihood decoding~\cite{Cuff2013} \\[.25\normalbaselineskip]
        \tabitem Polar codes with successive cancellation decoding~\cite{Arikan2009} & \tabitem Polar codes with successive cancellation decoding~\cite{Korada2010} \\[.25\normalbaselineskip]
        \tabitem Spatially-coupled codes with belief propagation decoding~\cite{Urbanke2013} & \\[.25\normalbaselineskip]
        \tabitem Reed--Muller codes with maximum-a-posteriori (MAP) decoding~\cite{Kudekar2017,Reeves2021,Abbe2023} & \\
        \bottomrule
    \end{tabular}
    \caption{Families of codes over BMS channels that have a vanishing error probability and a vanishing shaping distance asymptotically.}
	\label{table:codes}
\end{table*}

\section{Slepian--Wolf Coding} \label{sec:slepian}
In this section, we establish an ``equivalence'' between constructing a code for a point-to-point BMS channel and constructing a binary Slepian--Wolf code. In other words, we show that a code for any binary Slepian--Wolf problem can be designed starting from a code for a suitable BMS channel. Conversely, a code for any BMS channel can be constructed starting from a Slepian--Wolf code. In both cases, the optimal rate can be achieved asymptotically provided that the constituent code is rate-optimal. As a special case, the Slepian--Wolf coding scheme can be specialized to lossless source coding of a binary source. The duality between Slepian--Wolf coding and coding for a BMS channel has been previously noted in~\cite{Wyner1974,Ancheta1976,Yang2009,Lele2015}; in particular, the constructions provided in this section are equivalent to the ones in~\cite{Lele2015}. 

\subsection{Problem Statement} \label{sec:slepian-problem}
A binary Slepian--Wolf problem $p(x,y)$ consists of a source alphabet $\cX=\{0,1\}$, an arbitrary side information alphabet $\cY$, and a joint pmf $p(x,y)$ over $\cX \times \cY$~\cite{Slepian1973}. A discrete memoryless source $X$ with side information $Y$ generates a jointly i.i.d. random process $\{(X_i,Y_i)\}$ with $(X_i,Y_i)\sim p(x,y)$. An $(\ell,n)$ code $(g,\psi)$ for the Slepian--Wolf problem $p(x,y)$ consists of
\begin{itemize}
	\item an index set $\cI \subseteq \{0,1\}^n$ such that $|\cI| = 2^{\ell}$,
	\item an encoder $g:\cX^n \to \cI$ that maps each source sequence $x^n$ to an index $s^n = g(x^n)$, and
	\item a decoder $\psi:\cI\times\cY^n \to \cX^n$ that assigns a source estimate $\hat{x}^n=\psi(s^n,y^n)$ to each index $s^n$ and side information sequence $y^n$.
\end{itemize}
The rate of the code is $R=\ell/n$. The average probability of error of the code is $\epsilon = \P\{\widehat{X}^n \neq X^n\}$. A rate $R$, $0\leq R\leq 1$, is said to be \emph{achievable} for the binary Slepian--Wolf coding problem if there exists a sequence of $(nR,n)$ codes with vanishing error probability asymptotically. The classical result of Slepian and Wolf~\cite{Slepian1973} states that any rate $R > H(X \cond Y)$ is achievable for the binary Slepian--Wolf coding problem.

The Slepian--Wolf code is linear when the encoding function $g$ is linear, i.e., if for any $x^n, \tilde{x}^n \in \{0,1\}^n$, we have $g(x^n \oplus \tilde{x}^n) = g(x^n) \oplus g(\tilde{x}^n)$. When an $(\ell,n)$ linear Slepian--Wolf code $(g,\psi)$ has an encoding function that is defined as a matrix multiplication $g(x^n) = \begin{bmatrix}\bf 0 \\
Hx^n \end{bmatrix}$, where $H$ is an $\ell \times n$ matrix, we will refer to the code as an $(\ell,n)$ linear Slepian--Wolf code $(H,\psi)$.

\subsection{Code for P2P BMS Channel $\to$ Slepian--Wolf Code} \label{sec:p2p_slepian}
Consider a binary Slepian--Wolf problem $p(x,y)$, as defined in the previous section. We will construct a linear Slepian--Wolf code for this problem starting from a linear point-to-point channel code for a BMS channel. The BMS channel of interest is the symmetrized channel corresponding to $p(x,y)$, as defined in Section~\ref{sec:symmetrized}. The following lemma will be helpful to describe the Slepian--Wolf coding scheme.

\begin{lemma}\label{lemma:codify}
	Let $H$ be a parity-check matrix for a codebook $\cC$, and let $\wtilde{H}$ be as defined in~(\ref{eqn:H_tilde}). Define $\cS = \{s^n \in \{0,1\}^n: s^k = 0^k\}$. Then,
	\begin{enumerate}[label=(\roman*)]
		\item For any $x^n \in \{0,1\}^n$, there exists a unique $s^n \in \cS$ such that $x^n \oplus s^n \in \cC$. In particular, $s^n = \wtilde{H}x^n$.
		\item If $X^n$ is i.i.d. $\BERN(1/2)$, then $X^n\oplus \wtilde{H}X^n \sim \Unif(\cC)$.
		\item If $C^n \sim \Unif(\cC)$ and $S^n \sim \Unif(\cS)$ are independent, then $C^n \oplus S^n$ is i.i.d. $\BERN(1/2)$.
	\end{enumerate}
\end{lemma}

\begin{proof}
	See Appendix~\ref{appendix:codify}.
\end{proof}

\noindent Intuitively, part (ii) of Lemma~\ref{lemma:codify} gives a general way of generating a codeword uniformly at random starting from a uniformly distributed binary sequence, as illustrated in Fig.~\ref{fig:codify_diagram}. Conversely, part (iii) generates a uniformly distributed binary sequence starting from a codeword chosen uniformly at random.

\begin{figure}[t]
	\centering
	\def\svgscale{1.25}
\begingroup%
  \makeatletter%
  \providecommand\color[2][]{%
    \errmessage{(Inkscape) Color is used for the text in Inkscape, but the package 'color.sty' is not loaded}%
    \renewcommand\color[2][]{}%
  }%
  \providecommand\transparent[1]{%
    \errmessage{(Inkscape) Transparency is used (non-zero) for the text in Inkscape, but the package 'transparent.sty' is not loaded}%
    \renewcommand\transparent[1]{}%
  }%
  \providecommand\rotatebox[2]{#2}%
  \newcommand*\fsize{\dimexpr\f@size pt\relax}%
  \newcommand*\lineheight[1]{\fontsize{\fsize}{#1\fsize}\selectfont}%
  \ifx\svgwidth\undefined%
    \setlength{\unitlength}{247.5bp}%
    \ifx\svgscale\undefined%
      \relax%
    \else%
      \setlength{\unitlength}{\unitlength * \real{\svgscale}}%
    \fi%
  \else%
    \setlength{\unitlength}{\svgwidth}%
  \fi%
  \global\let\svgwidth\undefined%
  \global\let\svgscale\undefined%
  \makeatother%
  \begin{picture}(1,0.22811286)%
    \lineheight{1}%
    \setlength\tabcolsep{0pt}%
    \put(0,0){\includegraphics[width=\unitlength]{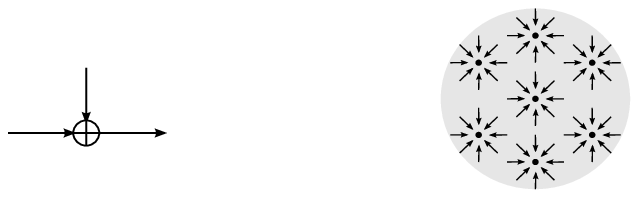}}%
    \put(0.00726232,0.0719912){\color[rgb]{0,0,0}\makebox(0,0)[lt]{\lineheight{1.25}\smash{\begin{tabular}[t]{l}$X^n$\end{tabular}}}}%
    \put(0.25271704,0.07199414){\color[rgb]{0,0,0}\makebox(0,0)[lt]{\lineheight{1.25}\smash{\begin{tabular}[t]{l}$C^n\sim\mathrm{Unif}(\mathcal{C})$\end{tabular}}}}%
    \put(0.11841453,0.16679961){\color[rgb]{0,0,0}\makebox(0,0)[lt]{\lineheight{1.25}\smash{\begin{tabular}[t]{l}$\widetilde{H}X^n$\end{tabular}}}}%
    \put(0.70487643,0.25535004){\color[rgb]{0,0,0}\makebox(0,0)[lt]{\begin{minipage}{0.6129125\unitlength}\raggedright  \end{minipage}}}%
  \end{picture}%
\endgroup%

	\caption{Illustration of a shift to $X^n \IID \Bern(1/2)$ by $\wtilde{H}X^n$ in $\{0,1\}^n$ space.}
	\label{fig:codify_diagram}
\end{figure}

The main idea of the construction of the Slepian--Wolf coding scheme is captured in the following lemma.

\begin{lemma} \label{lemma:p2p_slepian}
	Let $(X^n,Y^n)$ be i.i.d. according to $p(x,y)$, and $V^n$ be i.i.d. $\mathrm{Bern}(1/2)$ and independent of $(X^n,Y^n)$. Let $\bar{p}$ be the symmetrized channel corresponding to $p(x,y)$, as defined in Definition~\ref{def:symmetrized}. Consider a parity-check matrix $H$ for a codebook $\cC$, and let $\wtilde{H}$ be as defined in $(\ref{eqn:H_tilde})$. Consider the sequences
	\begin{equation} \label{eqn:codify}
		\begin{aligned}
			C^n &= X^n \oplus V^n \oplus \widetilde{H}X^n \oplus \wtilde{H} V^n, \\
			U^n &= V^n \oplus \widetilde{H}V^n \oplus \widetilde{H}X^n.
		\end{aligned}
	\end{equation}
	Then,
	\begin{equation*}
		\P\{C^n = c^n, U^n= u^n, Y^n = y^n\} = \frac{1}{2^k}\prod_{i=1}^{n}\bar{p}(y_i,u_i\cond c_i)
	\end{equation*}
	for every $c^n \in \cC$, $u^n \in \{0,1\}^n$ and $y^n \in \cY^n$.
\end{lemma}

\begin{proof}
	Lemma~\ref{lemma:p2p_slepian} can be seen as a recast of Lemmas 2 and 3 in~\cite{Lele2015}. For completion, the proof is provided in Appendix~\ref{appendix:p2p_slepian}.
\end{proof}

\noindent Lemma~\ref{lemma:p2p_slepian} says that if $Y^n$ is the output of the channel $p(y\cond x)$ when the channel input is $X^n$, then, for a uniformly distributed binary sequence $V^n$, the sequences $(Y^n,U^n)$ are distributed as the outputs of the channel $\bar{p}$ when the channel input is $C^n$, a uniformly distributed codeword in the codebook $\cC$, where $C^n$ and $U^n$ are as defined in~(\ref{eqn:codify}). The relations between the different random variables is illustrated in Fig.~\ref{fig:slepian_codify}.

\begin{figure}[t]
	\centering
	\def\svgscale{1.3}
	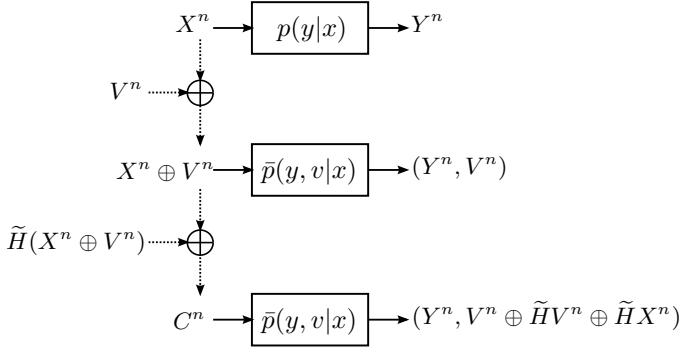
	\caption{The relations between the random variables $(X^n,Y^n,C^n,U^n)$ defined in Lemma~\ref{lemma:p2p_slepian}. Notice the similarity to Fig.~\ref{fig:symmetrized_channel} when $S^n$ in Fig.~\ref{fig:symmetrized_channel} is set to $\wtilde{H}(X^n\oplus V^n)$. To recover $X^n$ from $Y^n$, one can go through the path $(Y^n,U^n) \to C^n \to R^n \to X^n$. To get $C^n$ from $(Y^n,U^n)$, one can apply a decoder of a point-to-point channel code designed for the channel $\bar{p}$. This explains the Slepian--Wolf coding scheme shown in Fig.~\ref{fig:p2p_slepian}.}
	\label{fig:slepian_codify}
\end{figure}

Now, we are ready to construct a coding scheme for the Slepian--Wolf problem $p(x,y)$. The coding scheme uses the following point-to-point channel code.
\begin{lego}[\textbf{P2P} $\to$ \textbf{SW}]
	a $(k,n)$ linear point-to-point channel code $(H,\phi)$ with codebook $\cC$ for the symmetrized channel $\bar{p}$ corresponding to $p(x,y)$, which is defined over an input alphabet $\cX = \{0,1\}$ and output alphabet $\cY\times\{0,1\}$ by
	\begin{equation} \label{eqn:channelq}
		\bar{p}(y,v \cond x) = p_{X,Y}(x\oplus v,y).
	\end{equation}
	Let $\eps$ be the average probability of error of the code $(H,\phi)$ when used over the channel $\bar{p}$. 
\end{lego}

Fig.~\ref{fig:p2p_code} shows the channel code $(H,\phi)$ when used over the channel $\bar{p}$. The average probability of error $\eps$ of the code can be expressed as
\[
\eps = \P\{\phi(\widetilde{U}^n,\widetilde{Y}^n) \neq \widetilde{C}^n\}.
\]

\begin{figure}[t]
	\centering
	\hspace*{1em}
	\def\svgscale{1.25}
\begingroup%
  \makeatletter%
  \providecommand\color[2][]{%
    \errmessage{(Inkscape) Color is used for the text in Inkscape, but the package 'color.sty' is not loaded}%
    \renewcommand\color[2][]{}%
  }%
  \providecommand\transparent[1]{%
    \errmessage{(Inkscape) Transparency is used (non-zero) for the text in Inkscape, but the package 'transparent.sty' is not loaded}%
    \renewcommand\transparent[1]{}%
  }%
  \providecommand\rotatebox[2]{#2}%
  \newcommand*\fsize{\dimexpr\f@size pt\relax}%
  \newcommand*\lineheight[1]{\fontsize{\fsize}{#1\fsize}\selectfont}%
  \ifx\svgwidth\undefined%
    \setlength{\unitlength}{247.50002163bp}%
    \ifx\svgscale\undefined%
      \relax%
    \else%
      \setlength{\unitlength}{\unitlength * \real{\svgscale}}%
    \fi%
  \else%
    \setlength{\unitlength}{\svgwidth}%
  \fi%
  \global\let\svgwidth\undefined%
  \global\let\svgscale\undefined%
  \makeatother%
  \begin{picture}(1,0.15151515)%
    \lineheight{1}%
    \setlength\tabcolsep{0pt}%
    \put(0,0){\includegraphics[width=\unitlength]{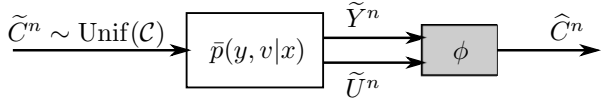}}%
    \put(0.0083719,0.08946315){\color[rgb]{0,0,0}\makebox(0,0)[lt]{\lineheight{1.25}\smash{\begin{tabular}[t]{l}$\widetilde{C}^n\sim \mathrm{Unif}(\mathcal{C})$\end{tabular}}}}%
    \put(0.66891111,0.09024457){\color[rgb]{0,0,0}\makebox(0,0)[lt]{\lineheight{1.25}\smash{\begin{tabular}[t]{l}$\widehat{C}^n$\end{tabular}}}}%
    \put(0.25565717,0.06841918){\color[rgb]{0,0,0}\makebox(0,0)[lt]{\lineheight{1.25}\smash{\begin{tabular}[t]{l}$\bar{p}(y,v|x)$\end{tabular}}}}%
    \put(0.55104806,0.06741878){\color[rgb]{0,0,0}\makebox(0,0)[lt]{\lineheight{1.25}\smash{\begin{tabular}[t]{l}$\phi$\end{tabular}}}}%
    \put(0.4206288,0.10563639){\color[rgb]{0,0,0}\makebox(0,0)[lt]{\lineheight{1.25}\smash{\begin{tabular}[t]{l}$\widetilde{Y}^n$\end{tabular}}}}%
    \put(0.42017445,0.02033093){\color[rgb]{0,0,0}\makebox(0,0)[lt]{\lineheight{1.25}\smash{\begin{tabular}[t]{l}$\widetilde{U}^n$\end{tabular}}}}%
  \end{picture}%
\endgroup%

	\caption{A code for the symmetric channel $\bar{p}$, defined in (\ref{eqn:channelq}).}
	\label{fig:p2p_code}
\end{figure}

Fig.~\ref{fig:p2p_slepian} illustrates the block diagram of the Slepian--Wolf coding scheme that uses the point-to-point channel code $(H,\phi)$. The coding scheme can be described as follows.

\begin{figure}[t]
	\centering
	\def\svgscale{1}
	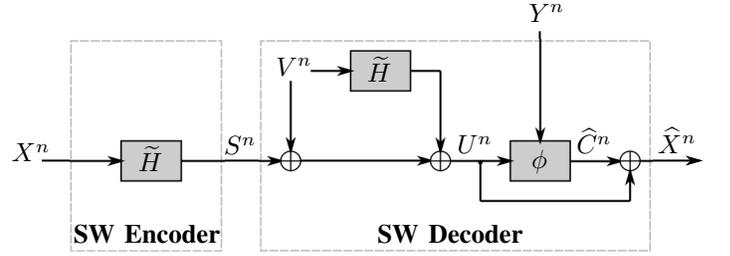
	\caption{A Slepian--Wolf code starting from a point-to-point channel code.}
	\label{fig:p2p_slepian}
\end{figure}

\vspace{0.25em}
\noindent \emph{Encoding:} Upon observing the source sequence $x^n$, the sender transmits $s^n = \widetilde{H}x^n$, where $\wtilde{H}$ is as defined in~(\ref{eqn:H_tilde}).

\vspace{0.25em}
\noindent \emph{Decoding:} Upon observing the side information sequence $y^n$ and receiving the index $s^n$, the decoder declares $\hat{x}^n = \phi(s^n \oplus v^n \oplus \widetilde{H}v^n, y^n) \oplus s^n \oplus v^n\oplus \widetilde{H}v^n$ as the source estimate, where $v^n$ is a realization of a random dither generated independently at the decoder. Notice the similarity of this decoding method with the observations made through Fig.~\ref{fig:slepian_codify}.

\vspace{0.25em}
\noindent \emph{Analysis of probability of error:} We have
\begin{align}
    &\P\{\widehat{X}^n \neq X^n\} \nonumber\\
    &= \P\{\phi(V^n \hspace*{-0.2em} \oplus \hspace*{-0.2em} \widetilde{H}V^n \hspace*{-0.2em} \oplus \hspace*{-0.2em} \widetilde{H}X^n, Y^n) \hspace*{-0.2em} \oplus \hspace*{-0.2em}  \widetilde{H}X^n \hspace*{-0.2em} \oplus \hspace*{-0.2em} \widetilde{H}V^n \hspace*{-0.2em} \oplus \hspace*{-0.2em} V^n \neq X^n \} \nonumber \\
    &= \P\{ \phi(V^n\hspace*{-0.2em} \oplus\hspace*{-0.2em} \widetilde{H}V^n \hspace*{-0.2em} \oplus \hspace*{-0.2em} \widetilde{H}X^n, Y^n)  \neq X^n \hspace*{-0.2em} \oplus \hspace*{-0.2em}  V^n \hspace*{-0.2em} \oplus \hspace*{-0.2em} \widetilde{H}(X^n \hspace*{-0.2em} \oplus\hspace*{-0.2em}  V^n)\} \nonumber \\
    &=\P\{\phi(U^n,Y^n) \neq C^n\}\nonumber \\
    &\overset{(a)}{=}\P\{\phi(\wtilde{U}^n,\wtilde{Y}^n) \neq \wtilde{C}^n\}= \epsilon,
\end{align}
where $(a)$ follows from Lemma~\ref{lemma:p2p_slepian}. Note that since the probability of error averaged over $V^n$ is $\eps$, there exists a deterministic $v^n$ sequence such that the error probability is bounded by $\eps$. 

\vspace{0.25em}
\noindent \emph{Rate:} By construction, the rate of the Slepian--Wolf code is $(n-k)/n$. 

\begin{remark} \label{remark:p2p_slepian_rates}
	Recall the definition of $(\widebar{X},\widebar{Y})$ in Remark~\ref{remark:symmetrized}. A sequence of codes for the channel $\bar{p}$ with a vanishing error probability exists if and only if the rate is smaller than
	\begin{equation*}
        \begin{aligned}
		  I(\widebar{X};\widebar{Y}) &= H(\widebar{X}) - H(\widebar{X}\cond \widebar{Y}) \\
            &= 1 - H(X\oplus U \cond Y,U) = 1-H(X\cond Y).
        \end{aligned}
	\end{equation*}
	It follows that, if the rate of the code for the channel $\bar{p}$ is $\frac{k}{n} = I(\widebar{X};\widebar{Y}) - \gamma$ for some $\gamma > 0$, then the rate of the Slepian--Wolf code is $\frac{n-k}{n} = H(X \cond Y) + \gamma$.
\end{remark}

\vspace{0.25em}
\noindent {\bf Conclusion:} From each linear $(k,n)$ code for the BMS channel $\bar{p}$ defined in~(\ref{eqn:channelq}) with average probability of error $\eps$, one can construct a linear $(n-k,n)$ code for the Slepian--Wolf problem $p(x,y)$ with average probability of error $\eps$.

\subsection{Slepian--Wolf Code $\to$ Code for P2P BMS Channel}  \label{sec:slepian_p2p}
Now, we consider a BMS channel $p(y \cond x)$. We show that a code for this channel can be constructed starting from the following Slepian--Wolf code.

\begin{lego}[\textbf{SW} $\to$ \textbf{P2P}]
an $(n-k,n)$ linear Slepian--Wolf code $(H,\psi)$ for the problem $p(x,y) = \frac{1}{2}p(y \cond x)$ with an average probability of error $\eps$.
\end{lego}

\begin{figure}[t]
	\centering
	\hspace*{0.5em}
	\def\svgscale{1}
\begingroup%
  \makeatletter%
  \providecommand\color[2][]{%
    \errmessage{(Inkscape) Color is used for the text in Inkscape, but the package 'color.sty' is not loaded}%
    \renewcommand\color[2][]{}%
  }%
  \providecommand\transparent[1]{%
    \errmessage{(Inkscape) Transparency is used (non-zero) for the text in Inkscape, but the package 'transparent.sty' is not loaded}%
    \renewcommand\transparent[1]{}%
  }%
  \providecommand\rotatebox[2]{#2}%
  \newcommand*\fsize{\dimexpr\f@size pt\relax}%
  \newcommand*\lineheight[1]{\fontsize{\fsize}{#1\fsize}\selectfont}%
  \ifx\svgwidth\undefined%
    \setlength{\unitlength}{397.50001442bp}%
    \ifx\svgscale\undefined%
      \relax%
    \else%
      \setlength{\unitlength}{\unitlength * \real{\svgscale}}%
    \fi%
  \else%
    \setlength{\unitlength}{\svgwidth}%
  \fi%
  \global\let\svgwidth\undefined%
  \global\let\svgscale\undefined%
  \makeatother%
  \begin{picture}(1,0.14150954)%
    \lineheight{1}%
    \setlength\tabcolsep{0pt}%
    \put(0,0){\includegraphics[width=\unitlength]{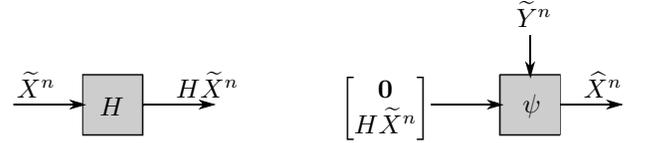}}%
    \put(0.55049539,0.04626952){\color[rgb]{0,0,0}\makebox(0,0)[lt]{\lineheight{1.25}\smash{\begin{tabular}[t]{l}$\widehat{X}^n$\end{tabular}}}}%
    \put(0.48647866,0.11365501){\color[rgb]{0,0,0}\makebox(0,0)[lt]{\lineheight{1.25}\smash{\begin{tabular}[t]{l}$\widetilde{Y}^n$\end{tabular}}}}%
    \put(0.31947326,0.0297477){\color[rgb]{0,0,0}\makebox(0,0)[lt]{\lineheight{1.25}\smash{\begin{tabular}[t]{l}$\begin{bmatrix}\bf 0 \\ H\widetilde{X}^n\end{bmatrix}$\end{tabular}}}}%
    \put(0.01275983,0.04626946){\color[rgb]{0,0,0}\makebox(0,0)[lt]{\lineheight{1.25}\smash{\begin{tabular}[t]{l}$\widetilde{X}^n$\end{tabular}}}}%
    \put(0.16370303,0.04626946){\color[rgb]{0,0,0}\makebox(0,0)[lt]{\lineheight{1.25}\smash{\begin{tabular}[t]{l}$H\widetilde{X}^n$\end{tabular}}}}%
    \put(0.09171351,0.02939306){\color[rgb]{0,0,0}\makebox(0,0)[lt]{\lineheight{1.25}\smash{\begin{tabular}[t]{l}$H$\end{tabular}}}}%
    \put(0.49262439,0.03254378){\color[rgb]{0,0,0}\makebox(0,0)[lt]{\lineheight{1.25}\smash{\begin{tabular}[t]{l}$\psi$\end{tabular}}}}%
  \end{picture}%
\endgroup%

	\caption{A linear Slepian--Wolf code, where $(\wtilde{X}^n, \wtilde{Y}^n)$ are i.i.d. according to $p(x,y)$.}
	\label{fig:slepian_code}
\end{figure}

Fig.~\ref{fig:slepian_code} shows the Slepian--Wolf code $(H,\psi)$, where $(\wtilde{X}^n, \wtilde{Y}^n)$ are i.i.d. sequences distributed according to $p(x,y)$. The average probability of error $\eps$ of the Slepian--Wolf code can be written as
\[
\eps = \P\left\{\psi\left(\begin{bmatrix}
    {\bf 0} \\
    H\wtilde{X}^n
\end{bmatrix}, \wtilde{Y}^n\right) \neq \wtilde{X}^n\right\}.
\]

To construct a code for the channel $p(y \cond x)$, let $V^n$ be an i.i.d. $\Bern(1/2)$ random dither shared between the encoder and the decoder, and let $C^n \in \cC$ represent the message to be transmitted, where $\cC$ is the codebook corresponding to $H$. Fig.~\ref{fig:slepian_p2p} illustrates the block diagram of the point-to-point channel code. The coding scheme can be summarized as follows:

\begin{figure}[t]
	\centering
	\def\svgscale{1}
	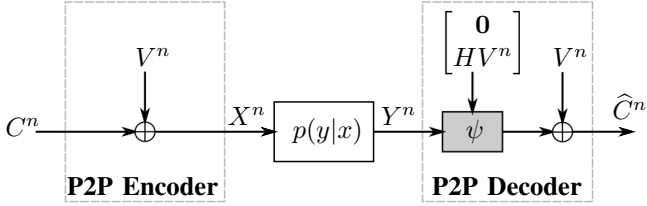
	\caption{A point-to-point channel code starting from a Slepian--Wolf code.}
	\label{fig:slepian_p2p}
\end{figure}

\vspace{0.25em}
\noindent \emph{Encoding:} To send the message $c^n \in \cC$, the sender transmits $x^n = c^n \oplus v^n$, where $v^n$ is a realization of a random dither shared between the encoder and the decoder.

\vspace{0.25em}
\noindent \emph{Decoding:} Upon observing $y^n$, the decoder declares $\hat{c}^n = \psi\left(\begin{bmatrix}
    {\bf 0} \\
    Hv^n
\end{bmatrix}, y^n\right) \oplus v^n$ as the message estimate.

\vspace{0.25em}
\noindent \emph{Analysis of probability of error:} We have
\begin{equation*}
	\begin{aligned}
		\P\{\widehat{C}^n \neq C^n\} &= \P\left\{\psi\left(\begin{bmatrix}
		    {\bf 0} \\
            HV^n
		\end{bmatrix},Y^n\right)\oplus V^n \neq C^n\right\} \\
		&\overset{(a)}{=} \P\left\{\psi\left(\begin{bmatrix}
		    {\bf 0} \\
            HX^n
		\end{bmatrix},Y^n\right)\neq X^n\right\} \\
		&\overset{(b)}{=} \P\left\{\psi\left(\begin{bmatrix}
		    {\bf 0} \\
            H\wtilde{X}^n
		\end{bmatrix},\wtilde{Y}^n\right)\neq \wtilde{X}^n\right\} \\
		&= \eps,
	\end{aligned}
\end{equation*}
where $(a)$ follows since $HX^n = HC^n \oplus HV^n = HV^n$, and $(b)$ follows since after dithering with the uniform $V^n$, $(X^n,Y^n)$ are identically distributed as $(\wtilde{X}^n, \wtilde{Y}^n)$ in the Slepian--Wolf problem. Note that since the probability of error averaged over $V^n$ is $\eps$, there exists a deterministic $v^n$ sequence such that the probability of error is bounded by $\eps$. 

\vspace{0.25em}
\noindent \emph{Rate:} Since $\abs{\cC}  = 2^k$, the rate of the point-to-point channel code is $k/n$.

\begin{remark}
If the rate of the Slepian--Wolf code is $\frac{n-k}{n} = H(X \cond Y) + \gamma$ for some $\gamma > 0$, then the rate of the point-to-point channel code is 
\begin{equation*}
    1-H(X \cond Y) - \gamma = I\big(\Bern(1/2),p(y \cond x)\big) - \gamma,
\end{equation*}
where $I\big(\Bern(1/2),p(y \cond x)\big)$ is the maximum achievable rate for the BMS channel $p(y \cond x)$.
\end{remark}

\vspace{0.25em}
\noindent {\bf Conclusion:} From each linear $(n-k,n)$ Slepian--Wolf code for the problem $p(x,y) = \frac{1}{2}p(y \cond x)$ with average probability of error $\eps$, one can construct a linear $(k,n)$ code for the BMS channel $p(y \cond x)$ with average probability of error $\eps$.

\subsection{Specialization to Lossless Source Coding} \label{sec:lossless}
As a special case of Slepian--Wolf coding, a lossless source code can be implemented using a point-to-point channel code that is designed for a binary symmetric channel. To see this, consider a binary memoryless source that generates an i.i.d. $\BERN(\theta)$ sequence $X^n$ for some $\theta \in (0,1/2)$. As in Slepian--Wolf coding, the goal is to represent the source sequence using as few bits as possible to a decoder that wishes to find an estimate $\widehat{X}^n$ of the sequence. The definition of a lossless source code, its rate and probability of error follow similarly as in Slepian--Wolf coding, with the exception that a lossless source decoder has no access to any side information sequence.

The coding scheme presented in Section~\ref{sec:p2p_slepian} can be specialized to lossless source coding when there is no side information. It is easy to check that, in the case of no side information, the symmetrized channel corresponding to $p(x)$ is a $\BSC(p_X(1))$. Therefore, a lossless source coding scheme for a $\BERN(\theta)$ source can be constructed starting from a point-to-point channel code designed for $\BSC(\theta)$, as given in the following Lego brick.
\begin{lego}[\textbf{P2P} $\to$ \textbf{Lossless}]
a $(k,n)$ linear point-to-point channel code $(H,\phi)$ designed for $\BSC(\theta)$ with average probability of error $\eps$ when used over the channel. 
\end{lego}

\begin{figure}[t]
	\centering
	\def\svgscale{1}
	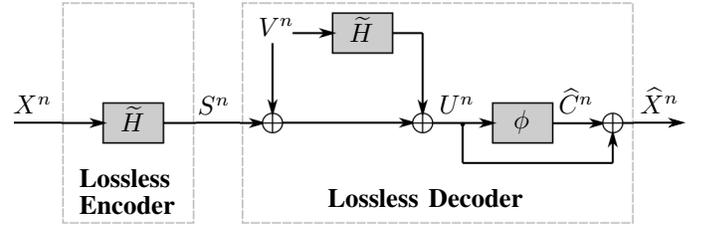
	\caption{A lossless source code starting from a point-to-point channel code.}
	\label{fig:lossless}
\end{figure}

Fig.~\ref{fig:lossless} shows the lossless source coding scheme, where $X^n$ is an i.i.d. $\BERN(\theta)$ source sequence, and $V^n$ is an i.i.d. $\BERN(1/2)$ sequence generated at the decoder independently of $X^n$. Notice that the coding scheme is very similar to the Slepian--Wolf coding scheme (Fig.~\ref{fig:p2p_slepian}), except that no side information sequence is used at the decoder side. As before, the rate of the lossless source code is $(n-k)/n$, and its average probability of error is given by $\eps$, which follows by specializing Lemma~\ref{lemma:p2p_slepian} to the case when $p(y)$ is deterministic and independent of $p(x)$.

\begin{remark}
If the rate of the channel code is $\frac{k}{n} = 1-H(\theta) - \gamma$ for some $\gamma > 0$, then the rate of the lossless source code is $H(\theta) + \gamma$.
\end{remark}

\begin{remark}
Conversely, a lossless source code can be used to construct a point-to-point channel code for a binary symmetric channel (BSC). This implies an ``equivalence'' between constructing a linear point-to-point channel code for the BSC and constructing a linear lossless source code for a binary source, which is well-understood in the literature~\cite{Weiss1962}.
\end{remark}

\section{Lossy Source Coding}  \label{sec:lossy}
In this section, we construct coding schemes for the lossy source coding problem starting from simple Lego bricks. We consider two cases: the first is the case of a symmetric source, and the second corresponds to a general asymmetric source. The distinction is made because a simpler construction is possible for the former case. More specifically, in the case of a symmetric source, our coding scheme is constructed starting from a single point-to-point symmetric channel code, whereas the coding scheme for the general asymmetric source uses both a point-to-point symmetric channel code and a lossless source code. In both cases, the proposed coding scheme is rate-optimal provided that the constituent Lego bricks are rate-optimal. First, let's start by defining the lossy source coding problem.

\subsection{Problem Statement} \label{sec:lossy_problem}
Introduced by Shannon in~\cite{Shannon1959}, the problem of lossy compression of a binary memoryless source consists of an i.i.d. random process $\{X_i\}$ with $X_i \sim \mathrm{Bern}(\theta)$ for some $\theta \in (0,1/2]$. The goal is to efficiently represent a source sequence $X^n$ when some distortion is allowed during reconstruction. More formally, an $(R,n)$ code for the lossy source coding problem consists of an encoder $g: \{0,1\}^n \to [2^{nR}]$ that assigns an index $M$ to the source sequence $X^n$, and a decoder $\psi: [2^{nR}] \to \{0,1\}^n$ that assigns an estimate $\hat{X}^n$ to the index $M$.
The rate of the code is $R$, and its expected distortion is $\frac{1}{n}\E\left[ d_H(X^n,\widehat{X}^n) \right]$,
where $d_H(.,.)$ denotes the Hamming distance metric. A rate-distortion pair $(R,D)$ is said to be \emph{achievable} if there exists a sequence of $(R,n)$ codes with
\[
\underset{n \to \infty}{\limsup} \frac{1}{n}\E\left[ d_H(X^n,\widehat{X}^n) \right] \leq D.
\]
The \emph{rate-distortion function} $R(D)$ is defined as the infimum of all rates $R$ such that $(R,D)$ is achievable.

Shannon~\cite{Shannon1959} showed that the rate-distortion function for a $\BERN(\theta)$ source can be expressed as
\[
R(D) = \begin{cases}
H(\theta) - H(D) &\quad \text{for } 0\leq D < \theta, \\
0 &\quad \text{for } D \geq \theta.
\end{cases}
\]
Shannon's random coding scheme assigns, for each typical source sequence $x^n$, a reconstruction sequence $\hat{x}^n$ that is jointly typical with $x^n$ for some desired conditional pmf $p(\hat{x}\cond x)$. For the case of a binary source, the desired conditional pmf $p(\hat{x}\cond x)$ corresponds to the case when the ``backward'' channel $p(x \cond \hat{x})$ is a $\BSC(D)$.

\subsection{Symmetric Source}
Consider a realization of a symmetric source $X^n \IID \Bern(1/2)$, and let $D\in (0,1/2)$ be some desired distortion level. We will construct a lossy source coding scheme for this source starting from the following point-to-point channel code.

\begin{lego}[\textbf{P2P} $\to$ \textbf{Sym. Lossy}]
a $(k,n)$ linear point-to-point channel code $(H,\phi)$ for $\mathrm{BSC}(D)$ with a shaping distance $\delta$.
\end{lego}

The main ingredient in the lossy source coding scheme is utilizing the shaping capability of the decoding function $\phi$ -- manifested by its shaping distance property $\delta$ -- in order to generate a sequence according to the desired distribution (or ``close'' to it). In the simple setting of a symmetric source, this can be done by simply declaring the output of the decoding function as the source reconstruction, and the corresponding information bits as the index shared to the decoder, as depicted in Fig.~\ref{fig:lossy_sym}, where $G$ is the generator matrix of the code $(H,\phi)$, and $\mathrm{Info(.)}$ is the function that takes as input a codeword in a linear code  and outputs the corresponding information bits. The coding scheme can be summarized as follows.

\vspace{0.25em}
\noindent \emph{Encoding:} Upon observing the source sequence $x^n$, the encoder stores the information sequence $u^k$ such that $\phi(x^n) = u^kG$.

\vspace{0.25em}
\noindent \emph{Decoding:} Upon observing the index $u^k$, the decoder declares the sequence $\hat{x}^n = u^kG$ as the source estimate.

\vspace{0.25em}
\noindent \emph{Analysis of the average distortion:} Let $q(x^n,\hat{x}^n)$ denote the distribution of $(X^n,\widehat{X}^n)$, and let $p(x^n,\hat{x}^n)$ be the desired i.i.d. distribution, i.e.,
\[
p(x^n,\hat{x}^n) = \frac{1}{2^n} D^{\weight(x^n \oplus \hat{x}^n)}(1-D)^{n-\weight(x^n \oplus \hat{x}^n)}
\]
Since $\widehat{X}^n = \phi(X^n)$ and $X^n \IID \Bern(1/2)$ (which is the channel output distribution of $\BSC(D)$ under the capacity-achieving input distribution), the average distortion is given by
\begin{align*}
    &\frac{1}{n}\E[d_H(X^n,\widehat{X}^n)] = \frac{1}{n}\sum_{x^n,\hat{x}^n}q(x^n,\hat{x}^n)d_H(x^n,\hat{x}^n)\\
    &= \frac{1}{n}\sum_{x^n,\hat{x}^n}p(x^n,\hat{x}^n)d_H(x^n,\hat{x}^n) \\
    &\quad + \frac{1}{n}\sum_{x^n,\hat{x}^n}\left(q(x^n,\hat{x}^n) - p(x^n,\hat{x}^n)\right) d_H(x^n,\hat{x}^n)\\
    &\overset{(a)}{\leq} D + \frac{1}{2}\sum_{x^n,\hat{x}^n}\left|q(x^n,\hat{x}^n) - p(x^n,\hat{x}^n)\right|\\
    &\overset{(b)}{=} D+\delta,
\end{align*}
where $(a)$ holds since $p(\hat{x}^n\cond x^n)$ is equivalent to $n$ independent uses of $\BSC(D)$ and the fact that $\sum_{i}c_i(a_i -b_i) \leq \frac{1}{2}\sum_{i}\left|a_i-b_i\right|$ whenever $0\leq c_i \leq 1$ and $\sum_{i}a_i = \sum_{i}b_i$, and $(b)$ follows by the definition of the shaping distance of the code $(H,\phi)$.

\vspace{0.25em}
\noindent \emph{Rate:} The rate of the coding scheme is $R = \frac{k}{n}$.

\vspace{0.25em}
\noindent Thus, we have obtained an $(R,n)$ lossy source code for a symmetric source with an expected distortion that is bounded by $D+\delta$ starting from a point-to-point channel code for $\BSC(D)$ with shaping distance $\delta$. 

\begin{figure}[t]
	\centering
	\hspace*{1em}
	\def\svgscale{1.25}
	\input{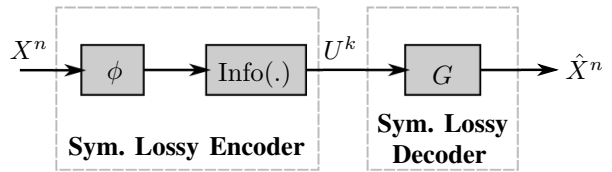}
	\caption{A lossy source coding scheme for a symmetric source starting from a point-to-point channel code.}
	\label{fig:lossy_sym}
\end{figure}

\begin{remark} \label{remark:lossy_sym}
Following the discussion in Section~\ref{sec:properties}, a sequence of linear point-to-point channel codes with a vanishing shaping distance $\delta$ exists if (and only if\footnote{The ``only if'' part clearly holds by our current construction and the lossy source coding theorem.}) the rate is larger than $1-H(D)$. 
\end{remark}

\begin{remark}
    Note that for this construction, it suffices to have a point-to-point channel code that satisfies 
    \[
        \abs{\frac{1}{n}\E[d_H(X^n,\phi(X^n))]-D} \leq \delta,
    \]
    rather than the more stringent condition of the shaping distance. Nonetheless, we show the shaping distance condition here for illustrative purposes, as it introduces the idea of using a decoding function for shaping a binary sequence, a theme that will be recurrent throughout the paper.
\end{remark}

\subsection{Asymmetric Source} \label{sec:lossy_asym}
Now, we consider the case of a general binary memoryless source that generates an i.i.d. $\BERN(\theta)$ sequence $X^n$ for some $\theta \in (0,1/2)$. Let $D \in (0,\theta)$ be some desired distortion level\footnote{Clearly, when $D \geq \theta$, a rate-zero coding scheme is possible by deterministically outputting the all-zero sequence as the source reconstruction.}, and define
\[
\alpha \triangleq \frac{\theta - D}{1-2D}.
\]
Note that $\widehat{X} \sim \Bern(\alpha)$ when the conditional distribution $p(x \cond \hat{x})$ is $\mathrm{BSC}(D)$, which is the desired conditional distribution of the source given the reconstruction as inspired by Shannon's random coding scheme~\cite{Shannon1959}. Let $p(x,\hat{x})$ denote the desired joint distribution between the source and the reconstruction, i.e.,
\begin{equation} \label{eqn:lossy_joint_distribution}
	p(x,\hat{x}) = \alpha^{\hat{x}}(1-\alpha)^{1-\hat{x}}D^{x \oplus \hat{x}}(1-D)^{1-x\oplus \hat{x}}.
\end{equation}
The proposed lossy source coding scheme in this general setting utilizes a point-to-point channel code  and a lossless source code. At the encoder side, the point-to-point channel code is used to generate a sequence according to the desired distribution of the reconstruction (i.e., the i.i.d. $\BERN(\alpha)$ distribution), and the lossless source code is used to compress that sequence to the decoder. Another key ingredient in the coding scheme is the assumption that the two codebooks are \emph{nested}. 

Before we describe the coding scheme, we state the following lemma, which will be useful in several constructions in this paper.
\begin{lemma} \label{lemma:shaping_distance}
	Let $\bar{p}(\td{y}, \td{v} \cond \td{x})$ be the symmetrized channel corresponding to a given joint distribution $p(\td{x},\td{y})$. Let $(H,\phi)$ be a point-to-point channel code designed for $\bar{p}$, and let $\delta$ be its shaping distance. Let $\td{Y}^n$ be i.i.d. according to $p(\td{y})$ and $\td{V}^n$ be i.i.d. $\BERN(1/2)$ such that $\td{Y}^n$ and $\td{V}^n$ are independent, and let $\td{U}^n = \phi(\td{Y}^n,\td{V}^n)\oplus \td{V}^n$. Then,
	\[
	\frac{1}{2}\sum_{\td{u}^n,\td{y}^n}\left| \P\{\td{U}^n=\td{u}^n,\td{Y}^n=\td{y}^n\}-\prod_{i=1}^n p_{\td{X},\td{Y}}(\td{u}_i,\td{y}_i) \right| \leq \delta.
	\]
\end{lemma}

\begin{proof}
	See Appendix~\ref{appendix:shaping_distance}.
\end{proof}

\noindent Intuitively, given an i.i.d sequence $\td{Y}^n$ distributed according to $p(\td{y})$, Lemma~\ref{lemma:shaping_distance} suggests a general method of constructing a sequence $\td{U}^n$ such that the joint distribution of $(\td{Y}^n,\td{U}^n)$ is $\delta$-away in total variable distance from a given joint i.i.d distribution $p(\td{x},\td{y})$. This can be done using a point-to-point channel code with a shaping distance $\delta$ over the symmetrized channel $\bar{p}$ corresponding to $p(\td{x},\td{y})$. This technique will be used in several constructions in the paper.

Now, we are ready to describe the lossy source coding scheme of a general asymmetric source. The coding scheme can be constructed from the following point-to-point channel code and lossless source code.

\begin{lego}[\textbf{P2P} $\to$ \textbf{Lossy}]
	a $(k_1,n)$ linear point-to-point channel code $(H_1,\phi_1)$ with codebook $\cC_1$ for the channel
	\begin{equation} \label{eqn:p2p_lossy}
		\bar{p}(x, v\cond\hat{x}) = p_{X, \widehat{X}}(x, \hat{x}\oplus v).
	\end{equation}
	Let $\delta$ denote the shaping distance of the code $(H_1,\phi_1)$ with respect to the channel $\bar{p}$.
\end{lego}

\begin{lego}[\textbf{Lossless} $\to$ \textbf{Lossy}]
	an $(n-k_2,n)$ lossless source code $(H_2,\phi_2)$ for a $\BERN(\alpha)$ source with average probability of error $\eps$. Let $\cC_2$ be the codebook corresponding to $H_2$. We assume that $\cC_2 \subseteq \cC_1$, i.e., the two codebooks are nested. 
\end{lego}

\begin{remark}
	The channel $\bar{p}$ defined in (\ref{eqn:p2p_lossy}) is the symmetrized channel corresponding to the joint distribution $p(x,\hat{x})$ (see Section~\ref{sec:symmetrized} for a formal definition of a symmetrized channel).
\end{remark}

\begin{remark}
	Since $\cC_2 \subseteq \cC_1$, we will assume, without loss of generality, that $H_1$ is a submatrix of $H_2$, i.e., $H_2 = \begin{bmatrix}
		H_1\\
		Q
	\end{bmatrix}$ for some $(k_1-k_2)\times n$ matrix $Q$.\footnote{Note that such a relation between $H_1$ and $H_2$ can be obtained for any pair of nested linear codes by basic row operations and column permutations.}
\end{remark}

Starting from the aforementioned building blocks, Figure~\ref{fig:lossy_asym} shows the block diagram of the lossy source coding scheme, where $V^n$ is an i.i.d. $\BERN(1/2)$ random dither shared between the encoder and the decoder. The lossy encoder generates the sequence $U^n=\phi_1(X^n,V^n)\oplus V^n$ which has a distribution that is $\delta$-away in total variation distance from the i.i.d. $\BERN(\alpha)$ distribution. This is a consequence of Lemma~\ref{lemma:shaping_distance}. Further, since $H_1U^n = H_1V^n$ and
\[
H_2U^n = \begin{bmatrix}
	H_1U^n \\
	QU^n
\end{bmatrix} = \begin{bmatrix}
	H_1V^n\\
	QU^n
\end{bmatrix},
\]
the lossy decoder is able to reconstruct an estimate $\widehat{X}^n$ of the sequence $U^n$ using only the index $QU^n$ (since $V^n$ is shared randomness with the decoder) and the lossless source decoder $\phi_2$. The following lemma states that the joint distribution of $(X^n, \widehat{X}^n)$ is $(\delta+\eps)$-away in total variation distance from the desired i.i.d. $p(x,\hat{x})$ distribution.

\begin{lemma} \label{lemma:lossy}
	Let $q(x^n,\hat{x}^n)$ denote the distribution of $(X^n,\widehat{X}^n)$, and let $p(x^n,\hat{x}^n)$ be the desired i.i.d. $p(x,\hat{x})$ distribution, where $p(x,\hat{x})$ is as defined in~(\ref{eqn:lossy_joint_distribution}). Then,
	\begin{equation*}
		\frac{1}{2}\sum_{x^n,\hat{x}^n}\left| q(x^n,\hat{x}^n) - p(x^n,\hat{x}^n) \right| \leq \delta + \eps.
	\end{equation*}
\end{lemma}

\begin{proof}
	See Appendix~\ref{appendix:lossy}.
\end{proof}

Therefore, the coding scheme can be summarized as follows.

\begin{figure*}[t]
	\centering
	\hspace*{4em}
	\def\svgscale{1.25}
	\input{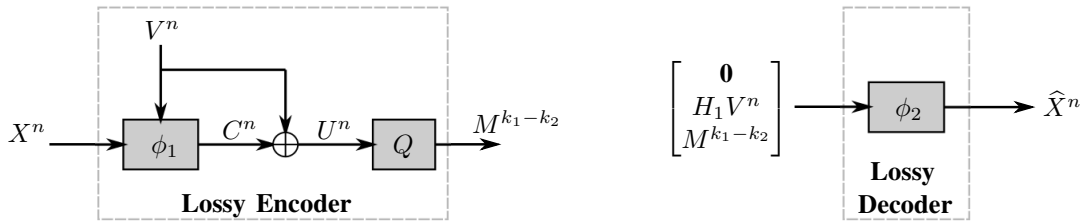}
	\caption{Encoder and decoder of a lossy source code for an asymmetric source starting from a point-to-point channel code and a lossless source code.}
	\label{fig:lossy_asym}
\end{figure*}

\vspace{0.25em}
\noindent \emph{Encoding:} Upon observing the source sequence $x^n$, the encoder computes the sequence $u^n = \phi_1\left(x^n, v^n\right)\oplus v^n$ and transmits the index $m^{k_1-k_2} = Qu^n$, where $v^n$ is a realization of a random dither shared with the decoder.

\vspace{0.25em}
\noindent \emph{Decoding:} Upon observing the index $m^{k_1-k_2}$, the decoder declares the sequence 
\[
\hat{x}^n = \phi_2\left( \begin{bmatrix}\bf 0 \\ H_1v^n \\ m^{k_1-k_2}\end{bmatrix} \right)
\]
as the source estimate.

\vspace{0.25em}
\noindent \emph{Analysis of the average distortion:} The average distortion of the coding scheme can be bounded as
\begin{align}
	&\frac{1}{n}\E[d_H(X^n,\widehat{X}^n)] = \frac{1}{n}\sum_{x^n,\hat{x}^n}q(x^n,\hat{x}^n)d_H(x^n,\hat{x}^n) \nonumber \\
	&= \frac{1}{n}\sum_{x^n,\hat{x}^n}p(x^n,\hat{x}^n)d_H(x^n,\hat{x}^n) \\
    &\quad + \frac{1}{n}\sum_{x^n,\hat{x}^n}\left(q(x^n,\hat{x}^n) - p(x^n,\hat{x}^n)\right) d_H(x^n,\hat{x}^n)\nonumber \\
	&\overset{(a)}{\leq} D + \frac{1}{2}\sum_{x^n,\hat{x}^n}\left|q(x^n,\hat{x}^n) - p(x^n,\hat{x}^n)\right|\nonumber \\
	&\overset{(b)}{\leq} D+\delta+\eps,
\end{align}
where $(a)$ holds since $p(x^n\cond \hat{x}^n)$ is equivalent to $n$ independent uses of $\BSC(D)$ and the fact that $\sum_{i}c_i(a_i -b_i) \leq \frac{1}{2}\sum_{i}\left|a_i-b_i\right|$ whenever $0\leq c_i \leq 1$ and $\sum_{i}a_i = \sum_{i}b_i$, and $(b)$ follows by Lemma~\ref{lemma:lossy}.

\vspace{0.25em}
\noindent \emph{Rate:} The rate of the coding scheme is $R = \frac{k_1-k_2}{n}$.

\begin{remark} \label{remark:lossy_asym}
	A sequence of linear point-to-point channel codes for the channel $\bar{p}$ with a vanishing shaping distance $\delta$ exists if (and only if) the rate is larger than $1-H(\widehat{X} \cond X)$. This follows by the discussion in Section~\ref{sec:properties} and the properties of a symmetrized channel (Remark~\ref{remark:symmetrized}).
\end{remark}

\begin{remark} \label{remark:lossy_asym_achievability}
	If the rate of the point-to-point channel code is $\frac{k_1}{n}=1-H(\widehat{X} \cond X)+\gamma_1$ for some $\gamma_1>0$, and the rate of the lossless source code is $\frac{n-k_2}{n} = H(\alpha)+\gamma_2 = H(\widehat{X}) + \gamma_2$ for some $\gamma_2 > 0$, then the rate of the lossy source code is
	\[
	\frac{k_1-k_2}{n} = I(X;\widehat{X}) + \gamma_1 + \gamma_2 = H(\theta)-H(D) + \gamma_1 + \gamma_2.
	\]
\end{remark}

\vspace{0.25em}
\noindent {\bf Conclusion:} Starting from a $(k_1,n)$ linear point-to-point channel code with shaping distance $\delta$ and an $(n-k_2,n)$ lossless source code with an average probability of error $\eps$, we have constructed a $(\frac{k_1-k_2}{n},n)$ lossy source code that targets a conditional distribution $p(\hat{x} \cond x)$ with an average distortion that is bounded by $D+\delta+\eps$. Note that the point-to-point channel code should be designed for the symmetrized channel $\bar{p}$ -- and not, for example, for $\BSC(D)$ -- since the source is asymmetric, and, hence, the source sequence $X^n$ is not distributed according to the channel output distribution of a $\BSC$.

\subsection{Simulation Results} \label{sec:lossy_simulation}
The lossy source coding scheme of Fig.~\ref{fig:lossy_asym} is simulated for a $\BERN(0.3)$ source (i.e., $\theta = 0.3$) using polar codes with successive cancellation decoding as the constituent point-to-point channel codes. The lossless source decoder used in the construction can be implemented using a polar code designed for a binary symmetric channel, as described in Section~\ref{sec:lossless}. To construct the polar codes (i.e., identify the information sets), we use Ar{\i}kan's method of sorting upper bounds on the Bhattacharyya parameters of the synthetic polar bit-channels~\cite{Arikan2009}. Since our coding scheme requires that the two codes are nested, the information set corresponding to the polar code $(H_2,\phi_2)$ is chosen to be a subset of that of the code $(H_1,\phi_1)$. Let $R_1$ and $R_2$ be the rates of the two polar codes.

\begin{figure}[tbp]
	\centering
	\hspace*{-0.8em}
	\includegraphics[width=\columnwidth]{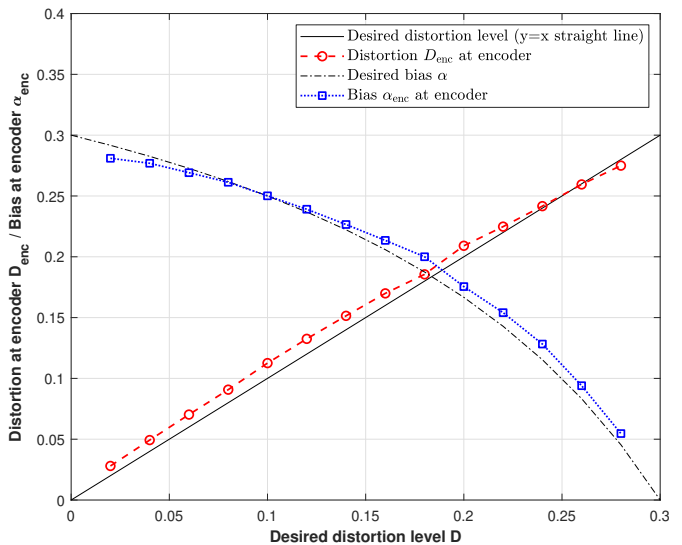}
	\caption{Distortion level and bias of the sequence $U^n$ at the encoder side assuming a $\BERN(0.3)$ source and a polar code of block length $n=1024$.}
	\label{fig:simulation_lossy_distortion_encoder}
\end{figure}

We first consider the encoder of Fig.~\ref{fig:lossy_asym}. We would like to see if the decoder $\phi_1$ correctly shapes the sequence $U^n$ according to the desired distribution. To this end, let $D_{\mathrm{enc}}$ and $\alpha_{\mathrm{enc}}$ denote respectively the distortion level and the bias of the sequence $U^n$ at the encoder side, i.e.,
\begin{equation*}
    \begin{aligned}
        D_{\mathrm{enc}} &= \frac{1}{n}\E[d_H(U^n,X^n)],\\
        \alpha_{\mathrm{enc}} &= \frac{1}{n}\E[\weight(U^n)].
    \end{aligned}
\end{equation*}
Fig.~\ref{fig:simulation_lossy_distortion_encoder} shows the plot of the achieved distortion level and bias at the encoder side for a block length $n=1024$. For comparison, the desired distortion and bias are also plotted. Note that the desired bias $\alpha$ corresponds to the mapping $D \mapsto \frac{\theta-D}{1-2D}$. At each distortion level $D$, the rate of the polar code $(H_1,\phi_1)$ is chosen to be close to the theoretical limit, i.e., we take 
\[
R_1 = 1-H(\widehat{X} \cond X).
\]
The results demonstrate that the achieved distortion and bias at the encoder side follow closely the desired design values. This implies that polar codes indeed have good shaping properties, even at finite block length.

Next, the entire lossy source coding scheme of Fig.~\ref{fig:lossy_asym} is simulated. The rate-distortion tradeoff for the coding scheme is shown in Fig.~\ref{fig:simulation_lossy} for different block lengths $n = 256$, $n=1024$, and $n=4096$. For good error performance of the lossless source decoder, the rate $R_2$ should be chosen with a small gap to its theoretical limit, i.e., we take 
\[
R_2 = 1-H(\widehat{X}) -\gamma,
\]
for some small ``back-off'' parameter $\gamma > 0$. In our simulations, we used $\gamma = 1/8$. Each simulation point shown in Fig.~\ref{fig:simulation_lossy} corresponds to a chosen rate pair $(R_1,R_2)$ of the constituent channel codes, where the rate of the coding scheme is $R=R_1-R_2$. For reference, the rate-distortion function is also shown. Clearly, the practical performance approaches the theoretical limit for increasing block lengths. The simulation results demonstrate that off-the-shelf codes can be leveraged in the construction of practical lossy source coding schemes. For more details about the simulation setup, our code is available on GitHub~\cite{Github}.

\begin{figure}[tbp]
	\centering
	\hspace*{-0.8em}
	\includegraphics[width=\columnwidth]{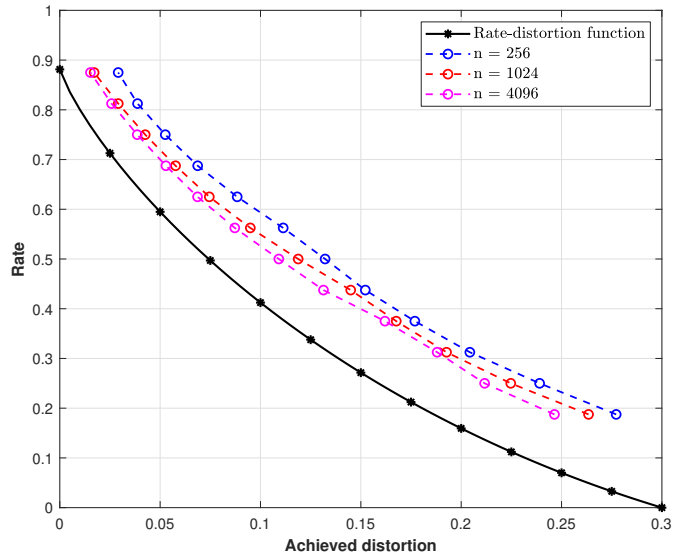}
	\caption{Rate-distortion tradeoff achieved by the lossy source coding scheme for a $\BERN(0.3)$ source using polar codes of different block lengths.}
	\label{fig:simulation_lossy}
\end{figure}

\subsection{Extension to Wyner--Ziv Coding} \label{sec:wyner}
The lossy source coding scheme presented in this section can be easily extended to the binary Wyner--Ziv coding problem, simply by replacing the lossless source code with a Slepian--Wolf code. Let us first review the binary Wyner--Ziv coding problem~\cite{Wyner1976}. This problem consists of a source alphabet $\cX = \{0,1\}$, arbitrary side information alphabet $\cY$, a reconstruction alphabet $\widehat{\cX} = \{0,1\}$, and a joint pmf $p(x,y)$ over $\cX \times \cY$. The source generates a jointly i.i.d. random process $\{(X_i,Y_i)\}$ with $(X_i,Y_i)\sim p(x,y)$. The goal is to efficiently represent a length-$n$ source sequence $X^n$ to a decoder which has access to the side information sequence $Y^n$ and wishes to reconstruct the source sequence up to some distortion level $D$. The definitions of a Wyner--Ziv code, its expected distortion and achievable rates are the same as in Section~\ref{sec:lossy_problem}, with the exception that the decoding function takes an additional input, namely, the side information sequence. Wyner and Ziv~\cite{Wyner1976} showed that for any conditional pmf $p(\hat{x}\cond x)$ such that $\E[d(X,\widehat{X})] \leq D$, any rate $R > I(X;\widehat{X} \cond Y)$ is achievable with a distortion level $D$.

In what follows, let $D$ be some desired distortion level, and let $p(\hat{x} \cond x)$ be a desired conditional pmf of the reconstruction given the source such that $\E[d(X,\widehat{X})] \leq D$, where $d(.,.)$ is the Hamming distortion metric. A code for the Wyner--Ziv coding problem can be constructed starting from the following point-to-point channel code and Slepian--Wolf code.

\begin{lego}[\textbf{P2P} $\to$ \textbf{WZ}]
a $(k_1,n)$ linear point-to-point channel code $(H_1,\phi_1)$ with codebook $\cC_1$ for the channel
\begin{equation*}
\bar{p}(x, v\cond\hat{x}) = p_{\widehat{X},X}(\hat{x}\oplus v, x).
\end{equation*}
Let $\delta$ denote the shaping distance of the code $(H_1,\phi_1)$ with respect to the channel $\bar{p}$.
\end{lego}

\begin{lego}[\textbf{SW} $\to$ \textbf{WZ}]
an $(n-k_2,n)$ Slepian--Wolf code $(H_2,\phi_2)$ for the problem
\[
p(\hat{x},y) = \sum_{x}p(x,y)p(\hat{x} \cond x),
\]
with codebook $\cC_2$ and average probability of error $\eps_2$. We assume that the two codebooks are nested, i.e., $\cC_2 \subseteq \cC_1$. 
\end{lego}

Fig.~\ref{fig:wyner} shows the block diagram of the Wyner--Ziv coding scheme, where $V^n$ is an i.i.d. $\BERN(1/2)$ random dither shared between the encoder and the decoder. The main difference in comparison to the lossy source coding scheme of Section~\ref{sec:lossy_asym} is that a Slepian--Wolf decoder is used instead of the lossless source decoder. The Slepian--Wolf decoder utilizes the available side information sequence at the decoder side. The description of the coding scheme and the analysis of its average distortion follow similarly as in the lossy source coding scheme.

\begin{remark} \label{remark:wyner_rates}
If the rate of the point-to-point channel code is $\frac{k_1}{n}=1-H(\widehat{X} \cond X)+\gamma_1$ for some $\gamma_1>0$, and the rate of the Slepian--Wolf code is $\frac{n-k_2}{n} = H(\widehat{X} \cond Y)+\gamma_2$ for some $\gamma_2 > 0$, then the rate of the Wyner--Ziv code is
\begin{equation*}
    \begin{aligned}
        \frac{k_1-k_2}{n} &= H(\widehat{X} \cond Y) - H(\widehat{X} \cond X)+ \gamma_1 + \gamma_2 \\
        &\overset{(a)}{=} I(X;\widehat{X} \cond Y) + \gamma_1 + \gamma_2,
    \end{aligned}
\end{equation*}
where the equality $(a)$ holds since $Y$ and $\widehat{X}$ are independent given $X$.
\end{remark}

\begin{figure*}[t]
	\centering
	\hspace*{4em}
	\def\svgscale{1.25}
	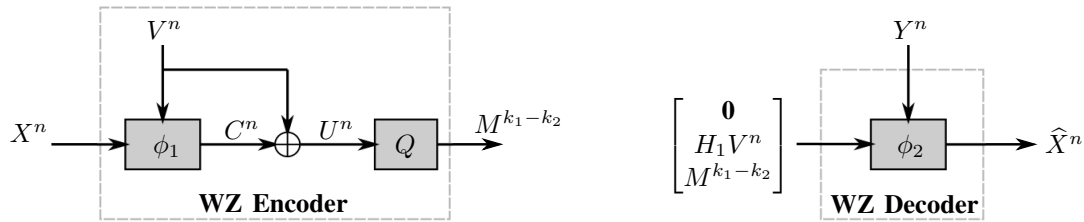
	\caption{Encoder and decoder of a Wyner--Ziv coding scheme starting from a point-to-point channel code and a Slepian--Wolf code.}
	\label{fig:wyner}
\end{figure*}

\vspace{0.25em}
\noindent {\bf Conclusion:} Starting from a $(k_1,n)$ linear point-to-point channel code with shaping distance $\delta$ and an $(n-k_2,n)$ Slepian--Wolf code with an average probability of error $\eps$, we have constructed a $(\frac{k_1-k_2}{n},n)$ Wyner--Ziv code that targets a conditional distribution $p(\hat{x} \cond x)$ with an average distortion that is bounded by $D+\delta+\eps$.

\section{Gelfand--Pinsker Coding} \label{sec:gelfand}
\subsection{Problem Statement}
The binary-input Gelfand--Pinsker problem consists of a discrete memoryless channel with state $p(y \cond x,s)p(s)$, input alphabet $\cX = \{0,1\}$, state alphabet $\cS$, output alphabet $\cY$, a collection of conditional probability mass functions $p(y,s\cond x)$ on $\cY \times \cS$ for each $x \in \cX$, and a probability mass function $p(s)$ on $\cS$, where the state sequence $(S_1, S_2, \ldots)$ is i.i.d. with $S_i \sim p(s_i)$ and is available noncausally only at the encoder~\cite{Gelfand1980}. An $(R, n)$ code $(g,\psi)$ for the Gelfand--Pinsker problem $p(y|x,s)p(s)$ consists of
\begin{itemize}
    \item a message set $\cM$ such that $\abs{\cM} = 2^{nR}$,
	\item an encoder $g: \cM \times \cS^n \to \cX^n$ that assigns a codeword $x^n = g(m,s^n)$ to each message $m$ and state sequence $s^n$, and
	\item a decoder $\psi: \cY^n \to \cM$ that assigns an estimate $\hat{m} = \psi(y^n)$ to each received sequence $y^n$.
\end{itemize}
The average probability of error of the code is $\P\{\widehat{M} \neq M\}$. An $(R, n)$ Gelfand--Pinsker code $(g,\psi)$ is depicted in Fig.~\ref{fig:gelfand_code}. 

A rate $R$ is said to be achievable for the Gelfand--Pinsker problem if there exists a sequence of $(R,n)$ Gelfand--Pinsker codes with vanishing error probability asymptotically. The classical result by Gelfand and Pinsker~\cite{Gelfand1980} states that any rate
\begin{equation} \label{eqn:gelfand_capacity}
R < \underset{p(x|s)}{\max} \big(I(X;Y) - I(X;S)\big),
\end{equation}
is achievable for the Gelfand--Pinsker problem. Gelfand and Pinsker's random coding scheme assigns, for each message $m$ and state sequence $s^n$, a codeword $x^n$ that is jointly typical with $s^n$ for some conditional pmf $p(x|s)$.

\begin{figure}[t] {
		\centering
		\hspace*{1em}
		\def\svgscale{1.25}
		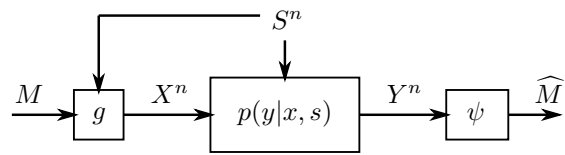
		\caption{A Gelfand--Pinsker code $(g,\psi)$ for a channel with a state that is known noncausally at the encoder.}
		\label{fig:gelfand_code}
	}
\end{figure}


\subsection{Coding Scheme} \label{sec:gelfand_coding_scheme}
In what follows, we describe a coding scheme for the Gelfand--Pinsker problem $p(y \cond x,s)p(s)$ starting from a Slepian--Wolf code and a point-to-point channel code designed for a symmetric channel. To this end, let $S^n$ be the state sequence that is i.i.d. according to $p(s)$, and let $p(x \cond s)$  be a desired conditional distribution of the channel input given the state sequence\footnote{For example, $p(x\cond s)$ can be chosen to be the maximizer of~(\ref{eqn:gelfand_capacity}).}. 

Inspired by Gelfand and Pinsker's random coding scheme, the main idea of the proposed coding scheme for the Gelfand--Pinsker problem is to shape the channel input sequence according to the desired conditional distribution $p(x\cond s)$. At the same time, the channel input sequence should encode the message to the decoder. As we shall see, a key ingredient in achieving these two goals is the nested structure of a pair of linear codes. More specifically, the Gelfand--Pinsker coding scheme can be constructed from the following Lego bricks.
\begin{lego}[\textbf{SW} $\to$ \textbf{GP}] \label{lego:slepian_gelfand}
an $(n-k_1,n)$ linear Slepian--Wolf code $(H_1,\phi_1)$ for the problem 
\[
p(x,y) = \sum_{s}p(s)p(x \cond s)p(y \cond x,s).
\]
Let $\cC_1$ be the codebook corresponding to $H_1$, and let $\eps$ be the average probability of error of the Slepian--Wolf code.
\end{lego}

\begin{lego}[\textbf{P2P} $\to$ \textbf{GP}] \label{lego:p2p_gelfand}
a $(k_2,n)$ linear point-to-point channel code $(H_2,\phi_2)$ with codebook $\cC_2$ for the channel
\begin{equation} \label{eqn:gelfand_symmetrized}
    \bar{p}(s,v \cond x) = p_{X,S}(x\oplus v, s).
\end{equation}
Let $\delta$ denote the shaping distance of the code $(H_2,\phi_2)$ with respect to the channel $\bar{p}$. Furthermore, we assume that the two codes are nested, i.e., $\cC_2 \subseteq \cC_1$.
\end{lego}

\begin{remark}
Since $\cC_2 \subseteq \cC_1$, we will assume, without loss of generality, that $H_1$ is a submatrix of $H_2$, i.e., $H_2 = \begin{bmatrix}
	H_1\\
	Q
\end{bmatrix}$ for some $(k_1-k_2)\times n$ matrix $Q$. Further, let $H_2 = \begin{bmatrix}A & B\end{bmatrix}$, where $B$ is nonsingular, and let $\wtilde{H}_2$ be as defined in Remark~\ref{remark:H_tilde}.
\end{remark}

\begin{remark}
The channel $\bar{p}$ in~(\ref{eqn:gelfand_symmetrized}) is the symmetrized channel corresponding to the desired joint distribution $p(x,s)$.
\end{remark}

\begin{figure*}[t]
	\centering
	\hspace*{1em}
	\def\svgscale{1.25}
	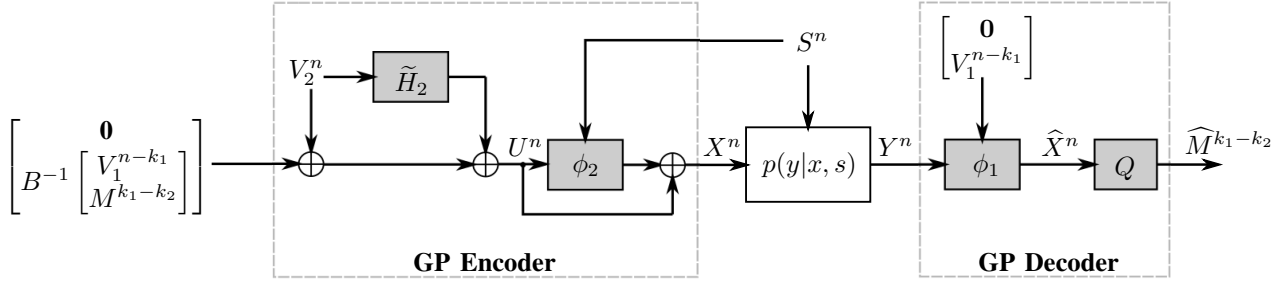
	\caption{A Gelfand--Pinsker coding scheme starting from a point-to-point channel code and a Slepian--Wolf code.}
	\label{fig:gelfand}
\end{figure*}

Fig.~\ref{fig:gelfand} shows the block diagram of the Gelfand--Pinsker coding scheme, where $V_1^{n-k_1}$ is an i.i.d. $\BERN(1/2)$ random dither shared between the encoder and the decoder, and $V_2^n$ is an i.i.d. $\BERN(1/2)$ sequence generated independently at the encoder (not necessarily shared with the decoder). To simplify the notation, let us denote the input to the Gelfand--Pinsker encoder in Fig.~\ref{fig:gelfand} by
\begin{equation} \label{eqn:zn}
    Z^n \triangleq \begin{bmatrix}\bf 0 \\
B^{-1}\begin{bmatrix}V_1^{n-k_1}\\
M^{k_1-k_2}\end{bmatrix}\end{bmatrix}.
\end{equation}
The coding scheme uses the decoding function $\phi_2$ to shape the channel input $X^n$ according to the desired distribution. To see why this holds in the construction of Fig.~\ref{fig:gelfand}, let $U^n = Z^n \oplus V_2^n \oplus \wtilde{H}_2V_2^n$ denote one of the two inputs to $\phi_2$, where the second input is $S^n$. Notice that by Lemma~\ref{lemma:codify}, $V_2^n \oplus \wtilde{H}_2V_2^n \, \sim \Unif(\cC_2)$, and $U^n$ is i.i.d. $\BERN(1/2)$ (and independent of $S^n$). Hence, by Lemma~\ref{lemma:shaping_distance}, the sequence $X^n =\phi_2(S^n,U^n)\oplus U^n$ satisfies that
\begin{equation} \label{eqn:gelfand}
    \frac{1}{2}\sum_{x^n,s^n} \left| \P\{X^n = x^n, S^n=s^n\} - \prod_{i=1}^n p(x_i, s_i) \right| \, \leq \delta,
\end{equation}
and thus, the distribution of $(X^n,S^n)$ is $\delta$-away in total variation distance from the desired joint distribution. 

Moreover, it holds that
\begin{equation*}
    \begin{aligned}
        \begin{bmatrix}
            H_1X^n \\
            QX^n
        \end{bmatrix} &= H_2X^n = H_2U^n \overset{(a)}{=} H_2Z^n \\
        &= \begin{bmatrix}
            A & B
        \end{bmatrix}\begin{bmatrix}
            \bf 0 \\
            B^{-1}\begin{bmatrix}V_1^{n-k_1}\\
            M^{k_1-k_2}\end{bmatrix}
        \end{bmatrix} = \begin{bmatrix}V_1^{n-k_1}\\
            M^{k_1-k_2}\end{bmatrix},
    \end{aligned}
\end{equation*}
where $(a)$ follows since $V_2^n \oplus \wtilde{H}_2V_2^n \in \cC_2$. Therefore, $H_1X^n = V_1^{n-k_1}$, which is the index inputted to the Slepian--Wolf decoder $\phi_1$ at the decoder side.

\begin{remark}
    Intuitively, the Gelfand--Pinsker construction can be understood as follows. The sequence $V_1^{n-k_1}$ represents a coset shift of the outer code $\cC_1$, whereas the message $M^{k_1-k_2}$ represents a coset shift of the inner code $\cC_2$ within the outer code. Since $H_2X^n = H_2Z^n$ (by construction), the channel input $X^n$ belongs to the coset of the inner code $\cC_2$ indexed by $(V_1^{n-k_1}, M^{k_1-k_2})$. Since the sequence $V_1^{n-k_1}$ is shared between the encoder and the decoder, and due to the nested structure of the two codes, the coset shift with respect to the outer code $\cC_1$ is known to the decoder, which can be leveraged by the Slepian--Wolf decoder to recover an estimate of the channel input sequence, and, hence, the message. Note that the idea of using such a nested structure to code over a Gelfand--Pinsker channel has been considered in~\cite{Pradhan2011}, where a joint typicality encoder and decoder was used. 
\end{remark}

Therefore, the coding scheme can be summarized as follows.

\vspace{0.25em}
\noindent \emph{Encoding:} To transmit the message $m^{k_1-k_2}$ upon observing the state sequence $s^n$, the encoder computes the sequence $z^n$ as in (\ref{eqn:zn}) using a random dither $v_1^{n-k_1}$ shared with the decoder, and sends $x^n = \phi_2(s^n, z^n \oplus v_2^n \oplus \wtilde{H}_2v_2^n) \oplus z^n \oplus v_2^n \oplus \wtilde{H}_2v_2^n$ over the channel, where $v_2^n$ is a random dither generated independently at the encoder.

\vspace{0.25em}
\noindent \emph{Decoding:} Upon observing the channel output $y^n$, the decoder computes $\hat{x}^n = \phi_1\left(\begin{bmatrix}\bf 0 \\ v_1^{n-k_1}\end{bmatrix}, y^n\right)$, and declares $\hat{m}^{k_1-k_2} = Q\hat{x}^n$ as the message estimate.

\vspace{0.25em}
\noindent \emph{Analysis of the probability of error:} Let $q(x^n,s^n)$ be the distribution of $(X^n,S^n)$ in Fig.~\ref{fig:gelfand}, and let $p(x^n,s^n)$ be the i.i.d. distribution according to $p(x,s)$. The average probability of error can be bounded as
\begin{align} 
	&\P\{\widehat{M}^{k_1-k_2} \neq M^{k_1-k_2}\} \leq \P\{\widehat{X}^n \neq X^n\}  \nonumber \\
	&=\hspace*{-0.2em}\sum_{x^n,s^n} \hspace*{-0.2em}\P\{\widehat{X}^n \hspace*{-0.2em}\neq \hspace*{-0.2em}X^n | X^n\hspace*{-0.2em}=\hspace*{-0.2em}x^n, S^n\hspace*{-0.2em}=\hspace*{-0.2em}s^n\} q(x^n,s^n)  \nonumber \\ 
	&=\hspace*{-0.2em} \sum_{x^n,s^n} \hspace*{-0.2em}\P\{\widehat{X}^n \hspace*{-0.2em}\neq X^n\hspace*{-0.2em} | X^n\hspace*{-0.2em}=\hspace*{-0.2em}x^n, S^n\hspace*{-0.2em}=\hspace*{-0.2em}s^n\}\big(q(x^n,s^n) - p(x^n,s^n) \big)\nonumber \\
	&\quad +\sum_{x^n,s^n}\P\{\widehat{X}^n \hspace*{-0.2em}\neq\hspace*{-0.2em} X^n | X^n\hspace*{-0.2em}=\hspace*{-0.2em}x^n, S^n\hspace*{-0.2em}=\hspace*{-0.2em}s^n\}p(x^n,s^n) \nonumber \\
	&\overset{(a)}{\leq} \frac{1}{2}\sum_{x^n,s^n} \big| q(x^n,s^n) - p(x^n,s^n) \big| + \epsilon \nonumber \\
	&\overset{(b)}{\leq} \delta + \epsilon,
\end{align}
where $(a)$ holds since $\sum_{i}c_i(a_i -b_i) \leq \frac{1}{2}\sum_{i}\left|a_i-b_i\right|$ whenever $0\leq c_i \leq 1$ and $\sum_{i}a_i = \sum_{i}b_i$ and by the fact that $\P\{\widehat{X}^n \neq X^n | X^n=x^n, S^n=s^n\}$ depends only on the channel $p(y|x,s)$, and $(b)$ follows by equation (\ref{eqn:gelfand}). 

\vspace{0.25em}
\noindent \emph{Rate:} The rate of the coding scheme is $R = \frac{k_1-k_2}{n}$.

\begin{remark}
Recall from Section~\ref{sec:properties} and Remark~\ref{remark:symmetrized} that a sequence of point-to-point channel codes with a vanishing shaping distance $\delta$ over the channel $\bar{p}$ defined in (\ref{eqn:gelfand_symmetrized}) exists if (and only if) the rate is larger than $1-H(X\cond S)$.
\end{remark}

\begin{remark} \label{remark:gelfand_rates}
If the rate of the point-to-point channel code is $\frac{k_2}{n} = 1-H(X \cond S) + \gamma_1$ for some $\gamma_1 > 0$, and the rate of the Slepian--Wolf code is $\frac{n-k_1}{n} = H(X \cond Y) + \gamma_2$ for some $\gamma_2 > 0$, then the rate of the Gelfand--Pinsker coding scheme is
\begin{equation*}
    \begin{aligned}
        R=\frac{k_1-k_2}{n} &= H(X \cond S) - H(X \cond Y) - \gamma_1 - \gamma_2 \\
        &= I(X;Y) - I(X;S) - \gamma_1 - \gamma_2.
    \end{aligned}
\end{equation*}
\end{remark}

\vspace{0.25em}
\noindent {\bf Conclusion:} Starting from an $(n-k_1,n)$ Slepian--Wolf code with an average probability of error $\eps$ and a $(k_2,n)$ linear point-to-point channel code with shaping distance $\delta$, we have constructed a $(\frac{k_1-k_2}{n},n)$ Gelfand--Pinsker code that targets a conditional distribution $p(x \cond s)$ of the channel input given the channel state and has an average probability of error that is bounded by $\delta+\eps$.

\subsection{Simulation Results} \label{sec:gelfand_simulation}
In this part, we simulate the performance of the Gelfand--Pinsker coding scheme over a binary-input Gaussian channel with a binary state. The channel output can be expressed as
\begin{equation} \label{eqn:gelfand_simulation_model}
Y = X + gS + Z,
\end{equation}
where $X \in \{-\sqrt{P},\sqrt{P}\}$ is the channel input, $S \in \{-\sqrt{P},\sqrt{P}\}$ is a channel state that is known noncausally only at the encoder with $\P\{S = -\sqrt{P}\} = \theta$, and $Z \sim \cN(0,1)$ is a sample from an i.i.d. Gaussian noise process. In our simulations, we use $g = 0.9$ and $\theta = 0.1$. 

By constructing the Slepian--Wolf code $(H_1,\phi_1)$ using its constituent point-to-point code (as described in Section~\ref{sec:p2p_slepian}), the Gelfand--Pinsker coding scheme can be implemented using two point-to-point channel codes. Inspired by the characterization (\ref{eqn:gelfand_capacity}) of achievable rates for the Gelfand--Pinsker problem, the proposed coding scheme targets the the conditional distribution $p^{*}(x \cond s)$ that maximizes $I(X;Y)- I(X;S)$, i.e.,
\begin{equation} \label{eqn:gelfand_optimization}
p^{*}(x \cond s) = \underset{p(x \cond s)}{\arg\max} \left( I(X;Y)- I(X;S) \right).
\end{equation}
Since $X$ and $S$ are binary, the optimization~(\ref{eqn:gelfand_optimization}) can be solved efficiently using off-the-shelf numerical solvers (or simply using a grid search over the two parameters that define $p(x\cond s)$). We compare the proposed coding scheme with the simple strategy that encodes the transmitter's message using a point-to-point channel code while ignoring the available state sequence. Since the induced channel input distribution would be $\BERN(1/2)$ in this case, we refer to this strategy as the ``symmetric coding'' strategy. 

We first look at the maximum achievable rate of the two coding strategies as a function of the power level $P$. For the proposed Gelfand--Pinsker coding scheme, the maximum achievable rate is given by
\[
C_{\mathrm{GP}} \triangleq \underset{p(x \cond s)}{\max} \left( I(X;Y)- I(X;S) \right),
\]
whereas the maximum achievable rate of the symmetric coding strategy is $C_{\mathrm{sym}} \triangleq I(X;Y)$, which is evaluated when $X$ is $\BERN(1/2)$ and independent of $S$. The solid lines in Fig.~\ref{fig:simulation_gelfand_rates} show the plots of $C_{\mathrm{GP}}$ and $C_{\mathrm{sym}}$ as a function of $P$. The plots demonstrate that, in theory, Gelfand--Pinsker coding can achieve larger rates asymptotically compared to the symmetric coding strategy.

\begin{figure}[tbp]
	\centering
	\hspace*{-0.8em}
	\includegraphics[width=\columnwidth]{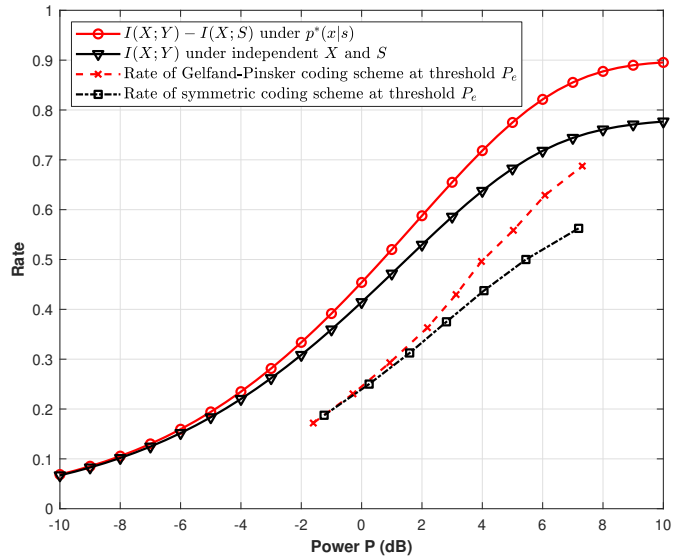}
	\caption{Achieved rates of the Gelfand--Pinsker coding scheme at a fixed block error probability $P_e^{\mathrm{threshold}} = 10^{-2}$.}
	\label{fig:simulation_gelfand_rates}
\end{figure}

Before we describe what is meant by the dashed curves of Fig.~\ref{fig:simulation_gelfand_rates}, we shift our attention to evaluating the block error rate performance of the two coding strategies over the channel model~(\ref{eqn:gelfand_simulation_model}). The coding strategies are compared for the same block length $n=1024$ and rate $R=0.5$. We use polar codes with successive cancellation decoding as the constituent point-to-point channel codes. To ensure nestedness of the constituent codes in the proposed Gelfand--Pinsker coding scheme, the information set of the polar code corresponding to $(H_2,\phi_2)$ is taken to be a subset of that corresponding to the code $(H_1,\phi_1)$. Let $R_1$ and $R_2$ be the rates of the two polar codes. The rates $R_1$ and $R_2$ are chosen ``close'' to their theoretical limits (see Remark~\ref{remark:gelfand_rates}) when the conditional distribution of the channel input given the channel state is $p^{*}(x \cond s)$. More precisely, we take
\begin{equation} \label{eqn:gelfand_rates}
    \begin{aligned}
    R_1 &= 1-H(X\cond Y)-\gamma,\\
    R_2 &= 1-H(X\cond S),
    \end{aligned}
\end{equation}
where $\gamma > 0$ is a ``back-off'' parameter from the theoretical limit, which allows for a reasonable error probability for the code $(H_1,\phi_1)$. In our Gelfand--Pinsker simulations, we used $\gamma = 3/16$. Note that for the symmetric coding strategy, the polar code is designed for the average channel $p_{\mathrm{avg}}(y \cond x) \triangleq \sum_{s}p(s)p(y \cond x, s)$. Fig.~\ref{fig:simulation_gelfand_rates} shows the block error rate performance of the two coding strategies. The Gelfand--Pinsker coding scheme shows significant performance gain compared to symmetric coding. 

By repeating the same experiment for different values of the rate $R$, one can plot $R$ as a function of the power level $P$, for a fixed block error probability $P_e^{\mathrm{threshold}}= 10^{-2}$. This is shown in the dashed curves of Fig.~\ref{fig:simulation_gelfand_rates}. For example, for a rate $R=0.5$, we know from Fig.~\ref{fig:simulation_gelfand} that the Gelfand--Pinkser coding scheme achieves an error probability of $10^{-2}$ when the power level is around 4 dB, whereas the symmetric coding scheme has an error probability of $10^{-2}$ when the power level is around 5.5 dB. A similar approach is taken for other values of $R$. The dashed curves of Fig.~\ref{fig:simulation_gelfand_rates} show that the Gelfand--Pinsker coding scheme can achieve larger rates in practice (i.e., the gain from Gelfand--Pinsker coding is not only theoretical). These results demonstrate the performance gain from ``shaping'' the channel input with respect to the known state sequence. For more details about the simulation setup, our code is available on GitHub~\cite{Github}.

\begin{figure}[tbp]
	\centering
	\hspace*{-0.8em}
	\includegraphics[width=\columnwidth]{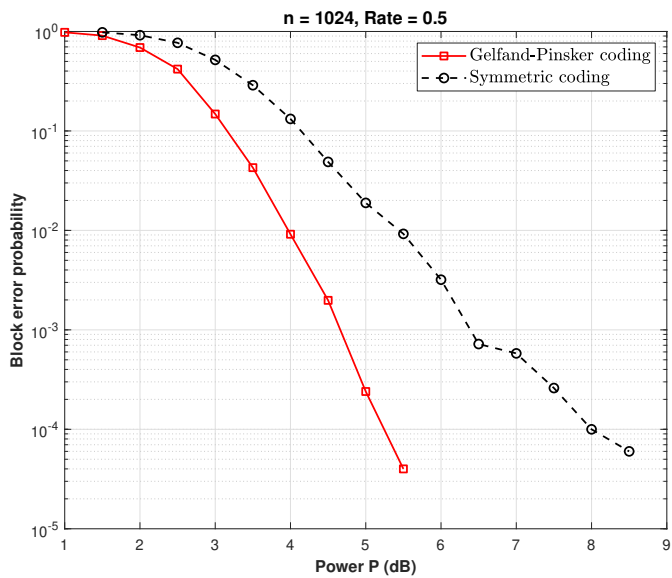}
	\caption{Simulation results of the Gelfand--Pinsker coding scheme for a block length $n=1024$ and rate $R = 0.5$ over a Gaussian channel with state.}
	\label{fig:simulation_gelfand}
\end{figure}

\begin{remark}
    Note that the symmetric coding scheme can be implemented as an instance of the Gelfand--Pinsker coding scheme (Fig.~\ref{fig:gelfand}) for the case when $R_2=0$. Therefore, the simulation results here demonstrate that using strictly positive rates $R_2$ to target the conditional pmf $p^{*}(x \cond s)$ can achieve superior performance compared to the naive approach of coding for the average channel $p_{\mathrm{avg}}$.
\end{remark}

\subsection{Specialization to Asymmetric Channel Coding} \label{sec:asym}
The problem of coding for an asymmetric point-to-point channel (Section~\ref{sec:p2p}) can be seen as a special case of the Gelfand--Pinkser coding problem when the channel state is constant and independent of the channel output. This observation has been previously formalized in~\cite{Ganguly2020}. For completion, we give the details here, along with the resulting construction of an asymmetric channel coding scheme.

Suppose we have an $(R,n)$ code $(g,\psi)$ for the binary-input Gelfand--Pinsker problem $p(y \cond x,s)p(s)$, where $p(y\cond x,s)=p(y \cond x)$ (i.e., the channel output is independent of the state given the channel input) and $p_S(0) = 1$. Let $\eps$ be the average probability of error of the code. Define $f:[2^{nR}] \to \{0,1\}^n$ by $f(m) = g(m, {\bf 0})$. Then, $(f,\psi)$ forms a code for the channel $p(y \cond x)$ with length $n$, rate $R$ and average probability of error
    \begin{align*}
        &\sum_{m}2^{-nR}.\P\{\psi(Y^n) \neq m \cond X^n = f(m) \}\\
        &= \hspace*{-0.2em}\sum_{m}2^{-nR}.\P\{\psi(Y^n) \neq m \cond X^n \hspace*{-0.2em}=\hspace*{-0.2em} g(m,{\bf 0}) \}\\
        &= \hspace*{-0.2em}\sum_{m}2^{-nR}\Big(\P\{\psi(Y^n) \neq m \cond X^n \hspace*{-0.2em}= \hspace*{-0.2em}g(m,{\bf 0}), S^n \hspace*{-0.2em}= \hspace*{-0.2em}{\bf 0} \}\prod_{i=1}^n p_S(0) \\
        &\quad + \sum_{s^n \neq {\bf 0}} \P\{\psi(Y^n) \neq m \cond X^n \hspace*{-0.2em}=\hspace*{-0.2em} g(m,{\bf 0}), S^n \hspace*{-0.2em}= \hspace*{-0.2em}s^n\}\prod_{i=1}^n p_S(s_i)\Big)\\
        &\overset{(a)}{=} \hspace*{-0.2em}\sum_{m}2^{-nR}\P\{\psi(Y^n) \neq m \cond X^n \hspace*{-0.2em}= \hspace*{-0.2em}g(m,{\bf 0}), S^n \hspace*{-0.2em}= \hspace*{-0.2em}{\bf 0} \} \\
        &\overset{(b)}{=} \hspace*{-0.2em}\sum_{m}2^{-nR}\Big(\P\{\psi(Y^n) \neq m \cond X^n \hspace*{-0.2em}= \hspace*{-0.2em}g(m,{\bf 0}), S^n \hspace*{-0.2em}= \hspace*{-0.2em}{\bf 0} \}\prod_{i=1}^n p_S(0) \\
        &\quad+ \sum_{s^n \neq {\bf 0}} \P\{\psi(Y^n) \neq m \cond X^n \hspace*{-0.2em}= \hspace*{-0.2em}g(m,s^n), S^n \hspace*{-0.2em}= \hspace*{-0.2em}s^n\}\prod_{i=1}^n p_S(s_i)\Big)\\
        &= \hspace*{-0.2em}\sum_{m, s^n}2^{-nR}\P\{\psi(Y^n) \neq m \cond X^n \hspace*{-0.2em}= \hspace*{-0.2em}g(m,s^n), S^n \hspace*{-0.2em}= \hspace*{-0.2em}s^n\}\prod_{i=1}^n p_S(s_i) \\
        &= \eps,
    \end{align*}
where $(a)$ and $(b)$ follow since $\prod_{i=1}^n p_S(s_i) = 0$ for $s^n \neq {\bf 0}$. 

Therefore, the Gelfand--Pinsker coding scheme described in Section~\ref{sec:gelfand_coding_scheme} can be used to construct an asymmetric channel code in the special case when the state sequence is the constant all-zero sequence, i.e., $p_S(0)=1$. To this end, consider a binary-input channel $p(y \cond x)$ that is not necessarily symmetric. Let the desired input distribution be $p(x) \sim \BERN(\alpha)$ for some $\alpha \in (0,1/2)$. By specializing the Lego bricks of the Gelfand--Pinsker coding scheme to the case when $p_S(0)=1$ and the desired $p(x\cond s)$ is $\BSC(\alpha)$, the asymmetric channel coding scheme can be constructed using the following Lego bricks.

\begin{lego}[\textbf{SW} $\to$ \textbf{Asym}]
an $(n-k_1,n)$ linear Slepian--Wolf code $(H_1,\phi_1)$ for the problem $p(x,y) = p(x)p(y \cond x)$ with codebook $\cC_1$ and average probability of error $\eps$.
\end{lego}

\begin{lego}[\textbf{P2P} $\to$ \textbf{Asym}]
a $(k_2,n)$ linear point-to-point channel code $(H_2,\phi_2)$ for $\mathrm{BSC}(\alpha)$ with codebook $\cC_2$ and shaping distance $\delta$. Furthermore, we assume that the two codes are nested, i.e., $\cC_2 \subseteq \cC_1$.
\end{lego}

\begin{remark}
As before, it is assumed that $H_1$ is a submatrix of $H_2$, i.e., $H_2 = \begin{bmatrix}
	H_1\\
	Q
\end{bmatrix}$ for some $(k_1-k_2)\times n$ matrix $Q$, and $H_2 = \begin{bmatrix}A & B\end{bmatrix}$, for some non-singular matrix $B$.
\end{remark}

Fig.~\ref{fig:asym} shows the block diagram for the asymmetric channel coding scheme, which can be seen as a specialization of the Gelfand--Pinsker coding scheme when the state sequence is the constant all-zero sequence. The description of the coding scheme and the analysis of its probability of error follow similarly as in the Gelfand--Pinsker case.

\begin{remark} \label{remark:asym_rates}
If the rate of the point-to-point channel code is $\frac{k_2}{n} = 1-H(X) + \gamma_1$ for some $\gamma_1 > 0$, and the rate of the Slepian--Wolf code is $\frac{n-k_1}{n} = H(X \cond Y) + \gamma_2$ for some $\gamma_2 > 0$, then the rate of the asymmetric channel coding scheme is
\[
\frac{k_1-k_2}{n} = I(X;Y) - \gamma_1 - \gamma_2.
\]
\end{remark}

\begin{figure*}[t]
	\centering
	\def\svgscale{1.25}
	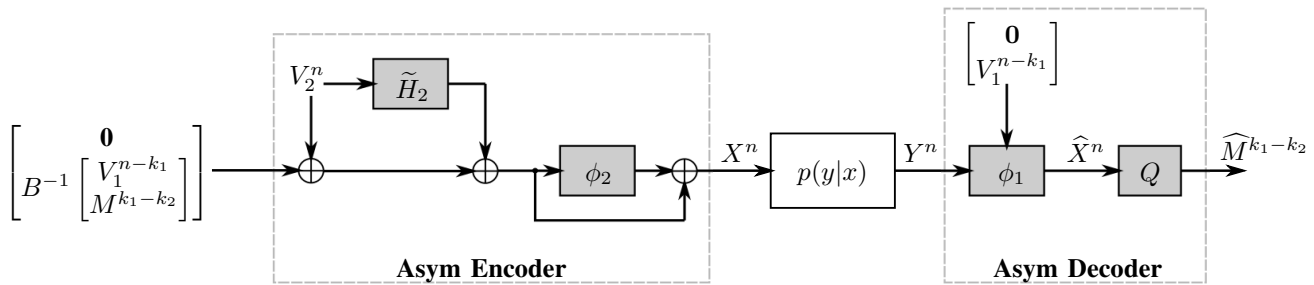
	\caption{A coding scheme for an asymmetric channel starting from a Slepian--Wolf code and a point-to-point channel code.}
	\label{fig:asym}
\end{figure*}

\begin{remark}
The asymmetric channel encoder shown in Fig.~\ref{fig:asym} is almost identical to the lossy source decoder shown in Fig.~\ref{fig:lossy_asym} (when the constituent lossless source decoder is implemented using a point-to-point channel code). Similarly, the lossy source encoder is almost identical to the asymmetric channel decoder (when the constituent Slepian--Wolf decoder is implemented using a point-to-point channel code). This suggests some form of duality between the two constructions. A similar observation can be made between the Gelfand--Pinsker and Wyner--Ziv constructions.
\end{remark}

\vspace{0.25em}
\noindent {\bf Conclusion:} Starting from an $(n-k_1,n)$ Slepian--Wolf code with an average probability of error $\eps$ and a $(k_2,n)$ linear point-to-point channel code with shaping distance $\delta$, we have constructed a $(k_1-k_2,n)$ code for an asymmetric point-to-point channel that targets an input distribution $p(x)$ and has an average probability of error that is bounded by $\delta+\eps$.

\section{Block-Markov Constructions} \label{sec:block_markov}
In the constructions we presented so far in this paper, the shaping distance of a point-to-point channel code designed for a BMS channel was used as a primitive property in one of the constituent Lego bricks. Unfortunately, the shaping distance is difficult to be estimated for commercial off-the-shelf codes since it involves the computation of the total variation distance between distributions over exponentially large alphabets. To circumvent this inconvenience, we present in this section modified coding schemes involving properties that are easily verifiable in practice. More precisely, our assumptions in this section will be over alphabets with size that is at most polynomial in the block length. Moreover, the coding schemes presented here will not make assumptions of nestedness between the constituent Lego bricks, which was often assumed in the constructions thus far. However, as we shall see, these benefits will come at the cost of larger implementation complexity, a small penalty incurred in achievable rates, and a worse performance guarantee. In particular, the coding schemes presented in this section will have a \emph{block-Markov structure}, i.e., they will be defined upon the transmission of several blocks of information and the inputs to one encoding/decoding block can depend on the outputs of previous blocks.

The following definitions will be used in the block-Markov constructions of coding schemes.

\begin{definition}[Shaping Hamming Distance Spectrum]
Given a linear point-to-point channel code $(H,\phi)$ for a binary-input binary-output channel $p(y \cond x)$, the \emph{shaping Hamming distance spectrum} $W$ is defined as
\[
W \triangleq d_H\big(V^n, \phi(V^n)\big),
\]
where $V^n$ is i.i.d. $\mathrm{Bern}(1/2)$, and $d_H(.,.)$ denotes the Hamming distance metric. Therefore, we have that $W \in \{0,1,\ldots, n\}$, and
\[
\P\{W=w\} = \frac{\abs{\{v^n \in \mathbb{F}_2^n: d_H\big(v^n,\phi(v^n)\big) = w\}}}{2^n},
\]
for $0\leq w\leq n$. 
\end{definition}

\begin{definition}[Joint-Type]
Given a pair of sequences $(x^n,y^n)\in \cX^n\times \cY^n$, the \emph{joint-type} of $(x^n,y^n)$, denoted $\mathrm{type}(x^n,y^n)$, is the number of occurrence of each symbol pair $(x,y)\in\cX\times \cY$ in $(x^n,y^n)$, i.e.,
\[
\mathrm{type}(x^n,y^n) = \left( \pi(x,y \cond x^n,y^n) \right)_{(x,y)\in \cX \times \cY},
\]
where $\pi(x,y \cond x^n,y^n) = \abs{\{i: 1\leq i \leq n, \, x_i=x, \, y_i=y\}}$.
\end{definition}

\begin{definition}[Shaping Joint-Type Spectrum]
Consider a symmetric binary-input channel $p(y \cond x)$ and a linear point-to-point channel code $(H,\phi)$. Let $Y^n$ be an i.i.d. according to $p(y)$, where $p(y) = \sum_{x}\frac{1}{2}p(y \cond x)$, and define $X^n = \phi(Y^n)$. The \emph{shaping joint-type spectrum} $T$ is defined as the joint-type of the pair $(X^n,Y^n)$, i.e.,
\[
T \triangleq \mathrm{type}(X^n,Y^n).
\]
For example, if $\cY = \{0,1\}$, then $T = (T_{00},T_{01},T_{10},T_{11})$, where $T_{xy}$ is the number of occurrences of the symbol pair $(x,y)$ in the sequences $(X^n,Y^n)$, for $x,y \in \{0,1\}$.
\end{definition}

The modified coding schemes presented in this section are based on the distributions of $W$ and $T$. The assumptions are easily-verifiable in practice since the distributions of $W$ and $T$ can be estimated efficiently for any commercial off-the-shelf code via Monte Carlo simulations. It turns out that all the previous constructions can be modified to work under these simplified assumptions. For conciseness, however, we will only focus on two problems: the asymmetric channel coding problem and the lossy source coding problem for an asymmetric source.

\subsection{Asymmetric Channel Coding} \label{sec:block_markov_asym}
Consider the problem of coding for a binary-input asymmetric channel $p(y\cond x)$, where the capacity-achieving input distribution is $p(x) \sim \mathrm{Bern}(\alpha)$ for some $\alpha \in (0,1/2)$. The coding scheme we present here is defined upon the transmission of $b$ blocks of information for some fixed $b$, and is constructed starting from the following three Lego bricks.

\begin{lego}[\textbf{SW} $\to$ \textbf{Asym}]
an $(n-k_1,n)$ linear Slepian--Wolf code $(H_1,\phi_1)$ for the problem $p(x,y) = p(x)p(y \cond x)$ with an average probability of error $\eps$.
\end{lego}

\begin{lego}[\textbf{P2P} $\to$ \textbf{Asym}]
a $(k_2,n)$ linear point-to-point channel code $(H_2,\phi_2)$ for $\mathrm{BSC}(\alpha)$ with a shaping Hamming distance spectrum $W_2$ satisfying
\begin{equation} \label{eqn:tv_asym_bm}
    \frac{1}{2}\sum_{w=0}^n \Big| \P\{W_2=w\} - \binom{n}{w}\alpha^{w}(1-\alpha)^{n-w} \Big| \, \leq \delta,
\end{equation}
for some $\delta > 0$. We assume that $k_2 < k_1$.
\end{lego}

\begin{lego}[\textbf{P2P} $\to$ \textbf{Asym}]
a $(k_{\mathrm{sym}},n)$ point-to-point channel code $(f_{\mathrm{sym}},\phi_{\mathrm{sym}})$ for the channel $p(y|x)$ with average probability of error $\eps_{\mathrm{sym}}$.\footnote{For example, one can use any code that approaches the symmetric capacity of the channel. Note that this code will be used only in the first transmission block.}
\end{lego}

\begin{remark}
Since the alphabet size of $W_2$ is linear in $n$, the distribution of $W_2$ can be estimated for any off-the-shelf channel code via simulations. Hence, condition~(\ref{eqn:tv_asym_bm}) is easily verifiable in practice. In words, this condition states that the distribution of the shaping Hamming distance spectrum is $\delta$-away in total variation distance from a $\mathrm{Binom}(n,\alpha)$ distribution.
\end{remark}

\begin{figure}[t]
	\centering
	\def\svgscale{1.25}
\begingroup%
  \makeatletter%
  \providecommand\color[2][]{%
    \errmessage{(Inkscape) Color is used for the text in Inkscape, but the package 'color.sty' is not loaded}%
    \renewcommand\color[2][]{}%
  }%
  \providecommand\transparent[1]{%
    \errmessage{(Inkscape) Transparency is used (non-zero) for the text in Inkscape, but the package 'transparent.sty' is not loaded}%
    \renewcommand\transparent[1]{}%
  }%
  \providecommand\rotatebox[2]{#2}%
  \newcommand*\fsize{\dimexpr\f@size pt\relax}%
  \newcommand*\lineheight[1]{\fontsize{\fsize}{#1\fsize}\selectfont}%
  \ifx\svgwidth\undefined%
    \setlength{\unitlength}{303.74993512bp}%
    \ifx\svgscale\undefined%
      \relax%
    \else%
      \setlength{\unitlength}{\unitlength * \real{\svgscale}}%
    \fi%
  \else%
    \setlength{\unitlength}{\svgwidth}%
  \fi%
  \global\let\svgwidth\undefined%
  \global\let\svgscale\undefined%
  \makeatother%
  \begin{picture}(1,0.0987654)%
    \lineheight{1}%
    \setlength\tabcolsep{0pt}%
    \put(0,0){\includegraphics[width=\unitlength]{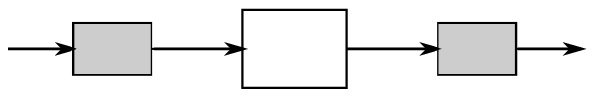}}%
    \put(0.28771915,0.04295436){\color[rgb]{0,0,0}\makebox(0,0)[lt]{\lineheight{1.25}\smash{\begin{tabular}[t]{l}$p(y|x)$\end{tabular}}}}%
    \put(0.12604036,0.04444135){\color[rgb]{0,0,0}\makebox(0,0)[lt]{\lineheight{1.25}\smash{\begin{tabular}[t]{l}$f_{\mathrm{sym}}$\end{tabular}}}}%
    \put(0.00461662,0.06309862){\color[rgb]{0,0,0}\makebox(0,0)[lt]{\lineheight{1.25}\smash{\begin{tabular}[t]{l}$H_1X_{(2)}^n$\end{tabular}}}}%
    \put(0.47036223,0.04376412){\color[rgb]{0,0,0}\makebox(0,0)[lt]{\lineheight{1.25}\smash{\begin{tabular}[t]{l}$\phi_{\mathrm{sym}}$\end{tabular}}}}%
    \put(0.56801592,0.06140678){\color[rgb]{0,0,0}\makebox(0,0)[lt]{\lineheight{1.25}\smash{\begin{tabular}[t]{l}$\widehat{H_1X_{(2)}^n}$\end{tabular}}}}%
  \end{picture}%
\endgroup%

	\caption{A block-Markov coding scheme for an asymmetric channel in the first transmission block using the code $(f_{\mathrm{sym}}, \phi_{\mathrm{sym}})$.}
	\label{fig:asym_BM_block1}
\end{figure}

Now, we are ready to describe the coding scheme. For each $j\in [b]$, let $V_{(j)}^n$ be an i.i.d. $\BERN(1/2)$ random dither shared between the encoder and the decoder, and let $\Gamma_{(j)}:[n] \to [n]$ be a permutation chosen uniformly at random and shared between the encoder and the decoder.  We describe the encoding procedure starting from the $b$th block. Given a message $M_{(b)} \in \{0,1\}^{n-k_2}$, the encoder computes the sequence $Z_{(b)}^n$ as follows.
\begin{equation} \label{eqn:asym_bm_zb}
	Z_{(b)}^n = \begin{bmatrix}
		{\bf 0} \\
		M_{(b)}
	\end{bmatrix},
\end{equation}
where ${\bf 0}$ consists of $k_2$ zeros. Then, for each $j = b,\ldots, 2$, the encoder computes the sequences $\wtilde{X}_{(j)}^n$, $X_{(j)}^n$ and $Z_{(j-1)}^n$ as follows.
\begin{equation} \label{eqn:asym_bm_zj}
	\begin{aligned}
	    \wtilde{X}_{(j)}^n &= \phi_2\big(Z_{(j)}^n \oplus V_{(j)}^n\big) \oplus Z_{(j)}^n \oplus V_{(j)}^n, \\
		X_{(j)}^n  &= \Gamma_{(j)}\left(\wtilde{X}_{(j)}^n\right), \\
		Z_{(j-1)}^n &= \begin{bmatrix}
			{\bf 0} \\
			H_1X_{(j)}^n \\
			M_{(j-1)}
		\end{bmatrix},
	\end{aligned}
\end{equation}
where $M_{(j-1)} \in \{0,1\}^{k_1-k_2}$. Note that the sequence $Z_{(j-1)}^n$ in the $(j-1)$-th transmission block depends on the syndrome vector corresponding to the channel input $X_{(j)}^n$ in the $j$th block. Furthermore, notice that the sequences $Z_{(j)}^n$ satisfy that $\wtilde{H}_2 Z_{(j)}^n = Z_{(j)}^n$ for each $j =2,\ldots, b$, where $\wtilde{H}_2$ is as defined in~(\ref{eqn:H_tilde}). Finally, in the first block, the transmitter uses the encoder $f_{\mathrm{sym}}$ to encode the syndrome vector $H_1X_{(2)}^n$, where $X_{(2)}^n$ is the transmitted sequence in the second block. Fig.~\ref{fig:asym_BM_block1} and Fig.~\ref{fig:asym_BM} show the block diagrams of the coding scheme in the first block and the remaining blocks, respectively. Note that a loss in the overall rate is incurred in the first block, where no message is transmitted. However, this loss decays as $1/b$, and, thus, by choosing $b$ large enough, the rate loss becomes negligible.

\begin{figure*}[t]
	\centering
	\hspace*{-1.5em}
	\def\svgscale{1.25}
	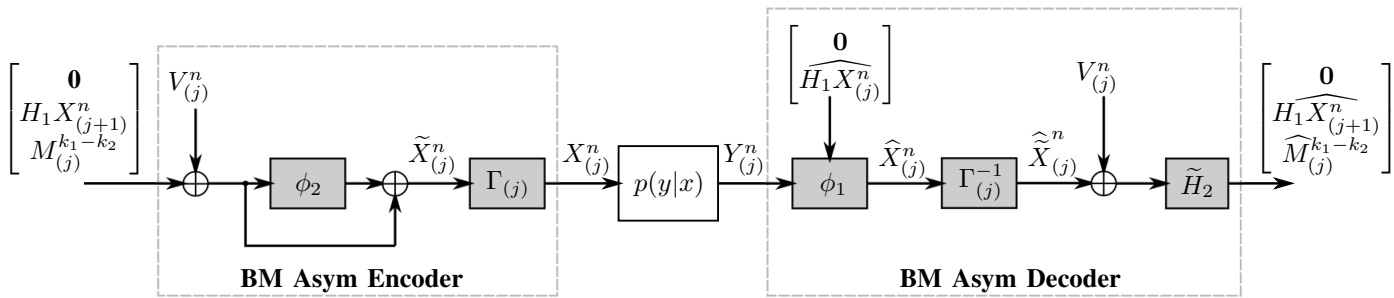
	\caption{A block-Markov coding scheme for an asymmetric channel in the $j$th transmission block, for $2\leq j \leq b-1$.}
	\label{fig:asym_BM}
\end{figure*}

The following lemma states that under this construction, the distribution of the channel input $X_{(j)}^n$ is $\delta$-away in total variation distance from the i.i.d. $\BERN(\alpha)$ distribution, for each $j=2,\ldots,b$.

\begin{lemma} \label{lemma:asym_bm}
	For the arrangement shown in Fig.~\ref{fig:asym_BM}, we have
	\begin{equation*}
		\frac{1}{2}\sum_{x^n} \left| \P\{X_{(j)}^n= x^n \} -  \alpha^{\weight(x^n)}(1-\alpha)^{n-\weight(x^n)}\right| \leq \delta,
	\end{equation*}
	for each $j=2,\ldots,b$.
\end{lemma}

\begin{proof}
    See Appendix~\ref{appendix:asym_bm}.
\end{proof}

At the decoder side, the coding scheme proceeds as follows. In the first transmission block, the decoder uses the decoder $\phi_{\mathrm{sym}}$ to get an estimate of $H_1X_{(2)}^n$. In the subsequent blocks, the decoder views $(X_{(j)}^n,Y_{(j)}^n)$ as realizations of the Slepian--Wolf problem $p(x,y)$, uses the Slepian--Wolf decoder $\phi_1$ to recover $X_{(j)}^n$ using the estimate of $H_1X_{(j)}^n$ from previous blocks, and then reverse-engineers the operations at the encoder to get an estimate of the transmitted message $M_{(j)}$ and the subsequent index $H_1X_{(j+1)}^n$. 

For the probability of error of the coding scheme, consider the performance of a genie-aided decoder which recovers an estimate of $M_{(j)}$ in the $j$th block based on the channel output $Y_{(j)}^n$ and the syndrome vector $H_1X_{(j)}^n$ (which is supplied correctly by a genie regardless of any decoding errors in previous blocks). Notice that such a decoder would have the same probability of error over the $b$ blocks as our decoder. To see this, observe that a decoding error can propagate from one block to another only through an error in the syndrome vector $H_1X_{(j)}^n$. Consider the first block where such an error happens. Both decoders would make an error in that block, which is precisely an error event over the $b$ blocks, irrespective of decisions made in subsequent blocks. A similar argument has been made in the analysis of the successive cancellation decoding of polar codes~\cite{Arikan2009}. Therefore, it suffices to analyze the error probability of the genie-aided decoder over the $b$ blocks of transmission. Using similar bounding techniques as before, the average probability of error of the genie-aided decoder in the $j$th block can be bounded as
\[
\P\{\widehat{M}_{(j)} \neq M_{(j)}\} \leq \delta + \eps,
\]
and, thus, the average probability of error of our coding scheme over the $b$ transmission blocks can be bounded using the union bound as
\begin{equation}
\P\left\{\widehat{M}_{(j)} \neq M_{(j)} \text{ for some } j \in [b]\right\} \leq (b-1)(\delta+\eps) + \eps_{\mathrm{sym}},
\end{equation}
where $\epsilon_{\mathrm{sym}}$ is the average probability of error of the code used in the first block. The rate of the coding scheme is given by
\[
R = \frac{(b-1)(k_1-k_2)+n-k_1}{nb}.
\]

\begin{remark}
A point-to-point channel code with a shaping Hamming distance spectrum that is arbitrarily close in total variation distance to a $\Binom(n,\alpha)$ distribution exists for sufficiently large $n$ if and only if the rate $k_2/n$ of the code is larger than $1-H(\alpha)=1-H(X)$. To show this, the same argument of Section~\ref{sec:properties} based on distributed channel synthesis can be used. Similarly, a Slepian--Wolf code for $p(x,y)$ with arbitrarily small error probability exists for sufficiently large $n$ if and only if its rate $\frac{n-k_1}{n}$ is larger than $H(X\cond Y)$. It follows that the rate $R$ of the block-Markov asymmetric channel coding scheme can be made arbitrarily close to $I(X;Y)$ for sufficiently large $n$ and $b$.
\end{remark}

\subsection{Lossy Source Coding for an Asymmetric Source} \label{sec:block_markov_lossy}
Consider the problem of lossy source coding of a $\BERN(\theta)$ source for some $\theta \in (0,1/2)$, as defined in Section~\ref{sec:lossy_problem}. In this part, we construct a coding scheme for this problem starting from a point-to-point channel code and a lossless source code that satisfy some easily-verifiable properties. The coding scheme will be defined upon its operation on $b$ blocks of source sequences. To this end, let $D$ be some desired distortion level, and $p(x \cond \hat{x}) \sim \mathrm{BSC}(D)$ be the desired conditional pmf of the source given the reconstruction. Under this conditional pmf, the distribution of the reconstruction is $p(\hat{x})\sim \BERN(\alpha)$, where
\[
\alpha = \frac{\theta-D}{1-2D}.
\]
The coding scheme is constructed using the following Lego bricks.

\begin{lego}[\textbf{P2P} $\to$ \textbf{Lossy}]
a $(k_1,n)$ linear point-to-point channel code $(H_1,\phi_1)$ for the channel
\[
\bar{p}(x,v \cond \hat{x}) = p_{X,\widehat{X}}(x,\hat{x}\oplus v),
\]
such that the code's shaping joint-type spectrum $T$ satisfies
\begin{equation} \label{eqn:tv_lossy_asym_bm}
d_{TV}\big(T,\, \mathrm{Multinomial}(n;p_{000},\ldots,p_{111})\big) \leq \delta
\end{equation}
for some $\delta > 0$, where $d_{TV}$ denotes the total variation distance and
\[
p_{i_1i_2i_3} = \frac{1}{2}p_{X,\widehat{X}}(i_1,i_2\oplus i_3), \qquad i_1,i_2,i_3 \in \{0,1\}.
\]
\end{lego}

\begin{lego}[\textbf{Lossless} $\to$ \textbf{Lossy}]
an $(n-k_2,n)$ lossless source code $(H_2,\phi_2)$ for a $\BSC(\alpha)$ source with average probability of error $\eps$ such that $k_2 < k_1$. 
\end{lego}

\begin{remark}
Condition~(\ref{eqn:tv_lossy_asym_bm}) says that the distribution of $T$ is $\delta$-away in total variation distance from a multinomial distribution. This condition can be easily verified in practice since the distribution of $T$ can be estimated efficiently for any commercial off-the-shelf code. Recall that a multinomial random variable $M\sim \mathrm{Multinomial}(n;p_1,\ldots,p_K)$ models the number of counts of each side of a $K$-sided die when rolled $n$ times, where $p_1, \ldots, p_K$ are the probabilities of observing each side. We have
\[
\P\{M = (m_1,\ldots, m_K)\} = \frac{n!}{m_1! \cdots m_K!}p_1^{m_1}\cdots p_K^{m_K}
\]
for each $(m_1,\ldots,m_K)$ such that $\sum_{i=1}^{K}m_i = n$.
\end{remark}

\begin{remark}
Let $P$ and $Q$ denote submatrices of $H_2$ such that $H_2 = \begin{bmatrix}P\\
	Q
\end{bmatrix}$, for some $(n-k_1) \times n$ matrix $P$ and $(k_1-k_2) \times n$ matrix $Q$. 
\end{remark}

\begin{figure*}[t]
	\centering
	\hspace*{11em}
	\def\svgscale{1.25}
	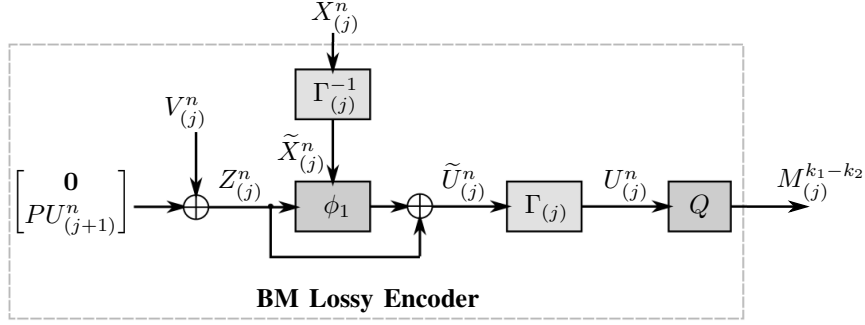
	\caption{Encoder of the block-Markov lossy source coding scheme in the $j$th coding block, for $2\leq j\leq b$.}
	\label{fig:lossy_encoder_BM}
\end{figure*}

Fig.~\ref{fig:lossy_encoder_BM} shows the encoding scheme of the lossy source code, where $V_{(j)}^{n}$ is a random dither shared with the decoder, and $\Gamma_{(j)}$ is a random permutation chosen uniformly at random and shared with the decoder. Similar to the coding scheme of Section~\ref{sec:lossy_asym}, the main idea is to use the available Lego bricks to generate a sequence $U_{(j)}^n$ with a distribution that is ``close'' to the i.i.d. $p(\hat{x})$ distribution and then convey the sequence ``losslessly'' to the decoder. 

The encoding procedure starts from the $b$th block. For each $j = b,\ldots, 2$, the encoder computes the sequences $Z_{(j)}^n$, $\wtilde{U}_{(j)}^n$, and $U_{(j)}^n$ as follows.
\begin{equation} \label{eqn:lossy_bm_zj}
	\begin{aligned}
	    Z_{(j)}^n &= \begin{bmatrix}\bf 0 \\
        PU_{(j+1)}^n\end{bmatrix} \oplus V_{(j)}^{n},\\
        \wtilde{U}_{(j)}^n &= \phi_1\left( \Gamma_{(j)}^{-1}(X_{(j)}^n), Z_{(j)}^n \right) \oplus Z_{(j)}^n,\\
        U_{(j)}^n &= \Gamma_{(j)}(\wtilde{U}_{(j)}^n),
	\end{aligned}
\end{equation}
where $PU_{(b+1)}^n = {\bf 0}^{n-k_1}$. Notice that, for each $j=2,\ldots, b$, we have
\begin{equation} \label{eqn:lossy_asym_BM_indexU}
\wtilde{H}_1\wtilde{U}_{(j)}^n = \wtilde{H}_1Z_{(j)}^n = \begin{bmatrix}\bf 0 \\
PU_{(j+1)}^n\end{bmatrix} \oplus \wtilde{H}_1V_{(j)}^{n}.
\end{equation}
The next lemma states that the distribution of $(X_{(j)}^n,U_{(j)}^n)$ in Fig.~\ref{fig:lossy_encoder_BM} is $\delta$-away in total variation distance from the i.i.d. $p(x,\hat{x})$ distribution.

\begin{lemma} \label{lemma:lossy_asym_BM}
	Let $q(x^n,u^n)$ be the true distribution of $(X_{(j)}^n,U_{(j)}^n)$. Then, we have
	\begin{equation*}
		\frac{1}{2}\sum_{x^n,u^n} \left| q(x^n,u^n) - \prod_{i=1}^n p_{X,\widehat{X}}(x_i,u_i) \right| \leq \delta,
	\end{equation*}
	for each $j=2,\ldots,b$.
\end{lemma}

\begin{proof}
    See Appendix~\ref{appendix:lossy_asym_BM}.
\end{proof}

\begin{figure*}[t]
	\centering
	\hspace*{8em}
	\def\svgscale{1.25}
	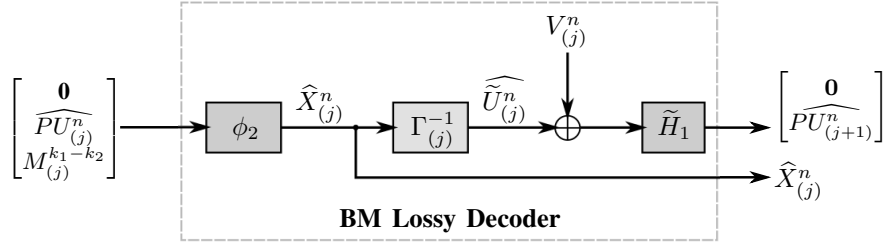
	\caption{Decoder of the block-Markov lossy source code in the $j$th coding block, for $2\leq j\leq b-1$.}
	\label{fig:lossy_decoder_BM}
\end{figure*}

Fig.~\ref{fig:lossy_decoder_BM} shows the decoding scheme of the lossy source code. The decoder uses the lossless source decoder and the estimate $\widehat{PU_{(j)}^n}$ of $PU_{(j)}^n$ which the decoder had found in previous coding blocks to compute the reconstruction sequence $\widehat{X}_{(j)}^n$. The decoder then reverse-engineers the operations at the encoder side and computes an estimate $\widehat{PU_{(j+1)}^n}$ for the subsequent coding block using the observation of~(\ref{eqn:lossy_asym_BM_indexU}), as shown in Fig.~\ref{fig:lossy_decoder_BM}. Using similar arguments as in Section~\ref{sec:lossy_asym}, it can be shown that the average distortion of the coding scheme over the $b$ blocks can be bounded as
\begin{equation}
\frac{1}{b-1}\sum_{j=2}^{b} \frac{1}{n}\E[d_H(X_{(j)}^n,\widehat{X}_{(j)}^n)] \leq D+\delta+b\eps,
\end{equation}
where the factor $b$ in front of $\eps$ pertains to the use of the union bound on the error probability over $b$ blocks. In this analysis, we assume that in the first coding block, $PU_{(2)}^n$ is passed to the decoder error-free. The rate of the coding scheme is
\[
R = \frac{(b-1)(k_1-k_2)+n-k_1}{(b-1)n}.
\]

\begin{remark}
A point-to-point channel code with a shaping joint-type spectrum that is arbitrarily close in total variation distance to a multinomial distribution exists for sufficiently large $n$ if and only if the rate $k_1/n$ of the code is larger than $1-H(\widehat{X}\cond X)$. A similar argument as in Section~\ref{sec:properties} based on distributed channel synthesis can be used to show this. Similarly, a lossless source code for a $\BERN(\alpha)$ source with arbitrarily small error probability exists for sufficiently large $n$ if and only if its rate $\frac{n-k_2}{n}$ is larger than $H(\alpha)=H(\widehat{X})$. It follows that the rate $R$ of the block-Markov lossy source coding scheme can be made arbitrarily close to $I(X;\widehat{X})$ for sufficiently large $n$ and $b$.
\end{remark}

\subsection{Other Coding Problems}
The two coding schemes presented in Sections~\ref{sec:block_markov_asym} and \ref{sec:block_markov_lossy} for the asymmetric channel coding problem and the lossy source coding problem were constructed starting from Lego bricks that satisfy some easily-verifiable properties. In a similar way, all other coding schemes presented in this paper (including the ones in the appendices) can be modified to start from such Lego bricks. For conciseness, we do not explicitly show those constructions; however, the two constructions of Sections~\ref{sec:block_markov_asym} and~\ref{sec:block_markov_lossy} already give the gist of the idea. To summarize, the following differences in comparison with the previous constructions are noted.
\begin{itemize}
    \item The assumptions made on the Lego bricks of the modified constructions involve the distribution of either the \emph{shaping Hamming distance spectrum} or the \emph{shaping joint-type spectrum} of the constituent codes, both of which can be estimated efficiently in practice for any commercial off-the-shelf code.
    \item The modified constructions do not require any constituent codebooks to be nested, unlike the previous constructions which often assumed that the codebook corresponding to one Lego brick is a subcode of another.
    \item A uniformly chosen random permutation that is shared between the encoder and the decoder is used in the modified constructions. The random permutation allows to shape the binary sequences according to the desired distribution, as shown in Lemma~\ref{lemma:asym_bm} and Lemma~\ref{lemma:lossy_asym_BM}.
    \item The modified coding schemes have a block-Markov structure, i.e., the inputs to one encoding/decoding block can depend on other coding blocks. Such a structure allows to share information about the transmitted sequences across coding blocks.
    \item The block-Markov constructions require the storage of sequences from the previous or subsequent transmission block. That is, the encoding and decoding schemes have a window size of 2.
\end{itemize}

Since the block-Markov constructions do not require any constituent codebooks to be nested, one major consequence of the constructions is the following: one can assemble \emph{any} set of point-to-point channel codes with good error correction capabilities (i.e., small probability of error) with \emph{any} set of point-to-point channel codes with good shaping properties (i.e., a shaping Hamming distance spectrum that is ``close'' to a Binomial, or a shaping joint-type spectrum that is ``close'' to a Multinomial, where ``closeness'' is measured via the total variation distance metric) to construct coding schemes for the network settings considered in this paper. That is, the block-Markov constructions imply that one can combine any set of codes that have a vanishing error probability asymptotically (e.g., first column in Table~\ref{table:codes}) with any set of codes that have a vanishing shaping distance (e.g., second column in Table~\ref{table:codes}) to construct coding schemes that can achieve the best known inner bounds for the network coding problems considered in this paper.


\section{Concluding Remarks} \label{sec:conclusion}
The Lego-brick framework developed in this paper is promising, and provides guidelines to implement practical coding schemes for network communication. The proposed code constructions can achieve larger rates in theory and demonstrate better performance in practice compared to the naive approach of coding over the point-to-point links of a network. Moreover, the computational complexity of the proposed constructions is governed by the encoding/decoding complexities of the constituent Lego bricks, which can be taken ``off-the-shelf''; hence, the constructions are friendly to hardware implementation using existing coding blocks. Finally, we point out that all the coding schemes that we presented in this paper can be extended to general (i.e., non-binary) discrete channels and sources when using non-binary linear point-to-point channel codes. For example, if a multi-user channel has an input alphabet of size $q$, similar coding schemes to the ones presented in this paper can be constructed using linear point-to-point channel codes over $\mathbb{F}_q$.

For future work, it is of particular interest to implement such a generalization to non-binary channels (including Gaussian channels with continuous inputs) starting from \emph{binary} point-to-point channel codes. In principle, this would allow one to establish a general framework for implementing coded modulation techniques for multi-user settings. Another direction that is inspired by this work is to investigate the ``shaping distance'' performance of point-to-point channel codes (compared to the more-prominently studied ``error probability'' performance). To the best of the knowledge of the authors, apart from polar codes, there are no other known \emph{low-complexity} point-to-point channel codes with a provably vanishing shaping distance over BMS channels. Ultimately, the goal would be to construct new families of codes with that property. As we saw in this paper, this can be helpful for constructing rate-optimal coding schemes in multi-user settings.

\begin{appendices}
    \section{Proof of Lemma~\ref{lemma:codify}} \label{appendix:codify}
\begin{enumerate}[label=(\roman*)]
     \item Let $x^n \in \{0,1\}^n$, and let $\wtilde{H}$ be as defined in~(\ref{eqn:H_tilde}). Then $s^n \triangleq \wtilde{H}x^n \, \in \cS$, and
    \begin{equation*}
        \begin{aligned}
            H(x^n \oplus s^n) &= Hx^n \oplus \begin{bmatrix}A & B\end{bmatrix} \begin{bmatrix}{\bf 0} & {\bf 0} \\
                 B^{-1}A & I
                 \end{bmatrix}x^n \\
            &= Hx^n \oplus Hx^n = 0,
        \end{aligned}
    \end{equation*}
     where $H = \begin{bmatrix}A & B\end{bmatrix}$ for some nonsingular matrix $B$. Therefore, $x^n \oplus s^n \in \cC$. Now, assume there exists $s_1^n, s_2^n \in \cS$, $s_1^n \neq s_2^n$, such that $x^n \oplus s_1^n \in \cC$ and $x^n \oplus s_2^n \in \cC$. Thus, $Hs_1^n = Hs_2^n$, which implies that $Bs_{1,k+1}^n = Bs_{2,k+1}^n$, and, therefore, $s_1^n=s_2^n$. A contradiction.
     
     \item Let $C^n = X^n\oplus \wtilde{H}X^n$. For any $c^n \in \cC$, we have
     \begin{equation*}
        \begin{aligned}
            &\P\{C^n = c^n\} = \sum_{s^n \in \cS}\P\{X^n\oplus \widetilde{H}X^n = c^n, \widetilde{H}X^n = s^n\} \\
            &= \hspace*{-0.1em}\sum_{s^n \in \cS}\P\{X^n \hspace*{-0.1em}= \hspace*{-0.1em}c^n\hspace*{-0.1em}\oplus\hspace*{-0.1em} s^n\} \P\{\widetilde{H}X^n \hspace*{-0.1em}=\hspace*{-0.1em} s^n \cond X^n \hspace*{-0.1em}=\hspace*{-0.1em} c^n\hspace*{-0.1em}\oplus \hspace*{-0.1em}s^n \} \\
            &\overset{(a)}{=} \sum_{s^n \in \cS} \frac{1}{2^n} \\
            &= \frac{1}{2^k},
        \end{aligned}
     \end{equation*}
     where $(a)$ follows since for any $c^n \in \cC$ and $s^n \in \cS$, $\widetilde{H}(c^n \oplus s^n) = 0^n \oplus \begin{bmatrix}
            {\bf 0} & {\bf 0} \\
            B^{-1}A & I
        \end{bmatrix}\begin{bmatrix}
            0^k \\
            s_{k+1}^n
        \end{bmatrix} = s^n$.
 
    \item Let $V^n = C^n \oplus S^n$. For any $v^n \in \{0,1\}^n$, we have
    \begin{align*}
        \P\{V^n=v^n\} &= \sum_{s^n\in \cS} \P\{V^n=v^n, S^n=s^n\}\\
        &= \sum_{s^n\in \cS} \P\{C^n=v^n \oplus s^n, S^n=s^n\}\\
        &= \frac{1}{2^{n-k}} \sum_{s^n\in \cS} \P\{C^n=v^n \oplus s^n\}\\
        &\overset{(a)}{=} \frac{1}{2^n},
    \end{align*}
    where $(a)$ follows by part (i).
\end{enumerate}

\section{Proof of Lemma~\ref{lemma:p2p_slepian}} \label{appendix:p2p_slepian}
Let $R^n = X^n \oplus V^n$. Thus, $C^n = R^n \oplus \widetilde{H}R^n$. Since $R^n$ is i.i.d. $\mathrm{Bern}(\frac{1}{2})$, it follows by Lemma~\ref{lemma:codify} that $C^n$ is uniformly distributed over $\cC$. Now, for any $y^n \in \cY^n$ and $u^n \in \{0,1\}^n$, we have
\begin{equation*}
    \begin{aligned}
        &\P\{Y^n\hspace*{-0.1em} =\hspace*{-0.1em} y^n, U^n\hspace*{-0.1em} =\hspace*{-0.1em} u^n \cond C^n\hspace*{-0.1em} =\hspace*{-0.1em} c^n\} \\
        &=\P\{Y^n \hspace*{-0.1em}=\hspace*{-0.1em} y^n, V^n \hspace*{-0.1em}\oplus \hspace*{-0.1em}\widetilde{H}V^n\hspace*{-0.1em}\oplus\hspace*{-0.1em} \widetilde{H}X^n = u^n \cond \\
        &\qquad \qquad \qquad \,\, X^n\hspace*{-0.1em} \oplus\hspace*{-0.1em} V^n\hspace*{-0.1em} \oplus \hspace*{-0.1em}\widetilde{H}X^n\hspace*{-0.1em}\oplus\hspace*{-0.1em} \widetilde{H}V^n \hspace*{-0.1em}=\hspace*{-0.1em} c^n\} \\
        &= \P\{Y^n\hspace*{-0.1em} = \hspace*{-0.1em}y^n, X^n\hspace*{-0.1em} = \hspace*{-0.1em}c^n \hspace*{-0.1em}\oplus\hspace*{-0.1em} u^n \cond X^n \hspace*{-0.1em}\oplus \hspace*{-0.1em}V^n \hspace*{-0.1em}\oplus\hspace*{-0.1em} \widetilde{H}(X^n\oplus V^n) \hspace*{-0.1em}= \hspace*{-0.1em}c^n \}\\
        &\overset{(b)}{=} \P\{Y^n \hspace*{-0.1em}= \hspace*{-0.1em}y^n, X^n \hspace*{-0.1em}= \hspace*{-0.1em}c^n \hspace*{-0.1em}\oplus\hspace*{-0.1em} u^n\} \\
        &=\prod_{i=1}^n p_{X,Y}(c_i \oplus u_i, y_i) \\
        &=\prod_{i=1}^n \bar{p}(y_i,u_i \cond c_i),
    \end{aligned}
\end{equation*}
where $(b)$ follows since $R^n = X^n \oplus V^n$ is independent of $X^n$, which implies that $C^n$ and $(X^n,Y^n)$ are independent. This completes the proof.

\vspace{0.5em}
\section{Proof of Lemma~\ref{lemma:shaping_distance}} \label{appendix:shaping_distance}
\begin{figure*}[b]
    \hrulefill
    \begin{equation} \label{eqn:lemma_shaping_distance}
        \begin{aligned}
            \delta &= \frac{1}{2}\sum_{\td{y}^n,\td{v}^n,\td{x}^n}\left|\P\{\td{Y}^n=\td{y}^n,\td{V}^n=\td{v}^n,\td{X}^n=\td{x}^n\} - \frac{1}{2^n}\prod_{i=1}^{n}\bar{p}(\td{y}_i,\td{v}_i \cond \td{x}_i)\right| \\
            &= \frac{1}{2}\sum_{\td{y}^n,\td{v}^n,\td{x}^n}\left|\P\{\td{Y}^n=\td{y}^n,\td{V}^n=\td{v}^n,\td{U}^n=\td{x}^n \oplus \td{v}^n\} - \frac{1}{2^n}\prod_{i=1}^{n}p_{X,Y}(\td{x}_i \oplus \td{v}_i,\td{y}_i)\right| \\
            &= \frac{1}{2}\sum_{\td{y}^n,\td{v}^n,\td{u}^n}\left|\P\{\td{Y}^n=\td{y}^n,\td{V}^n=\td{v}^n,\td{U}^n=\td{u}^n\} - \frac{1}{2^n}\prod_{i=1}^{n}p_{X,Y}(\td{u}_i,\td{x}_i)\right| \\
            &\geq \frac{1}{2}\sum_{\td{y}^n,\td{u}^n}\left| \sum_{\td{v}^n}\left(\P\{\td{Y}^n=\td{y}^n,\td{V}^n=\td{v}^n,\td{U}^n=\td{u}^n\} - \frac{1}{2^n}\prod_{i=1}^{n}p_{X,Y}(\td{u}_i,\td{y}_i)\right)\right| \\
            &= \frac{1}{2}\sum_{\td{y}^n,\td{u}^n}\left| \P\{\td{U}^n=\td{u}^n,\td{Y}^n=\td{y}^n\} - \prod_{i=1}^{n}p_{X,Y}(\td{u}_i,\td{y}_i)\right|
        \end{aligned}
    \end{equation}
\end{figure*}
Denote $\td{X}^n = \phi(\td{Y}^n,\td{V}^n)$. Thus, $\td{U}^n = \td{X}^n \oplus \td{V}^n$. The shaping distance $\delta$ of the code $(H,\phi)$ with respect to the channel $\bar{p}$ can be bounded as in equation~(\ref{eqn:lemma_shaping_distance}) at the bottom of this page.

\vspace{0.5em}
\section{Proof of Lemma~\ref{lemma:lossy}} \label{appendix:lossy}
Let $q(x^n,u^n,\hat{x}^n)$ be the joint distribution of $(X^n,U^n,\widehat{X}^n)$, and let $p(x^n,u^n,\hat{x}^n)$ be defined as
\[
p(x^n,u^n,\hat{x}^n) = \left(\prod_{i=1}^n p_{X,\widehat{X}}(x_i,u_i)\right) \mathbbm{1}_{\{\hat{x}^n=u^n\}}.
\]
In the block diagram of Fig.~\ref{fig:lossy_asym}, since $U^n = \phi_1(X^n,V^n) \oplus V^n$, we have that $H_1U^n = H_1V^n$, and, thus, $\begin{bmatrix}
    \bf 0 \\ 
    H_2U^n 
\end{bmatrix} = \begin{bmatrix}\bf 0 \\ H_1U^n \\ QU^n\end{bmatrix} = \begin{bmatrix}\bf 0 \\ H_1V^n \\ M^{k_1-k_2}\end{bmatrix}$,
i.e., the input to the lossless source decoder is the index corresponding to $U^n$. It follows that
\begin{align*}
    &\frac{1}{2}\sum_{x^n,\hat{x}^n}\left| q(x^n,\hat{x}^n) - p(x^n,\hat{x}^n) \right| \\
    &\leq \frac{1}{2}\sum_{x^n,u^n,\hat{x}^n}\left| q(x^n,u^n,\hat{x}^n) - p(x^n,u^n,\hat{x}^n) \right| \\
    &\leq \frac{1}{2}\sum_{x^n,u^n,\hat{x}^n} \left| q(x^n,u^n) - p(x^n,u^n)\right|q(\hat{x}^n \cond x^n,u^n) \\
    &\quad + \frac{1}{2}\sum_{x^n,u^n,\hat{x}^n} p(x^n,u^n) \left| q(\hat{x}^n \cond x^n,u^n) - \mathbbm{1}_{\{u^n=\hat{x}^n\}} \right|\\
    &= \frac{1}{2}\sum_{x^n,u^n} \left| q(x^n,u^n) - p(x^n,u^n)\right| \\
    &\quad+  \frac{1}{2}\sum_{\substack{x^n,u^n,\hat{x}^n: \\ u^n\neq \hat{x}^n}}p(x^n,u^n) q(\hat{x}^n \cond x^n,u^n) \\
    &\quad + \frac{1}{2}\sum_{x^n,u^n}p(x^n,u^n)\Big(1-q_{\widehat{X}^n \cond X^n,U^n}(u^n \cond x^n,u^n)\Big)\\
    &\overset{(a)}{\leq} \delta + \frac{1}{2}\eps+\frac{1}{2}\eps\\
    &=\delta+\eps,
\end{align*}
where $(a)$ follows by Lemma~\ref{lemma:shaping_distance} and the definition of the probability of error of a lossless source code.

\vspace{0.5em}
\section{Proof of Lemma~\ref{lemma:asym_bm}} \label{appendix:asym_bm}
Consider the block diagram shown in Fig.~\ref{fig:asym_BM}. Since $Z_{(j)}^n \oplus V_{(j)}^n$ is i.i.d $\BERN(1/2)$ for each $j=2,\ldots,b$, then the sequence $\wtilde{X}_{(j)}^n$ satisfies that
\[
\weight(\wtilde{X}_{(j)}^n) = d_H\big(\phi_2(Z_{(j)}^n \oplus V_{(j)}^n), Z_{(j)}^n \oplus V_{(j)}^n \big) \: \overset{d}{=} W_2,
\]
where $\overset{d}{=}$ denotes equality in distribution and $W_2$ is the shaping Hamming distance spectrum of the code $(H_2,\phi_2)$. It follows that, for any sequence $x^n$ with $\mathrm{wt}(x^n) = w$, we have that
\begin{equation*}
    \begin{aligned}
        &\P\{X_{(j)}^n= x^n \} \\
        &= \sum_{\widetilde{x}^n: \mathrm{wt}(\widetilde{x}^n) = w}\P\{\Gamma_{(j)}(\widetilde{x}^n) = x^n\}\P\{\widetilde{X}_{(j)}^n=\widetilde{x}^n\} \\
        &= \sum_{\widetilde{x}^n: \mathrm{wt}(\widetilde{x}^n) = w} \frac{w!(n-w)!}{n!} \P\{\widetilde{X}_{(j)}^n=\widetilde{x}^n\}\\
        &= \frac{1}{\binom{n}{w}}\P\{\weight\big(\widetilde{X}_{(j)}^n\big) = w\} = \frac{1}{\binom{n}{w}}\P\{W_2 = w\}.
    \end{aligned}
\end{equation*}
Hence,
\begin{equation*}
    \begin{aligned}
        &\frac{1}{2}\sum_{x^n} \left| \P\{X_{(j)}^n = x^n \} - \alpha^{\weight(x^n)}(1-\alpha)^{n-\weight(x^n)} \right| \\
        &= \frac{1}{2}\sum_{x^n} \left| \frac{\P\{W_2 = \mathrm{wt}(x^n)\}}{\binom{n}{\mathrm{wt}(x^n)}} - \alpha^{\mathrm{wt}(x^n)}\big(1-\alpha\big)^{n-\mathrm{wt}(x^n)}  \right| \\
        &= \frac{1}{2}\sum_{w=0}^{n} \left|\P\{W_2=w\} - \binom{n}{w}\alpha^w(1-\alpha)^{n-w} \right| \\
        &\leq \delta,
    \end{aligned}
\end{equation*}
where the last step holds since $W_2$ satisfies condition~(\ref{eqn:tv_asym_bm}).

\vspace{0.5em}
\section{Proof of Lemma~\ref{lemma:lossy_asym_BM}} \label{appendix:lossy_asym_BM}
\begin{figure*}[b]
    \hrulefill
    \begin{equation*}
        \begin{aligned}
            &\frac{1}{2}\sum_{x^n,u^n} \left| q(x^n,u^n) - \prod_{i=1}^n p_{X,\widehat{X}}(x_i,u_i)\right| = \frac{1}{2}\hspace*{-0.25em}\sum_{\substack{t = (t_{00}, t_{01}, t_{10},t_{11}): \\ \sum t_{ij} = n}} \sum_{\substack{x^n,u^n: \\ \mathrm{type}(x^n,u^n) = t}} \left| q(x^n,u^n) - \prod_{i=1}^n p_{X,\widehat{X}}(x_i,u_i)\right| \\
            &= \frac{1}{2}\hspace*{-0.25em}\sum_{\substack{t = (t_{00}, t_{01}, t_{10},t_{11}): \\ \sum t_{ij} = n}} \sum_{\substack{x^n,u^n: \\ \mathrm{type}(x^n,u^n) = t}} \left| \frac{t_{00}!t_{01}!t_{10}!t_{11}!}{n!} \P\left\{\mathrm{type}(\widetilde{X}_{(t)}^n,\widetilde{U}_{(t)}^n) = t\right\} - \prod_{i,j=0}^{1} p_{ij}^{t_{ij}}\right| \\
            &= \frac{1}{2}\hspace*{-0.25em}\sum_{\substack{t = (t_{00}, t_{01}, t_{10},t_{11}): \\ \sum t_{ij} = n}} \left| \P\left\{\mathrm{type}(\widetilde{X}_{(t)}^n,\widetilde{U}_{(t)}^n) = t\right\} - \frac{n!}{t_{00}!t_{01}!t_{10}!t_{11}!}\prod_{i,j=0}^{1} p_{ij}^{t_{ij}}\right| \\
            &\leq \delta
        \end{aligned}
    \end{equation*}
\end{figure*}
Consider the encoding scheme shown in Fig.~\ref{fig:lossy_encoder_BM}. Denote $C_{(j)}^{n} = \phi_1(\wtilde{X}_{(j)}^n, Z_{(j)}^{n})$. Since $\wtilde{X}_{(j)}^n$ is i.i.d. according to $p(x)$ and $Z_{(j)}^{n}$ is i.i.d. $\BERN(1/2)$, and $(\wtilde{X}_{(j)}^n,Z_{(j)}^{n})$ are independent, we have that
\[
\mathrm{type}\left(C_{(j)}^{n}, (\wtilde{X}_{(j)}^n, Z_{(j)}^{n})\right)
\]
satisfies condition~(\ref{eqn:tv_lossy_asym_bm}). Since $\wtilde{U}_{(j)}^n = C_{(j)}^{n} \oplus Z_{(j)}^{n}$, it follows using standard bounding techniques that
\begin{equation} \label{eqn:tv_lossy_asym_BM_proof}
\begin{aligned}
   d_{TV}\big(&\mathrm{type}(\wtilde{X}_{(j)}^n,\wtilde{U}_{(j)}^n), \\
   &\mathrm{Multinomial}(n;p_{00},p_{01},p_{10}, p_{11})\big) \leq \delta, 
\end{aligned}
\end{equation}
where $p_{ij} = p_{X,\widehat{X}}(i,j)$, for $i,j \in \{0,1\}$. 

Now, let $(x^n,u^n)$ be a pair of sequences such that $\mathrm{type}(x^n,u^n) = t$, where $t \triangleq (t_{00}, t_{01},t_{10}, t_{11})$ such that $\sum_{i,j}t_{ij} = n$. Then, we have that
\begin{equation*}
    \begin{aligned}
        &q(x^n, u^n) \\
        &= \hspace*{-1em}\sum_{\substack{(\widetilde{x}^n, \widetilde{u}^n): \\ \mathrm{type}(\widetilde{x}^n, \widetilde{u}^n) = t}}\hspace*{-1.5em}\P\{\Gamma(\widetilde{x}^n) \hspace*{-0.1em}= \hspace*{-0.1em}x^n, \Gamma(\widetilde{u}^n) \hspace*{-0.1em}= \hspace*{-0.1em}u^n\}\P\{\widetilde{X}_{(j)}^n \hspace*{-0.1em}= \hspace*{-0.1em}\widetilde{x}^n, \widetilde{U}_{(j)}^n\hspace*{-0.1em}=\hspace*{-0.1em}\widetilde{u}^n\} \\
        &= \sum_{\substack{(\widetilde{x}^n, \widetilde{u}^n): \\ \mathrm{type}(\widetilde{x}^n, \widetilde{u}^n) = t}} \frac{t_{00}!t_{01}!t_{10}!t_{11}!}{n!}\P\{\widetilde{X}_{(j)}^n = \widetilde{x}^n, \widetilde{U}_{(j)}^n=\widetilde{u}^n\} \\
        &= \frac{t_{00}!t_{01}!t_{10}!t_{11}!}{n!} \P\left\{\mathrm{type}(\widetilde{X}_{(t)}^n,\widetilde{U}_{(t)}^n) = t\right\}.
    \end{aligned}
\end{equation*}
Lemma~\ref{lemma:lossy_asym_BM} then follows from the sequence shown at the bottom of this page, where the last inequality holds from~(\ref{eqn:tv_lossy_asym_BM_proof}).

    \section{Coding for Multiple Access Channels} \label{sec:mac}
    In this appendix, we describe a coding scheme for the multiple access channel that uses two asymmetric channel codes, along with a successive cancellation decoder. The coding scheme can achieve a corner point in the optimal rate region of the multiple access channel, provided that the constituent asymmetric channel codes are rate-optimal. 

\subsection{Problem Statement}
We consider a two-user multiple access channel $p(y \cond x_1,x_2)$ with binary input alphabets $\cX_1 = \cX_2 = \{0,1\}$ and an arbitrary output alphabet $\cY$~\cite[Chapter 4]{NIT}. An $(R_1,R_2,n)$ code for this channel consists of encoders $g_j: [2^{nR_j}] \to \cX_j^n$ that map each message $M_j$ to an channel input sequence $X_j^n$ for $j=1,2$, and a decoder $\psi: \cY^n \to [2^{nR_1}] \times [2^{nR_2}]$ that assigns message estimates $(\hat{M}_1,\hat{M}_2)$ to each received sequence $Y^n$. The average probability of error of the code is $\eps = \P\left\{\{\widehat{M}_1 \neq M_1\} \cup \{\widehat{M}_2 \neq M_2\}\right\}$. It is well-known that a rate pair $(R_1,R_2)$ is achievable if
\begin{equation} \label{eqn:mac_region}
	\begin{aligned}
		R_1 &< I(X_1;Y \cond X_2),\\
		R_2 &< I(X_2;Y \cond X_1),\\
		R_1+R_2 &< I(X_1,X_2;Y),
	\end{aligned}
\end{equation}
for some pmf $p(x_1)p(x_2)$~\cite{Shannon1961,Ahlswede1971,Slepian1973b}.

\subsection{Coding Scheme} \label{sec:mac_coding_scheme}
Let $p(x_1)$ and $p(x_2)$ be some desired input pmf's. The coding scheme for the multiple access channel uses the following two Lego bricks.

\begin{lego}[\textbf{Asym} $\to$ \textbf{MAC}]
	an $(R_1,n)$ asymmetric channel code $(g_1,\psi_1)$ for the channel 
	\[
	p(y\cond x_1) = \sum_{x_2}p(x_2) p(y\cond  x_1,x_2),
	\]
	which targets an input distribution $p(x_1)$, such that, for $M_1 \sim \Unif([2^{nR_1}])$, the channel input $X_1^n = g_1(M_1)$ satisfies
	\begin{equation} \label{eqn:p2p_mac_1}
		\frac{1}{2}\sum_{x_1^n}\left| \P\{X_1^n=x_1^n\} - \prod_{i=1}^n p(x_{1i}) \right| \leq \delta_1,
	\end{equation}
	for some $\delta_1 > 0$. Let $\eps_1$ be the average probability of error of the asymmetric channel code when the channel input distribution is i.i.d. according to $p(x_1)$.
\end{lego}

\begin{lego}[\textbf{Asym} $\to$ \textbf{MAC}]
	an $(R_2,n)$ asymmetric channel code $(g_2,\psi_2)$ for the channel 
	\[
	p(y,x_1\cond x_2) = p(x_1) p(y\cond  x_1,x_2),
	\]
	which targets an input distribution $p(x_2)$, such that, for $M_2 \sim \Unif([2^{nR_2}])$, the channel input $X_2^n = g_2(M_2)$ satisfies
	\begin{equation} \label{eqn:p2p_mac_2}
		\frac{1}{2}\sum_{x_2^n}\left| \P\{X_2^n=x_2^n\} - \prod_{i=1}^n p(x_{2i}) \right| \leq \delta_2,
	\end{equation}
	for some $\delta_2 > 0$. Let $\eps_2$ be the average probability of error of the asymmetric channel code when the channel input distribution is i.i.d. according to $p(x_2)$.
\end{lego}

\begin{remark}
    Each of the asymmetric channel codes can be constructed starting from a point-to-point channel code and a Slepian--Wolf code, as described in Section~\ref{sec:asym}.
\end{remark}

Fig.~\ref{fig:mac} shows the block diagram of the coding scheme. Inspired by Slepian and Wolf's random coding scheme for the multiple access channel~\cite{Slepian1973b}, the two asymmetric channel encoders are used to target the desired joint pmf's $p(x_1)$ and $p(x_2)$ at the encoder side, and the decoder recovers estimates of the messages sequentially, as per successive cancellation decoding. It can be shown that the average probability of error of the coding scheme can be bounded as
\[
\P\left\{\{\widehat{M}_1 \neq M_1\} \cup \{\widehat{M}_2 \neq M_2\}\right\} \leq \eps_1 + \eps_2 + \delta_1 + \delta_2.
\]

\begin{figure*}[t]
	\centering
	\hspace*{5em}
	\def\svgscale{1.15}
	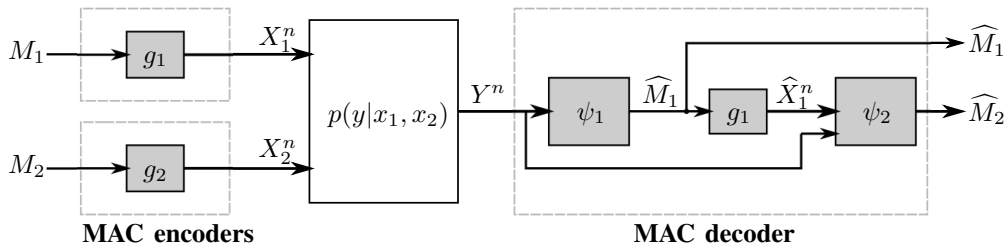
	\caption{Coding scheme for multiple access channel using two asymmetric channel codes.}
	\label{fig:mac}
\end{figure*}

\begin{remark}
	If the rate of the first asymmetric channel code is $R_1 = I(X_1;Y) - \gamma_1$ for some $\gamma_1 > 0$, and the rate of the second asymmetric channel code is   $R_2 = I(X_2;X_1,Y) - \gamma_2 = I(X_2;Y \cond X_1) - \gamma_2$ for some $\gamma_2 > 0$, then the coding scheme attains the rate pair
	\[
	(R_1,R_2) = \left(I(X_1;Y) - \gamma_1,\: I(X_2;Y \cond X_1) - \gamma_2   \right).
	\]
	Note that the rate pair $\left(I(X_1;Y),\: I(X_2;Y \cond X_1)\right)$ is a corner point of the optimal rate region~(\ref{eqn:mac_region}) of the multiple access channel. By reversing the decoding order and appropriately modifying the channels that the two asymmetric channel codes are designed for, another corner point of the rate region can be achieved.
\end{remark}

\begin{remark}
	A similar coding scheme can be constructed for the $K$-user multiple access channel using $K$ asymmetric channel codes, where the decoder successively decodes the user messages.
\end{remark}

    \section{Marton Coding for Broadcast Channels} \label{sec:marton}
    In this part, we construct a Marton coding scheme for a discrete memoryless broadcast channel starting from an asymmetric channel code and a Gelfand--Pinsker code. When the constituent codes are rate-optimal, the coding scheme can achieve a corner point in the achievable rate region of Marton coding.

\subsection{Problem Statement}
We consider a broadcast channel $p(y_1^K \cond x_1^K)$ with $K$ users and $K$ transmit antennas with a binary input alphabet $\cX=\{0,1\}$ and output alphabets $\cY_1, \ldots, \cY_K$~\cite[Chapter 8]{NIT}. An $(R_1,\ldots,R_K,n)$ code $(g,\psi_1, \ldots,\psi_K)$ for this channel consists of an encoder $g:[2^{nR_1}] \times \cdots \times [2^{nR_K}] \to \cX^{Kn}$ that maps each message tuple $(M_1,\ldots,M_K)$ to $K$ input sequences $(X_1^n,\ldots,X_K^n)$, and decoders $\psi_j: \cY_j^n \to [2^{nR_j}]$, for $j\in [K]$, that assign message estimates $\hat{M}_j$ to each received sequence $Y_j^n$. The average probability of error of the code is $\eps = \P\left\{ \widehat{M}_j \neq M_j \text{ for some } j \in [K] \right\}$. For an input distribution $p(x_1,\ldots,x_K)$, the achievable rate region of Marton's coding scheme is the set of rate tuples $(R_1,\ldots,R_K)$ such that
\[
R(\cS) < \sum_{j \in \cS}I(X_j;Y_j) - I^{*}(X_\cS),
\]
for all $\cS \subseteq [K]$, where $R(\cS) = \sum_{j \in \cS} R_j$, $X_\cS = (X_j: j\in\cS)$, and
\begin{equation} \label{eqn:mutual_info_star}
    I^{*}(X_\cS) = \sum_{j\in\cS}I(X_j;X_{[j-1]\cap \cS}).
\end{equation}
For example, for $K=2$ and an input distribution $p(x_1,x_2)$, the achievable rate region is the set of rate pairs $(R_1,R_2)$ such that
\begin{equation} \label{eqn:marton_rate}
	\begin{aligned}
		R_1 &< I(X_1;Y_1),\\
		R_2 &< I(X_2;Y_2),\\
		R_1 + R_2 &< I(X_1;Y_1) + I(X_2;Y_2) - I(X_1;X_2).
	\end{aligned}
\end{equation}

\subsection{Coding Scheme} \label{sec:marton_coding_scheme}
We consider a two-user binary-input broadcast channel $p(y_1,y_2\cond x_1,x_2)$ with two transmit antennas. The proposed Marton coding scheme targets an input distribution $p(x_1,x_2)$ and uses the following asymmetric channel code and Gelfand--Pinsker code as its constituent Lego bricks.

\begin{lego}[\textbf{Asym} $\to$ \textbf{Marton}]
an $(R_1,n)$ asymmetric channel code $(g_1,\psi_1)$ for the channel 
\[
    p(y_1\cond x_1) = \sum_{y_2,x_2}p(x_2\cond x_1) p(y_1,y_2\cond  x_1,x_2),
\]
which targets an input distribution $p(x_1)$, such that, for $M_1 \sim \Unif([2^{nR_1}])$, the channel input $X_1^n = g_1(M_1)$ satisfies
\begin{equation} \label{eqn:tv_marton_asym}
	\frac{1}{2}\sum_{x_1^n}\left| \P\{X_1^n=x_1^n\} - \prod_{i=1}^n p(x_{1i}) \right| \leq \delta_1,
\end{equation}
for some $\delta_1 > 0$. Let $\eps_1$ be the average probability of error of the asymmetric channel code when the channel input distribution is i.i.d. according to $p(x_1)$.
\end{lego}

\begin{lego}[\textbf{GP} $\to$ \textbf{Marton}]
an $(R_2,n)$ code $(g_2, \psi_2)$ for the Gelfand--Pinsker problem defined by
\[
    p(x_1)p(y_2\cond  x_2,x_1) = p(x_1)\sum_{y_1}p(y_1,y_2\cond  x_1,x_2),
\]
which targets a conditional distribution $p(x_2\cond x_1)$, such that, when $M_2 \sim \Unif([2^{nR_2}])$ and $\wtilde{X}_1^n$ is i.i.d. $p(x_1)$ sequence, the channel input $\wtilde{X}_2^n = g_2(M_2, \wtilde{X}_1^n)$ satisfies
\begin{equation} \label{eqn:tv_marton_gelfand}
	\frac{1}{2}\sum_{x_1^n,x_2^n}\left| \P\{\wtilde{X}_1^n=x_1^n,\wtilde{X}_2^n=x_2^n\} - \prod_{i=1}^n p(x_{1i},x_{2i}) \right| \leq \delta_2,
\end{equation}
for some $\delta_2 > 0$. Let $\eps_2$ be the average probability of error of the Gelfand--Pinsker code when the conditional distribution of the channel input given the channel state is i.i.d. according to $p(x_2 \cond x_1)$.
\end{lego}


\begin{figure*}[t]
	\centering
	\hspace*{7em}
	\def\svgscale{1.15}
	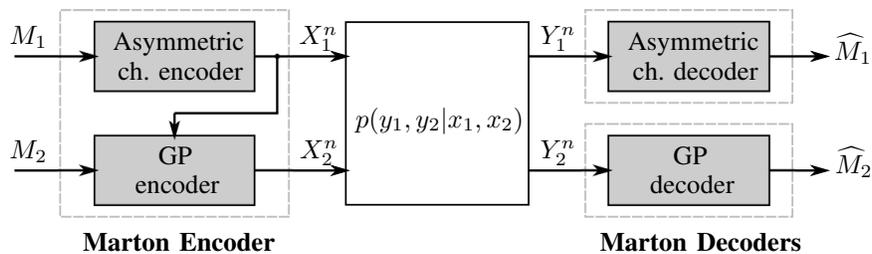
	\caption{Marton coding for the two-user broadcast channel using an asymmetric channel code and a Gelfand--Pinsker code.}
	\label{fig:marton}
\end{figure*}

Fig.~\ref{fig:marton} shows the block diagram of the Marton coding scheme for the broadcast channel. The main idea is to view the sequence that is encoded by the asymmetric channel code as a state sequence to the Gelfand--Pinsker encoder. From conditions~(\ref{eqn:tv_marton_asym}) and (\ref{eqn:tv_marton_gelfand}), it follows from standard arguments that
\[
\frac{1}{2}\sum_{x_1^n,x_2^n}\left| \P\{X_1^n=x_1^n,X_2^n=x_2^n\} - \prod_{i=1}^n p(x_{1i},x_{2i}) \right| \leq \delta_1+\delta_2,
\]
which says that the channel input distribution is $(\delta_1+\delta_2)$-away in total variation distance from the desired target distribution. The average probability of error of the coding scheme can be bounded as
\begin{equation*}
    \P\left\{ \{\widehat{M}_1 \neq M_1\} \cup \{\widehat{M}_2 \neq M_2\} \right\} \, \leq \, \delta_1 + \delta_2 + \epsilon_1 + \epsilon_2,
\end{equation*}
which follows by a similar analysis as in Section~\ref{sec:gelfand}. 


\begin{remark}
If the rate of the asymmetric channel code is $R_1 = I(X_1;Y_1) - \gamma_1$ for some $\gamma_1 > 0$, and the rate of the Gelfand--Pinsker code is $R_2 = I(X_2;Y_2) - I(X_1;X_2) -\gamma_2$ for some $\gamma_2 > 0$, it follows that the Marton coding scheme achieves the rate pair
\[
(R_1,R_2) = \left(I(X_1;Y_1)- \gamma_1, \, I(X_2;Y_2) - I(X_1;X_2)- \gamma_2\right).
\]
Note that the rate pair $\left(I(X_1;Y_1), \, I(X_2;Y_2) - I(X_1;X_2)\right)$ is a corner point of Marton's rate region for the broadcast channel. If the encoding order is reversed (i.e., $X_2^n$ is encoded using an asymmetric channel code and used as a state sequence to encode $M_1$), another corner point of the rate region can be achieved.
\end{remark}

\begin{remark}
A similar coding scheme can be constructed for the $K$-user broadcast channel with $K$ transmit antennas using one asymmetric channel code and $K-1$ Gelfand--Pinsker codes, where the encoder successively encodes the user messages.
\end{remark}

\subsection{Simulation Results} \label{sec:marton_simulation}
Now, we verify the practicality of the Marton coding scheme through simulations. We use polar codes with successive cancellation decoding as the constituent point-to-point channel codes. We consider a two-user Gaussian broadcast channel with two transmit antennas and a single receive antenna for each user. The channel can be modeled by the following input-output relations:
\begin{equation} \label{eqn:gaussian_channel}
    \begin{bmatrix}
        Y_1 \\
        Y_2
    \end{bmatrix} = H_{\mathrm{ch}}W\begin{bmatrix}
        X_1 \\
        X_2
    \end{bmatrix} + \begin{bmatrix}
        Z_1 \\
        Z_2
    \end{bmatrix},
\end{equation}
where $H_{\mathrm{ch}} = \begin{bmatrix}
	1 & g \\
	g & 1
\end{bmatrix}$ is the channel gain matrix with $g=0.9$, $W$ is a $2 \times 2$ precoding matrix used by the transmitter (if needed), $X_1 \in \{\pm 1\}$ and $X_2 \in \{\pm 1\}$ are BPSK-modulated signals corresponding to the channel input codewords (i.e., bit 0 is mapped to +1 and bit 1 is mapped to $-1$), $\begin{bmatrix} 
Z_1\\
Z_2
\end{bmatrix} \sim \mathcal{N}\big({\bf 0}, I\big)$ is a vector of independent Gaussian noise components ($I$ is the $2\times 2$ identity matrix), and $Y_1$ and $Y_2$ are the channel outputs at the users' side. The transmitter wishes to send messages to the users at an overall \emph{sum-rate} $R_{\mathrm{sum}}$, while being subject to a sum-power constraint $P$ such that
\begin{equation} \label{eqn:marton_power_constraint}
\E\left[\left\|W\begin{bmatrix}
    X_1 \\
    X_2
\end{bmatrix}\right\|^2\right] \leq P.
\end{equation}

For a given input distribution $p(x_1,x_2)$ and a precoding matrix $W$ satisfying the power constraint (\ref{eqn:marton_power_constraint}), the maximum achievable sum-rate of the Marton coding scheme can be expressed as
\begin{equation}\label{eqn:marton_sum_rate}
    C_{\mathrm{sum}}\big(p(x_1,x_2), W\big) \triangleq I(X_1;Y_1) + I(X_2;Y_2) - I(X_1;X_2).
\end{equation}
The following coding strategies for the broadcast channel will be compared.
\begin{itemize}
    \item ``Marton coding with optimal precoding''\footnote{Marton coding over the Gaussian broadcast channel is often seen as an instance of ``dirty paper coding''~\cite{Costa1983}. However, note that the channel input is binary in our case.}: corresponds to the proposed coding scheme that targets the channel input distribution $p^{*}_{\mathrm{marton,opt}}(x_1,x_2)$ and precoding matrix $W^{*}_{\mathrm{marton,\, opt}}$ that maximize $C_{\mathrm{sum}}$, while satisfying the power constraint (\ref{eqn:marton_power_constraint}).
    \item ``Marton coding without precoding'': corresponds to the proposed coding scheme that targets the channel input distribution $p^{*}_{\mathrm{marton}}(x_1,x_2)$ that maximizes $C_{\mathrm{sum}}$, while setting the precoding matrix to $W = \sqrt{\frac{P}{2}}I$.
    \item ``Symmetric coding with optimal precoding'': corresponds to the strategy that sets the precoding matrix to $W^{*}_{\mathrm{sym,opt}}$ that maximizes $C_{\mathrm{sum}}$, while setting the input distribution to the i.i.d. $\Bern(1/2)$ distribution (i.e., the two messages are encoded independently using two separate point-to-point channel codes, and the correlation between the two transmitted signals is attributed only to the linear precoder).
    \item ``Symmetric coding without precoding'': corresponds to the strategy that encodes the two messages independently using two separate point-to-point channel codes, while setting the precoding matrix to $W = \sqrt{\frac{P}{2}}I$.
    \item ``Minimum mean-square error (MMSE) precoding'': corresponds to the strategy that encodes the two messages independently using two separate point-to-point channel codes, while applying the MMSE precoding matrix $W_{\mathrm{MMSE}} = \lambda(H_{\mathrm{ch}}^TH_{\mathrm{ch}} + \frac{K}{P}I)^{-1}H_{\mathrm{ch}}^T$ at the transmitter side, where $K=2$ is the number of users and $\lambda$ is a constant to satisfy the power constraint (\ref{eqn:marton_power_constraint}). The MMSE precoding matrix is sometimes referred to as the ``transmit Wiener filter''~\cite{Joham2005}.
    \item ``Zero-forcing precoding'': corresponds to the strategy that encodes the two messages independently using two separate point-to-point channel codes, while applying the zero-forcing precoding matrix $W_{\mathrm{ZF}} = \lambda H_{\mathrm{ch}}^{-1}$ at the transmitter side, where $\lambda$ is a constant to satisfy the power constraint (\ref{eqn:marton_power_constraint}). Such a strategy suppresses the interference in the channel. 
    \item ``Time division'': corresponds to the strategy that serves only a single user of the channel. That is, a single message is encoded using a point-to-point channel code, and the same codeword is transmitted across both antennas. The power allocation between the two transmit antennas is done so as to optimize the received signal-to-noise ratio. This is done by setting the precoding matrix to $W_{\mathrm{time-division}} = \begin{bmatrix}
		\sqrt{\lambda_1} & 0\\
		0 & \sqrt{\lambda_2}
	\end{bmatrix}$, where $(\lambda_1,\lambda_2)$ is the solution of
	\[
	\begin{cases}
		\text{maximize}\qquad & (\sqrt{\lambda_1} + g\sqrt{\lambda_2})^2 \\
		\text{subject to}\qquad & \lambda_1+\lambda_2=P,\\
		& \lambda_1\geq 0, \lambda_2 \geq 0.
	\end{cases}
	\]
\end{itemize}

Before simulating the proposed Marton coding scheme, we first plot the maximum achievable sum-rate $C_{\mathrm{sum}}$ given in equation~(\ref{eqn:marton_sum_rate}) for the different coding strategies. For the coding strategies that optimize over the channel input distribution and/or the precoding matrix, we use particle swarm optimization~\cite{Kennedy1995,Luke2013} as an efficient heuristic method to perform the optimization. Fig.~\ref{fig:marton_achievable_rates} shows the plot of $C_{\mathrm{sum}}$ as a function of the sum-power constraint $P$ for the different coding strategies over the two-user Gaussian broadcast channel model given by~(\ref{eqn:gaussian_channel}). Clearly, Marton coding with optimal precoding can achieve strictly larger sum-rates asymptotically compared to the other coding strategies. When no precoding is employed, Marton coding can still achieve significantly larger sum-rates compared to common linear precoding strategies used in practice such as MMSE precoding and zero-forcing precoding. This demonstrates the significance of stochastically shaping the channel input signals over the broadcast channel and motivates our Lego-brick design of a realizable Marton coding scheme using commercial off-the-shelf codes. Such an observation has been previously noted in the literature (e.g.,~\cite{Lee2006,Goldsmith2005}).

\begin{figure}[t]
	\centering
	\includegraphics[width=\columnwidth]{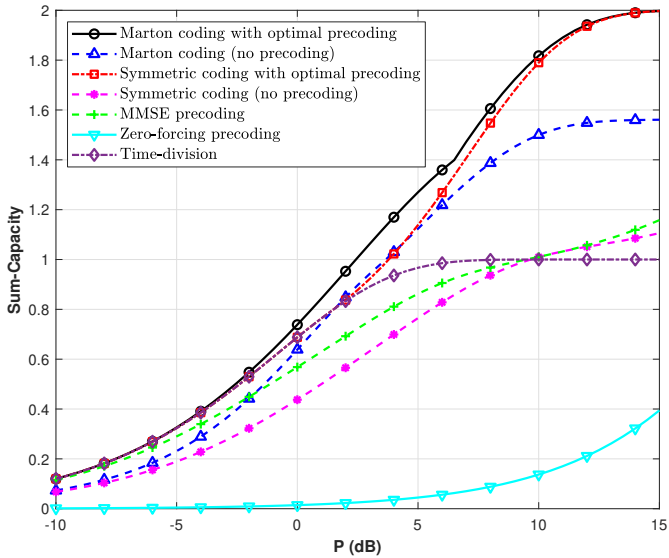}
	\caption{The maximum achievable sum-rate $C_{\mathrm{sum}}$ for the different coding strategies over the broadcast channel.}
	\label{fig:marton_achievable_rates}
\end{figure}

Next, we simulate the different coding strategies over the Gaussian broadcast channel model using polar codes with successive cancellation decoding as the constituent point-to-point channel codes. The coding strategies are compared for the same block length $n=1024$ and sum-rate $R_{\mathrm{sum}} = 1$. Recall that each of the asymmetric channel code and the Gelfand--Pinkser code can be implemented using a pair of point-to-point channel codes, as described in Section~\ref{sec:gelfand}. It follows that four polar codes are needed to implement the Marton coding scheme. Let $(R_{11},R_{12},R_{21},R_{22})$ be the rates of the polar codes, where $(R_{11},R_{12})$ are the rates of the two polar codes needed to construct the asymmetric channel code and $(R_{21},R_{22})$ are the rates of the two polar codes needed to construct the Gelfand--Pinsker code. The rates $(R_{11},R_{12},R_{21},R_{22})$ are chosen ``close'' to their theoretical limits (see Remarks~\ref{remark:gelfand_rates} and \ref{remark:asym_rates}). More precisely, for the Marton coding scheme, we set
\begin{equation} \label{eqn:marton_rates}
    \begin{aligned}
    R_{11} &= 1-H(X_1\cond Y_1)-\gamma,\\
    R_{12} &= 1-H(X_1), \\
    R_{21} &= 1-H(X_2\cond Y_2)-\gamma,\\
    R_{22} &= 1-H(X_2 \cond X_1), 
    \end{aligned}
\end{equation}
where $\gamma > 0$ is a ``back-off'' parameter from the theoretical limit, which allows for a reasonable error probability\footnote{Note that the back-off parameter is only used for $R_{11}$ and $R_{21}$ since these rates correspond to the polar codes that are used for error correction. In contrast, our simulations show that the polar codes used for shaping perform pretty well, even at rates that are extremely close to the theoretical limit; hence, no back-off is needed for $R_{12}$ and $R_{22}$.}. Note that the sum-rate attained by the coding scheme is equal to 
\[
R_{11}-R_{12}+R_{21}-R_{22} = I(X_1;Y_1)+I(X_2;Y_2)-I(X_1;X_2)-2\gamma.
\]
Hence, in order to guarantee that this sum-rate is equal to $R_{\mathrm{sum}}$, we consider the power level $P^{*}$ at which $C_{\mathrm{sum}} = R_{\mathrm{sum}} + 2\gamma$.\footnote{For each coding strategy, the power level $P^{*}$ can be estimated from the plot of Fig.~\ref{fig:marton_achievable_rates}.} For example, for the Marton coding strategy with optimal precoding, the target distribution $p^{*}_{\mathrm{marton,opt}}(x_1,x_2)$ and the precoding matrix $W^{*}_{\mathrm{marton,opt}}$ are chosen to be the maximizers of $C_{\mathrm{sum}}$, when the sum-power constraint is $P^{*}$. A similar approach is taken to find the target distribution $p^{*}_{\mathrm{marton}}(x_1,x_2)$ and the precoding matrix $W^{*}_{\mathrm{sym,opt}}$. Given the target distribution and the precoding matrix, the rates $(R_{11},R_{12},R_{21},R_{22})$ are set according to~(\ref{eqn:marton_rates}). This allows us to compare the different coding strategies for the same sum-rate $R_{\mathrm{sum}}$. In our simulations over the broadcast channel, we used $\gamma = 1/16$. For more details about the simulation setup, our code is available on GitHub~\cite{Github}.

The simulation results are shown in Fig.~\ref{fig:marton_simulation}. Clearly, the Marton coding scheme with optimal precoding can achieve improved block error rate performance compared to the other coding strategies. Even when no precoding is used, the Marton coding scheme can achieve better performance compared to common coding strategies often employed in practice, such as time division, MMSE precoding and zero-forcing precoding.

\begin{figure}[t]
	\centering
	\hspace*{-1em}
	\includegraphics[width=\columnwidth]{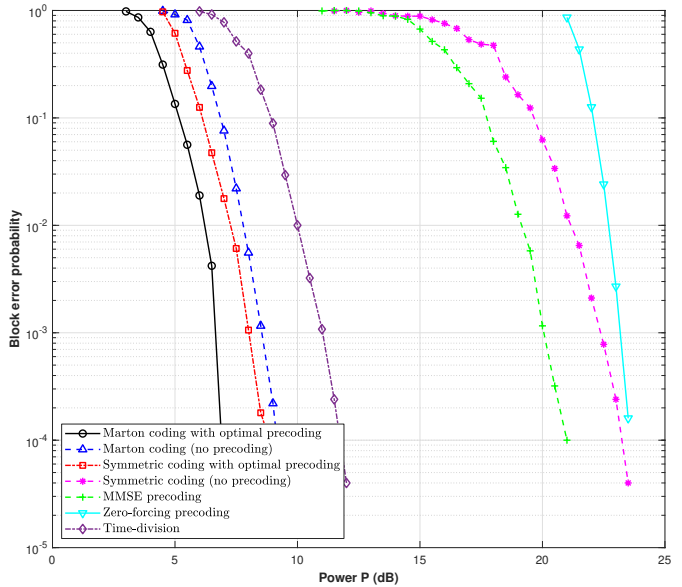}
	\caption{Simulation results for the different coding strategies over a two-user Gaussian broadcast channel for the same block length $n=1024$ and sum-rate $R_\mathrm{sum} = 1$.}
	\label{fig:marton_simulation}
\end{figure}

\begin{remark}
Note that all the coding strategies that are considered in this part can be implemented as instances of the proposed Marton coding scheme (Fig.~\ref{fig:marton}) for particular choices of the rates $(R_{11},R_{12},R_{21},R_{22})$, the target channel input distribution $p(x_1,x_2)$ and the precoding matrix $W$.
\end{remark}

    \section{Berger--Tung Coding for Distributed Lossy Compression}  \label{sec:berger}
    In this section, we describe a Berger--Tung coding scheme for distributed lossy compression starting from a lossy source code and a Wyner--Ziv code. When the constituent codes are rate-optimal, the coding scheme can achieve a corner point in the Berger--Tung rate region.

\subsection{Problem Statement}
We consider the problem of distributed lossy compression of two binary sources that generate two jointly i.i.d. binary source sequences $X_1^n$ and $X_2^n$ such that $(X_{1i},X_{2i}) \sim p(x_1,x_2)$~\cite[Chapter 12]{NIT}. The goal is to efficiently represent the two sequences using two separate encoders to a decoder that wishes to reconstruct the sequences with some distortion levels $D_1$ and $D_2$. More specifically, an $(R_1,R_2,n)$ code for the distributed lossy compression problem consists of encoders $g_j: \{0,1\}^n \to [2^{nR_j}]$ that assign an index $M_j$ to the source sequence $X_j^n$ for $j=1,2$, and a decoder $\psi: [2^{nR_1}] \times [2^{nR_2}] \to \{0,1\}^{2n}$ that assigns a pair of estimates $(\hat{X}_1^n,\hat{X}_2^n)$ to each index pair $(M_1,M_2)$. A rate-distortion quadruple $(R_1,R_2,D_1,D_2)$ is said to be achievable if there exists a sequence of $(R_1,R_2,n)$ codes such that
\[
\limsup_{n\to \infty}\tfrac{1}{n}\E[d_H(X_j^n, \widehat{X}_j^n)] \leq D_j, \qquad j=1,2,
\]
where $d_H(.,.)$ denotes the Hamming distance metric. Berger~\cite{Berger1978} and Tung~\cite{Tung1978} showed that a rate pair $(R_1,R_2)$ is achievable for the distributed lossy compression problem with distortion pair $(D_1,D_2)$ if
\begin{equation} \label{eqn:berger_region}
	\begin{aligned}
		R_1 &> I(X_1;\widehat{X}_1 \cond \widehat{X}_2),\\
		R_2 &> I(X_2;\widehat{X}_2 \cond \widehat{X}_1),\\
		R_1+R_2 &> I(X_1,X_2;\widehat{X}_1,\widehat{X}_2)
	\end{aligned}
\end{equation}
for some conditional pmf $p(\hat{x}_1 \cond x_1)p(\hat{x}_2 \cond x_2)$ such that $\E[d_H(X_j,\widehat{X}_j)] \leq D_j$, $j=1,2$.

\subsection{Coding Scheme}   \label{sec:berger_coding_scheme}
Consider two conditional pmf's $p(\hat{x}_1 \cond x_1)$ and $p(\hat{x}_2 \cond x_2)$ such that $\E[d_H(X,\widehat{X}_j)] \leq D_j$, for $j=1,2$. This completely specifies the source-reconstruction joint distribution as
\[
p(x_1,x_2,\hat{x}_1,\hat{x}_2) = p(x_1,x_2)p(\hat{x}_1 \cond x_1)p(\hat{x}_2 \cond x_2).
\]
The coding scheme for distributed lossy compression uses a lossy source code and a Wyner--Ziv code that target the desired conditional pmf's, as described in the following Lego-brick definitions.

\begin{lego}[\textbf{Lossy} $\to$ \textbf{BT}] \label{lego:lossy_berger}
	an $(R_1,n)$ lossy source code $(g_1,\psi_1)$ for a $p(x_1)$-source that targets a conditional distribution $p(\hat{x}_1 \cond x_1)$ s.t. for $X_1^n \IID p(x_1)$ and $\widehat{X}_1^n = \psi_1\big(g_1(X_1^n)\big)$, we have
	\begin{equation}\label{eqn:lossy_berger}
		\frac{1}{2}\sum_{x_1^n,\hat{x}_1^n} \left| \P\{X_1^n=x_1^n,\widehat{X}_1^n = \hat{x}_1^n\} - \prod_{i=1}^n  p(x_{1i})p(\hat{x}_{1i} \cond x_{1i})\right| \leq \delta_1
	\end{equation}
	for some $\delta_1 > 0$.
\end{lego}

\begin{lego}[\textbf{WZ} $\to$ \textbf{BT}] \label{lego:wz_berger}
	an $(R_2,n)$ Wyner--Ziv code $(g_2,\psi_2)$ for a $p(x_2,\hat{x}_1)$-source that targets a conditional distribution $p(\hat{x}_2 \cond x_2)$ s.t. for $(X_2^n,\widehat{X}_1^n) \IID p(x_2,\hat{x}_1)$ and $\widehat{X}_2^n = \psi_2\big(g_2(X_2^n), \widehat{X}_1^n\big)$, we have
	\begin{equation} \label{eqn:wz_berger}
		\frac{1}{2}\sum_{x_2^n,\hat{x}_2^n} \left| \P\{X_2^n=x_2^n,\widehat{X}_2^n = \hat{x}_2^n\} - \prod_{i=1}^n  p(x_{2i})p(\hat{x}_{2i} \cond x_{2i})\right| \leq \delta_2
	\end{equation}
	for some $\delta_2 > 0$.
\end{lego}

\begin{remark}
	Note that a lossy source code satisfying condition~(\ref{eqn:lossy_berger}) can be constructed starting from a point-to-point channel code and a lossless source code, as described in Section~\ref{sec:lossy_asym}. Also, a Wyner--Ziv code satisfying condition~(\ref{eqn:wz_berger}) can be constructed starting from a point-to-point channel code and a Slepian--Wolf code, as described in Section~\ref{sec:wyner}.
\end{remark}

Fig.~\ref{fig:berger} shows the block diagram of a distributed lossy compression code that uses the aforementioned Lego bricks. Similar to the distortion analysis in Section~\ref{sec:lossy_asym}, it can be shown using conditions~(\ref{eqn:lossy_berger}) and~(\ref{eqn:wz_berger}) that the average distortions of the Berger--Tung coding scheme can be bounded as
\begin{align*}
	\tfrac{1}{n}\E[d_H(X_1^n,\widehat{X}_1^n)] &\leq D_1 + \delta_1,\\
	\tfrac{1}{n}\E[d_H(X_2^n,\widehat{X}_2^n)] &\leq D_2  +\delta_1 + \delta_2.
\end{align*}

\begin{figure}[t]
	\centering
	\hspace*{0.5em}
	\def\svgscale{1.25}
	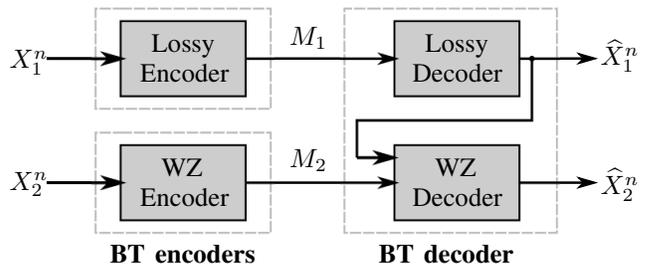
	\caption{Coding scheme for distributed lossy compression using a lossy source code and a Wyner--Ziv code.}
	\label{fig:berger}
\end{figure}

\begin{remark}
	If the rate of the lossy source code is $R_1 = I(X_1;\widehat{X}_1) + \gamma_1$ for some $\gamma_1 > 0$, and the rate of the Wyner--Ziv code is $R_2 = I(X_2;\widehat{X}_2 \cond \widehat{X}_1) + \gamma_2$ for some $\gamma_2>0$, then the coding scheme can achieve the rate pair
	\[
	(R_1,R_2) = \left(I(X_1;\widehat{X}_1) + \gamma_1, I(X_2;\widehat{X}_2 \cond \widehat{X}_1) + \gamma_2\right).
	\]
	Note that $\left(I(X_1;\widehat{X}_1), I(X_2;\widehat{X}_2 \cond \widehat{X}_1)\right)$ is a corner point of the rate-distortion region given in~(\ref{eqn:berger_region}).
\end{remark}

\begin{remark}
	A similar coding scheme can be constructed for a distributed compression problem with $L$ sources using one lossy source code and $L-1$ Wyner--Ziv codes, where the decoder successively decodes the source sequences.
\end{remark}

    \section{Multiple Description Coding} \label{sec:mdc}
    In this appendix, we describe a coding scheme for the multiple description coding problem starting from two lossy source codes and one Wyner--Ziv code. For illustrative purposes, each of the lossy source codes and Wyner--Ziv code will be described using their constituent codes, as described in previous sections. Therefore, we will construct a multiple description coding scheme starting from three point-to-point channel codes, two lossless source codes and a Slepian--Wolf code. Provided that these constituent codes are rate-optimal, the coding scheme can achieve a corner point of the El Gamal--Cover rate-distortion region for the multiple description coding problem.

\subsection{Problem Statement}
Consider a binary source sequence $X^n \IID \BERN(\theta)$ to be encoded through two descriptions such that each description by itself can be used to reconstruct the source with some distortions $D_1$ and $D_2$, and the two descriptions together can be used to reconstruct the source with a lower distortion $D_0$~\cite[Chapter 13]{NIT}, as depicted in Fig.~\ref{fig:mdc_code}. The goal is to characterize the optimal tradeoff between the description rate pair $(R_1,R_2)$ and the distortion triple $(D_0,D_1,D_2)$. An $(R_1,R_2,n)$ \emph{multiple description code} consists of an encoder $g:\{0,1\}^n \to [2^{nR_1}] \times [2^{nR_2}]$ that assigns two indices $M_1$ and $M_2$ to the source sequence $X^n$, and three decoders $\psi_j$, for $j=0,1,2$, such that $\psi_1$ assigns an estimate $\hat{X}_1^n$ to the index $M_1$, $\psi_2$ assigns an estimate $\hat{X}_2^n$ to the index $M_2$, and $\psi_0$ assigns an estimate $\hat{X}_0^n$ to the index pair $(M_1,M_2)$. A rate-distortion quintuple $(R_1,R_2,D_0,D_1,D_2)$ is \emph{achievable} if there exists a sequence of $(R_1,R_2,n)$ codes such that
\[
\limsup_{n\to \infty}\frac{1}{n}\E[d_H(X^n, \widehat{X}_j^n)] \leq D_j\qquad j=0,1,2,
\]
where $d_H(.,.)$ denotes the Hamming distance metric. El Gamal and Cover~\cite{ElGamal1982} showed that a rate pair $(R_1,R_2)$ is achievable with distortions $(D_0,D_1,D_2)$ if
\begin{equation} \label{eqn:mdc_region}
	\begin{aligned}
		R_1 &> I(X;\widehat{X}_1)\\
		R_2 &> I(X;\widehat{X}_2)\\
		R_1+R_2 &> I(X;\widehat{X}_0,\widehat{X}_1,\widehat{X}_2)+I(\widehat{X}_1;\widehat{X}_2)
	\end{aligned}
\end{equation}
for some conditional pmf $p(\hat{x}_0,\hat{x}_1,\hat{x}_2 \cond x)$ such that $\E[d_H(X,\widehat{X}_j)] \leq D_j$, $j=0,1,2$.

\begin{figure}[t] {
		\centering
		\hspace*{0.5em}
		\def\svgscale{1.3}
		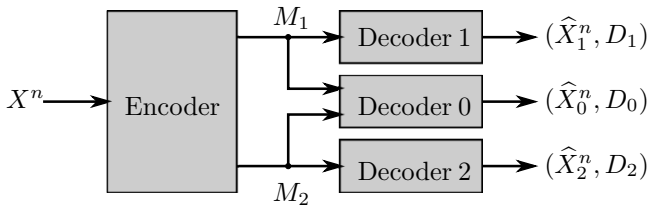
		\caption{A multiple description code.}
		\label{fig:mdc_code}
	}
\end{figure}

\subsection{Coding Scheme} \label{sec:mdc_coding_scheme}
Let $p(\hat{x}_0,\hat{x}_1,\hat{x}_2 \cond x)$ be some desired conditional pmf that satisfies $\E[d_H(X,\widehat{X}_j)] \leq D_j$, for each $j=0,1,2$. Therefore, the source-reconstruction joint distribution can be written as
\[
p(x,\hat{x}_0,\hat{x}_1,\hat{x}_2) = p(x)p(\hat{x}_0,\hat{x}_1,\hat{x}_2 \cond x).
\]
The coding scheme for the multiple description coding problem can be constructed starting from the following Lego bricks.

\begin{lego}[\textbf{P2P} $\to$ \textbf{MDC}]
	a $(k_{11},n)$ linear point-to-point channel code $(H_{11},\phi_{11})$ with codebook $\cC_{11}$ for the channel
	\begin{equation} \label{eqn:p2p_lossy_1}
		\bar{p}_1(x, v\cond\hat{x}_1) = p_{X,\widehat{X}_1}(x,\hat{x}_1\oplus v).
	\end{equation}
	Let $\delta_1$ denote the shaping distance of the code $(H_{11},\phi_{11})$ with respect to the channel $\bar{p}_1$.
\end{lego}

\begin{lego}[\textbf{Lossless} $\to$ \textbf{MDC}]
	an $(n-k_{12},n)$ lossless source code $(H_{12},\phi_{12})$ for a $\BERN(p_{\widehat{X}_1}(1))$ source with codebook $\cC_{12}$ and average probability of error $\eps_1$. We further assume that $\cC_{12} \subseteq \cC_{11}$. 
\end{lego}

\begin{lego}[\textbf{P2P} $\to$ \textbf{MDC}]
	a $(k_{21},n)$ linear point-to-point channel code $(H_{21},\phi_{21})$ with codebook $\cC_{21}$ for the channel
	\begin{equation} \label{eqn:p2p_lossy_2}
		\bar{p}_2(x, \hat{x}_1, v\cond\hat{x}_2) = p_{X,\widehat{X}_1,\widehat{X}_2}(x,\hat{x}_1, \hat{x}_2\oplus v).
	\end{equation}
	Let $\delta_2$ denote the shaping distance of the code $(H_{21},\phi_{21})$ with respect to the channel $\bar{p}_2$.
\end{lego}

\begin{lego}[\textbf{Lossless} $\to$ \textbf{MDC}]
	an $(n-k_{22},n)$ lossless source code $(H_{22},\phi_{22})$ for a $\BERN(p_{\widehat{X}_2}(1))$ source with codebook $\cC_{22}$ and average probability of error $\eps_2$. We further assume that $\cC_{22} \subseteq \cC_{21}$. 
\end{lego}

\begin{lego}[\textbf{P2P} $\to$ \textbf{MDC}]
	a $(k_{01},n)$ linear point-to-point channel code $(H_{01},\phi_{31})$ with codebook $\cC_{01}$ for the channel
	\begin{equation} \label{eqn:p2p_lossy_3}
		\bar{p}_0(x, \hat{x}_1, \hat{x}_2, v\cond\hat{x}_0) = p_{X,\widehat{X}_1,\widehat{X}_2,\widehat{X}_0}(x,\hat{x}_1, \hat{x}_2, \hat{x}_0\oplus v).
	\end{equation}
	Let $\delta_0$ denote the shaping distance of the code $(H_{01},\phi_{01})$ with respect to the channel $\bar{p}_0$.
\end{lego}

\begin{lego}[\textbf{SW} $\to$ \textbf{MDC}]
	an $(n-k_{02},n)$ Slepian--Wolf code $(H_{02},\phi_{02})$ for the problem $p\big(\hat{x}_0,(\hat{x}_1,\hat{x}_2)\big)$ with codebook $\cC_{02}$ and average probability of error $\eps_0$. We further assume that $\cC_{02} \subseteq \cC_{01}$. 
\end{lego}

\begin{remark}
	Due to nestedness, we will assume, without loss of generality, that $H_{j1}$ is a submatrix of $H_{j2}$ for each $j=0,1,2$, i.e., $H_{j2} = \begin{bmatrix}
		H_{j1}\\
		Q_j
	\end{bmatrix}$ for some $(k_{j1}-k_{j2})\times n$ matrix $Q_j$.
\end{remark}

\begin{figure}[t]
	\centering
	\def\svgscale{1.25}
	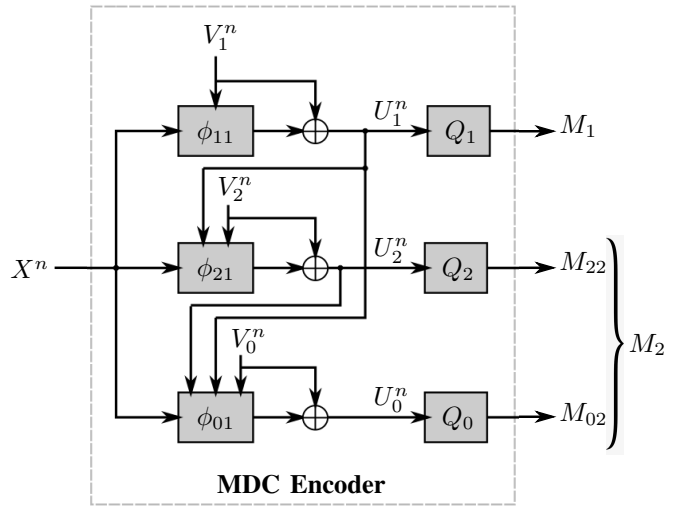
	\caption{Encoder of a multiple description code using three point-to-point channel codes.}
	\label{fig:mdc_encoder}
\end{figure}

\begin{remark}
    The point-to-point channels $\bar{p}_1$, $\bar{p}_2$ and $\bar{p}_0$ are the symmetrized channels corresponding to the joint distributions $p\big(\hat{x}_1,x\big)$, $p\big(\hat{x}_2,(\hat{x}_1,x\big))$ and $p\big(\hat{x}_0,(\hat{x}_1,\hat{x}_2,x)\big)$, respectively.
\end{remark}

Fig.~\ref{fig:mdc_encoder} and Fig.~\ref{fig:mdc_decoder} show the block diagrams of the encoder and decoder of the multiple description code, respectively, where $(V_0^n,V_1^n,V_2^n)$ are i.i.d. $\BERN(1/2)$ random dithers such that $V_1^n$ is shared with decoder 1, $V_2^n$ is shared with decoder 2, and all three random dithers are shared with decoder 0. Similar to lossy source coding, the basic idea of the coding scheme is to generate three sequences $(U_0^n,U_1^n,U_2^n)$ whose distribution is ``close'' to the i.i.d. $p(\hat{x}_0,\hat{x}_1,\hat{x}_2)$ distribution, and then perform lossless source compression to recover estimates of the sequences at the decoders. The construction of the sequences $(U_0^n,U_1^n,U_2^n)$ is done successively at the encoder side (Fig.~\ref{fig:mdc_encoder}). That is, first, the sequence $U_1^n$ is generated using the decoding function $\phi_{11}$. Then, $U_1^n$ is inputted to the decoding function $\phi_{21}$ to construct the sequence $U_2^n$. Intuitively, this step attempts to generate $U_2^n$ according to the conditional distribution $p(\hat{x}_2 \cond x, \hat{x}_1)$. Similarly, $(U_1^n,U_2^n)$ are inputted to the decoding function $\phi_{01}$ to construct the sequence $U_0^n$ according to the conditional distribution $p(\hat{x}_0 \cond x, \hat{x}_1,\hat{x}_2)$. By repeated applications of Lemma~\ref{lemma:shaping_distance} and the definition of shaping distance, it follows that
\begin{equation*}
    \begin{aligned}
        &\frac{1}{2}\sum_{u_0^n,u_1^n,u_2^n}\Big| \P\{U_0^n=u_0^n,U_1^n=u_1^n,U_2^n=u_2^n\} \\
        &\qquad \qquad \quad - \prod_{i=1}^{n}p_{\widehat{X}_0,\widehat{X}_1,\widehat{X}_2}(u_{0i},u_{1i},u_{2i})\Big| \\
        &\leq \delta_0 + \delta_1 + \delta_2.
    \end{aligned}
\end{equation*}
Moreover, for each $j=0,1,2$, we have
\[
H_{j2}U_j^n = \begin{bmatrix}
	H_{j1}\\
	Q_j
\end{bmatrix}U_j^n = \begin{bmatrix}
	H_{j1}V_j^n\\
	Q_jU_j^n
\end{bmatrix},
\]
where $Q_jU_j^n$, $j=0,1,2$, are indices transmitted to the decoders. Therefore, the decoder can recover the estimates of the source sequence knowing the shared indices and random dithers. The coding scheme can be summarized as follows.

\begin{figure}[t]
	\centering
	\hspace*{1em}
	\def\svgscale{1.15}
	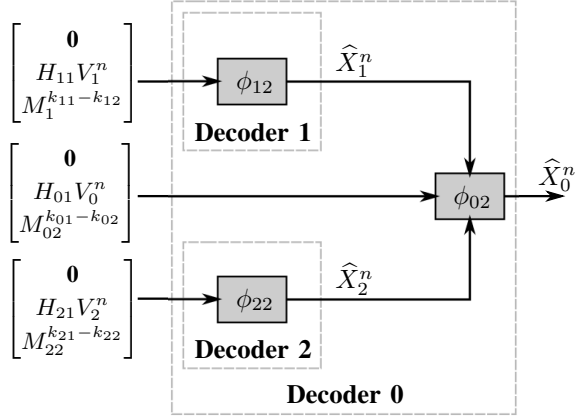
	\caption{Decoder of a multiple description code using two lossless source decoders and a Slepian--Wolf decoder.}
	\label{fig:mdc_decoder}
\end{figure}

\vspace{0.25em}
\noindent \emph{Encoding:} Upon observing the source sequence $x^n$, the encoder computes the sequence $u_1^n = \phi_{11}\left(x^n, v_1^n\right)\oplus v_1^n$ and transmits the index $m_1^{k_{11}-k_{12}} = Q_1u_1^n$ to decoder 1 and decoder 0, where $v_1^n$ is a realization of a random dither shared with the decoders. The encoder then computes the sequences $u_2^n = \phi_{21}\left(x^n, u_1^n, v_2^n\right)\oplus v_2^n$ and $u_0^n = \phi_{01}\left(x^n, u_1^n, u_2^n, v_0^n\right)\oplus v_2^n$, and transmits the index pair
\[
(m_{22}^{k_{21}-k_{22}}, m_{02}^{k_{01}-k_{02}}) = (Q_2u_2^n, Q_0u_0^n)
\]
to decoder 2 and decoder 0, where $v_2^n$ is a random dither shared with decoder 2 and decoder 0, and $v_0^n$ is a random dither shared with the decoder 0.

\vspace{0.25em}
\noindent \emph{Decoding:} Upon observing the index $m_1^{k_{11}-k_{12}}$, decoder 1 declares the sequence $\hat{x}_1^n = \phi_{12}\left( \begin{bmatrix}\bf 0 \\ H_{11}v_1^n \\ m_1^{k_{11}-k_{12}}\end{bmatrix} \right)$ as the source estimate. Upon observing the index $m_{22}^{k_{21}-k_{22}}$, decoder 2 declares the sequence $\hat{x}_2^n = \phi_{22}\left( \begin{bmatrix}\bf 0 \\ H_{21}v_2^n \\ m_{22}^{k_{21}-k_{22}}\end{bmatrix} \right)$ as the source estimate. And upon observing the index triplet $(m_1^{k_{11}-k_{12}},m_{22}^{k_{21}-k_{22}},m_{02}^{k_{01}-k_{02}})$, decoder 0 first computes the sequences $\hat{x}_1^n$ and $\hat{x}_2^n$, and declares the sequence $\hat{x}_0^n = \phi_{02}\left( \begin{bmatrix}\bf 0 \\ H_{01}v_0^n \\ m_{02}^{k_{01}-k_{02}}\end{bmatrix}, \hat{x}_1^n, \hat{x}_2^n\right)$ as the source estimate (recall that $\phi_{02}(.)$ is a Slepian--Wolf decoder with side information sequences $\hat{x}_1^n$ and $\hat{x}_2^n$).

\vspace{0.25em}
\noindent \emph{Analysis of the average distortion:} Similar to the analysis in lossy source coding, it can be shown that the average distortions can be bounded as
\begin{align*}
	\tfrac{1}{n}\E[d_H(X^n,\widehat{X}_1^n)] &\leq D_1 + \delta_1 + \eps_1\\
	\tfrac{1}{n}\E[d_H(X^n,\widehat{X}_2^n)] &\leq D_2  +\delta_1 + \delta_2 + \eps_2\\
	\tfrac{1}{n}\E[d_H(X^n,\widehat{X}_0^n)] &\leq D_0 +\delta_0 + \delta_1 + \delta_2 + \eps_0 + \eps_1 + \eps_2
\end{align*}

\vspace{0.25em}
\noindent \emph{Rate:} The coding scheme attains the rate pair $(R_1,R_2) = \left(\frac{k_{11}-k_{12}}{n}, \frac{k_{21}-k_{22}+k_{01}-k_{02}}{n}\right)$.

\begin{remark}
	Following the discussion in Section~\ref{sec:properties} and the properties of a symmetrized channel (Remark~\ref{remark:symmetrized}), it follows that if the rates of point-to-point channel codes are
	\begin{equation*}
		\begin{aligned}
			\tfrac{k_{11}}{n}&=1-H(\widehat{X}_1 \cond X)+\gamma_{11},\\
			\tfrac{k_{21}}{n}&=1-H(\widehat{X}_2 \cond X, \widehat{X}_1)+\gamma_{21},\\
			\tfrac{k_{01}}{n}&=1-H(\widehat{X}_0 \cond X, \widehat{X}_1, \widehat{X}_2)+\gamma_{01},
		\end{aligned}
	\end{equation*}
	for some $\gamma_{11}$, $\gamma_{21}$, $\gamma_{01}>0$, and the rates of the two lossless source codes and Slepian--Wolf code are, respectively,
	\begin{equation*}
		\begin{aligned}
			\tfrac{n-k_{12}}{n}&=H(\widehat{X}_1)+\gamma_{12},\\
			\tfrac{n-k_{22}}{n}&=H(\widehat{X}_2)+\gamma_{22},\\
			\tfrac{n-k_{02}}{n}&=H(\widehat{X}_0 \cond \widehat{X}_1, \widehat{X}_2)+\gamma_{02},
		\end{aligned}
	\end{equation*}
	for some $\gamma_{12}$, $\gamma_{22}$, $\gamma_{02}>0$, then the multiple description coding scheme attains the rate pair
	\begin{equation*}
		\begin{aligned}
			&(R_1,R_2)  = \left(\tfrac{k_{11}-k_{12}}{n}, \tfrac{k_{21}-k_{22}+k_{01}-k_{02}}{n}\right) \\
			&= \hspace*{-0.1em}\Big( I(X;\widehat{X}_1) \hspace*{-0.1em}+\hspace*{-0.1em} \gamma_{11}\hspace*{-0.1em}+\hspace*{-0.1em}\gamma_{12}, \: I(X,\widehat{X}_1;\widehat{X}_2) \hspace*{-0.1em}+\hspace*{-0.1em} I(X;\widehat{X}_0 \cond \widehat{X}_1,\widehat{X}_2)\\
            &\qquad\qquad\qquad\qquad\qquad\quad +\hspace*{-0.1em} \gamma_{21} \hspace*{-0.1em} + \hspace*{-0.1em} \gamma_{22} \hspace*{-0.1em} + \hspace*{-0.1em} \gamma_{01}\hspace*{-0.1em}+ \hspace*{-0.1em} \gamma_{02}\Big).
		\end{aligned}
	\end{equation*}
	Note that the rate pair $\big( I(X;\widehat{X}_1), \:I(X,\widehat{X}_1;\widehat{X}_2) + I(X;\widehat{X}_0 \cond \widehat{X}_1,\widehat{X}_2) \big)$ is a corner point of the El~Gamal--Cover rate-distortion region given in~(\ref{eqn:mdc_region}). By reversing the order of generating the sequences $U_1^n$ and $U_2^n$ at the encoder side, another corner point can be achieved.
\end{remark}

    \section{Coding for Downlink C-RAN} \label{sec:cran}
    In this appendix, we construct a coding scheme for the downlink of a cloud radio access network (C-RAN)~\cite{Simeone2016} starting from a Marton code and two lossless source codes. When the constituent codes are rate-optimal, the coding scheme can achieve a corner point of the achievable rate region of distributed decode-forward~\cite{Ganguly2021}, which is the best known inner bound for the downlink C-RAN problem.

\subsection{Problem Statement} \label{sec:cran_problem}
Consider the downlink of the C-RAN model with two users and two relays, as shown in Fig.~\ref{fig:cran_code}. A central processor (CP) communicates with the relays through noiseless fronthaul links of finite capacities $C_1$ and $C_2$. A memoryless channel $p(y_1,y_2 \cond x_1,x_2)$ is assumed between the relays and the users, with an input alphabet $\cX^2 = \{0,1\}^2$ and output alphabet $\cY_1 \times \cY_2$. An $(R_1,R_2,n)$ code for the downlink C-RAN problem consists of an encoder $g:[2^{nR_1}] \times [2^{nR_2}] \to [2^{nC_1}] \times [2^{nC_2}]$ at the CP that maps a message pair $(M_1,M_2)$ to a pair of indices $(S_1,S_2)$, encoders $h_j: [2^{nC_j}] \to \cX^n$ at the $j$th relay for $j=1,2$, that map the index $S_j$ to a channel input sequence $X_j^n$, and decoders $\psi_j: \cY_j^n \to [2^{nR_j}]$ at the $j$th user for $j=1,2$, that assign message estimates $\hat{M}_j$ to the received sequence $Y_j^n$. The average probability of error of the code is $\eps = \P\big\{ \{\widehat{M}_1 \neq M_1\} \cup \{\widehat{M}_2 \neq M_2\} \big\}$. A rate pair $(R_1, R_2)$ is said to be achievable for the downlink C-RAN problem if there exists a sequence of $(R_1,R_2,n)$ codes with vanishing error probability asymptotically.

The C-RAN architecture is a key part in the deployment of 5G standards~\cite{Johnson2019, 3gpp.38.874}. Coding schemes for the C-RAN model have been proposed in~\cite{Xiao2015,Park2014,Ganguly2021,Patil2019}; in particular,~\cite{Ganguly2021} specializes the distributed decode-forward relaying scheme~\cite{SungHoonLim2017} to the downlink C-RAN problem under consideration and shows that the achievable rate region for the two-user two-relay problem is the set of rate tuples $(R_1,R_2)$ satisfying
\begin{equation} \label{eqn:cran_rate_region}
    \sum_{k\in \cK}R_k < \sum_{k \in \cK}I(U_k;Y_k) + \sum_{l\in \cL} C_l - I(U_\cK;X_{\cL}) - I^{*}(U_\cK) - I^{*}(X_\cL),
\end{equation}
for all $\cK \subseteq \{1,2\}$ and $\cL \subseteq \{1,2\}$ and some joint distribution $p(u_1,u_2,x_1,x_2)$, where $I^{*}(.)$ is as defined in~(\ref{eqn:mutual_info_star}).

\begin{figure*}[t] {
    \centering
	\hspace*{9em}
	\def\svgscale{1.15}
	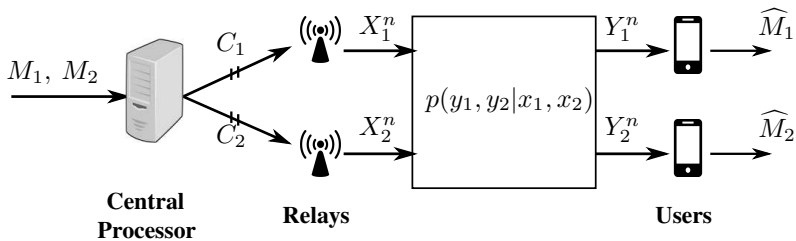
	\caption{Downlink C-RAN problem with two users, two relays and a channel $p(y_1,y_2 \cond x_1, x_2)$ between the relays and the users.}
	\label{fig:cran_code}
}
\end{figure*}

\subsection{Coding Scheme} \label{sec:cran_coding_scheme}
The coding scheme for the two-user, two-relay downlink C-RAN problem is constructed starting from a Marton code for a broadcast channel and two lossy source codes. The Marton code shapes the channel input distribution while reliably sending messages to the intended users, whereas the lossy source codes compress the encoded signals to accommodate the limited capacity constraints of the fronthaul links. For clarity of representation, we will describe the Marton code by its constituent asymmetric channel code and Gelfand--Pinsker code, as described in Appendix~\ref{sec:marton}. Similarly, each lossy source code will be implemented using its constituent point-to-point channel code and lossless source code, as described in Section~\ref{sec:lossy_asym}. Therefore, the following Lego bricks (one asymmetric channel code, one Gelfand--Pinsker code, two point-to-point channel codes and two lossless source codes) will be used in the construction, which will target some desired joint distribution $p(u_1,u_2,x_1,x_2)$.

\begin{lego}[\textbf{Asym} $\to$ \textbf{DL C-RAN}]
an $(R_1,n)$ asymmetric channel code $(f_1,\psi_1)$ for the channel 
\[
    p(y_1|u_1) = \sum_{\substack{u_2,x_1,\\x_2,y_2}}p(u_2,x_1,x_2|u_1) p(y_1,y_2| x_1,x_2),
\]
which targets an input distribution $p(u_1)$, such that, for $M_1 \sim \Unif([2^{nR_1}])$, the sequence $U_1^n = g_1(M_1)$ satisfies
\begin{equation} \label{eqn:tv_cran_asym}
	\frac{1}{2}\sum_{u_1^n}\left| \P\{U_1^n=u_1^n\} - \prod_{i=1}^n p(u_{1i}) \right| \leq \delta_1,
\end{equation}
for some $\delta_1 > 0$. Let $\eps_1$ be the average probability of error of the asymmetric channel code when the channel input distribution is i.i.d. according to $p(u_1)$.
\end{lego}

\begin{lego}[\textbf{GP} $\to$ \textbf{DL C-RAN}]
an $(R_2,n)$ code $(f_2, \psi_2)$ for the Gelfand--Pinsker problem defined by
\[
    p(u_1)p(y_2| u_2,u_1) = p(u_1)\sum_{x_1,x_2,y_1}p(x_1,x_2|u_1,u_2)p(y_1,y_2| x_1,x_2),
\]
which targets a conditional distribution $p(u_2 \cond u_1)$, such that, when $M_2 \sim \Unif([2^{nR_2}])$ and $\wtilde{U}_1^n$ is an i.i.d. $p(u_1)$ sequence, the sequence $\wtilde{U}_2^n = g_2(M_2, \wtilde{U}_1^n)$ satisfies
\begin{equation} \label{eqn:tv_cran_gelfand}
	\frac{1}{2}\sum_{u_1^n,u_2^n}\left| \P\{\wtilde{U}_1^n=u_1^n,\wtilde{U}_2^n=u_2^n\} - \prod_{i=1}^n p(u_{1i},u_{2i}) \right| \leq \delta_2,
\end{equation}
for some $\delta_2 > 0$. Let $\eps_2$ be the average probability of error of the Gelfand--Pinsker code when the conditional distribution of the channel input given the channel state is i.i.d. according to $p(u_2 \cond u_1)$.
\end{lego}

\begin{lego}[\textbf{P2P} $\to$ \textbf{Lossy} $\to$ \textbf{DL C-RAN}]
a $(k_{31},n)$ linear point-to-point channel code $(H_{31},\phi_{31})$ with codebook $\cC_{31}$ for the channel
\begin{equation} \label{eqn:p2p_lossy_cran1}
\bar{p}_3(u_1, u_2, v\cond x_1) = p_{U_1,U_2,X_1}(u_1,u_2,x_1\oplus v).
\end{equation}
Let $\delta_3$ denote the shaping distance of the code $(H_{31},\phi_{31})$ with respect to the channel $\bar{p}_3$.
\end{lego}

\begin{lego}[\textbf{Lossless} $\to$ \textbf{Lossy} $\to$ \textbf{DL C-RAN}]
an $(n-k_{32},n)$ lossless source code $(H_{32},\phi_{32})$ for a $\BERN(p_{X_1}(1))$ source with codebook $\cC_{32}$ and average probability of error $\eps_3$. We further assume that $\cC_{32} \subseteq \cC_{31}$ (i.e., the two codes are nested) and that $k_{31}-k_{32} < nC_1$. 
\end{lego}

\begin{lego}[\textbf{P2P} $\to$ \textbf{Lossy} $\to$ \textbf{DL C-RAN}]
a $(k_{41},n)$ linear point-to-point channel code $(H_{41},\phi_{41})$ with codebook $\cC_{41}$ for the channel
\begin{equation} \label{eqn:p2p_lossy_cran2}
\bar{p}_4(u_1, u_2, x_1, v\cond x_2) = p_{U_1,U_2,X_1,X_2}(u_1,u_2,x_1,x_2\oplus v).
\end{equation}
Let $\delta_4$ denote the shaping distance of the code $(H_{41},\phi_{41})$ with respect to the channel $\bar{p}_4$.
\end{lego}

\begin{lego}[\textbf{Lossless} $\to$ \textbf{Lossy} $\to$ \textbf{DL C-RAN}]
an $(n-k_{42},n)$ lossless source code $(H_{42},\phi_{42})$ for a $\BERN(p_{X_2}(1))$ source with codebook $\cC_{42}$ and average probability of error $\eps_4$. We further assume that $\cC_{42} \subseteq \cC_{41}$ (i.e., the two codes are nested) and that $k_{41}-k_{42} < nC_2$.
\end{lego}

\begin{remark}
The channels $\bar{p}_3$ and $\bar{p}_4$ are the symmetrized channels corresponding to the joint distributions $p\big(x_1,(u_1,u_2)\big)$ and $p\big(x_2,(u_1,u_2,x_1)\big)$, respectively.
\end{remark}

\begin{remark}
Without loss of generality, assume that $H_{j1}$ is a submatrix of $H_{j2}$ for $j=3,4$, i.e., $H_{j2} = \begin{bmatrix}
	H_{j1}\\
	Q_j
\end{bmatrix}$ for some $(k_{j1}-k_{j2})\times n$ matrix $Q_j$. 
\end{remark}

\begin{figure}[t]
	\centering
	\hspace*{0.5em}
	\def\svgscale{1.15}
	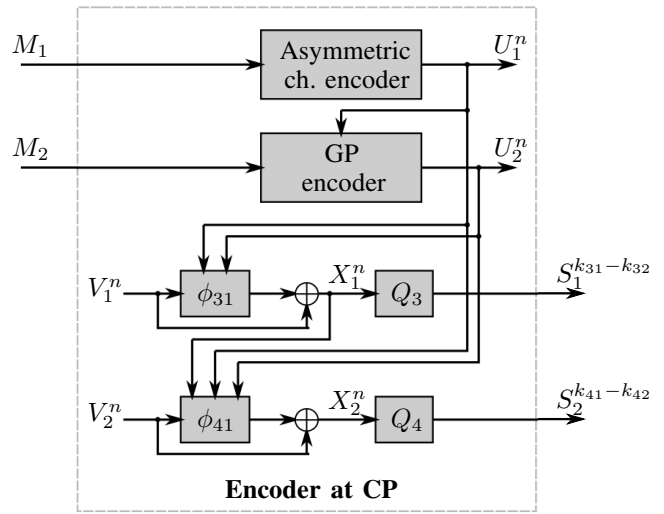
	\caption{Encoding scheme at the central processor for the two-user, two-relay downlink C-RAN problem using an asymmetric channel code, a Gelfand--Pinsker code and two point-to-point channel codes.}
	\label{fig:cran_cp}
\end{figure}

Fig.~\ref{fig:cran_cp} and Fig.~\ref{fig:cran_relay} show the block diagrams of the coding scheme at the central processor, relays and users. The key point at the central processor is to construct a tuple $(U_1^n,U_2^n,X_1^n,X_2^n)$ that is close in total variation distance to the i.i.d. distribution according to $p(u_1,u_2,x_1,x_2)$, and then compress $(X_1^n,X_2^n)$ through a pair of indices that are sent to the relays through the fronthaul links. The relays recover estimates of $(X_1^n,X_2^n)$, which are transmitted to the users through the channel. Using standard arguments that were used throughout this paper, it follows from conditions~(\ref{eqn:tv_cran_asym}), (\ref{eqn:tv_cran_gelfand}), and Lemma~\ref{lemma:shaping_distance} that
\begin{equation*}
    \begin{aligned}
        &\frac{1}{2}\sum_{\substack{u_1^n,u_2^n,\\ x_1^n,x_2^n}}\Big| \P\{U_1^n=u_1^n, U_2^n=u_2^n,X_1^n=x_1^n,X_2^n=x_2^n\} \\
        &\qquad \qquad - \prod_{i=1}^n p(u_{1i},u_{2i},x_{1i},x_{2i}) \Big| \\
        &\leq \delta_1+\delta_2+\delta_3+\delta_4,
    \end{aligned}
\end{equation*}
which says that the distribution of $(U_1^n,U_2^n,X_1^n,X_2^n)$ is $(\delta_1+\delta_2+\delta_3+\delta_4)$-away in total variation distance from the i.i.d. $p(u_1,u_2,x_1,x_2)$ distribution. Furthermore, the average probability of error of the coding scheme can be bounded as
\begin{equation*}
    \begin{aligned}
        &\P\left\{ \{\widehat{M}_1 \neq M_1\} \cup \{\widehat{M}_2 \neq M_2\} \right\} \\
        &\leq \, \delta_1 + \delta_2 + \delta_3 + \delta_4 + \epsilon_1 + \epsilon_2 + \eps_3 + \eps_4,
    \end{aligned}
\end{equation*}
following similar bounding techniques as in previous sections. 

The coding scheme attains the rate pair $(R_1,R_2)$, and the compression rates are $\frac{k_{31}-k_{32}}{n}$ and $\frac{k_{41}-k_{42}}{n}$. Notice that these compression rates do not exceed the fronthaul link capacities (by assumption).

\begin{remark} \label{remark:cran_achievable_rates}
Similar to the analysis of the Marton coding scheme of Appendix~\ref{sec:marton}, the coding rate pair $(R_1,R_2)$ can be made arbitrarily close to $(I(U_1;Y_1), \, I(U_2;Y_2) - I(U_2; U_1))$ for sufficiently large $n$. On the other hand, following the discussion on the achievable rates of the lossy source coding scheme (specifically, Remarks~\ref{remark:lossy_asym} and \ref{remark:lossy_asym_achievability}), the compression rate pair $\left(\frac{k_{31}-k_{32}}{n}, \frac{k_{41}-k_{42}}{n}\right)$ can be made arbitrarily close to $\big(I(U_1,U_2;X_1), I(U_1,U_2,X_1;X_2)\big)$ for sufficiently large $n$. It can be checked that this corresponds to a corner point in the achievable rate region of the distributed decode-forward relaying scheme given in~(\ref{eqn:cran_rate_region}). If the encoding order at the central processor is modified (e.g., if $X_1^n$ is generated before $U_2^n$ in Fig.~\ref{fig:cran_cp}), other corner points of the rate region can be achieved.
\end{remark}

\begin{remark}
A similar approach can be taken to construct a coding scheme for the general $K$-user, $L$-relay downlink C-RAN problem using one asymmetric channel code, $K-1$ Gelfand--Pinsker codes, $L$ point-to-point channel codes and $L$ lossless source codes.
\end{remark}

\begin{figure*}[t]
	\centering
	\hspace*{7em}
	\def\svgscale{1.15}
	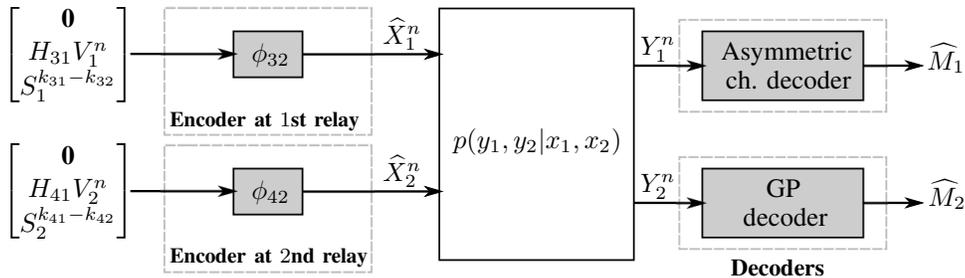
	\caption{Encoding scheme at the relays and decoding scheme at the users for the downlink C-RAN problem.}
	\label{fig:cran_relay}
\end{figure*}

\subsection{Simulation Results} \label{sec:cran_simulation}
By implementing each of the lossless source codes, the asymmetric channel code and the Gelfand--Pinsker code using their constituent point-to-point channel codes as described in Section~\ref{sec:lossless} and Section~\ref{sec:gelfand}, the coding scheme for downlink C-RAN can be simulated using point-to-point channel codes. We consider a two-user two-relay model with a Gaussian channel between the relays and the users, i.e., the channel output can be expressed as
\begin{equation} \label{eqn:gaussian_channel_cran}
    \begin{bmatrix}
        Y_1 \\
        Y_2
    \end{bmatrix} = H_{\mathrm{ch}}\Lambda \begin{bmatrix}
        X_1\\
        X_2
    \end{bmatrix} + \begin{bmatrix}
        Z_1\\
        Z_2
    \end{bmatrix}
\end{equation}
where $H_{\mathrm{ch}} = \begin{bmatrix}
	1 & g \\
	g & 1
\end{bmatrix}$ is the channel gain matrix with $g=0.9$, $\Lambda = \begin{bmatrix}
\sqrt{\lambda_1} & 0\\
0 & \sqrt{\lambda_2} 
\end{bmatrix}$ is a power allocation matrix for the relays\footnote{Note that since the two relays do not communicate and the fronthaul links are capacity-limited, we do not consider the possibility of applying a general precoding matrix as in Appendix~\ref{sec:marton_simulation}.}, $X_1 \in \{\pm 1\}$ and $X_2 \in \{\pm 1\}$ are BPSK-modulated signals, and $\begin{bmatrix}
    Z_1\\
    Z_2
\end{bmatrix} \sim \mathcal{N}\big({\bf 0}, I\big)$ is a vector of independent Gaussian noise components. The transmitter wishes to send messages to the users at an overall sum-rate $R_{\mathrm{sum}}$, and each relay is subject to a power constraint $P$ such that $\lambda_j \leq P$ for each $j=1,2$.

We first look at the maximum achievable sum-rate of the proposed coding scheme under different power constraints $P$ and fronthaul capacity constraints $C_1$ and $C_2$. We know from (\ref{eqn:cran_rate_region}) that the maximum sum-rate $C_{\mathrm{sum}}$ that can be achieved by our coding scheme is given by equation~(\ref{eqn:cran_optimization_ddf}) below,
\begin{figure*}[b]
    \hrulefill
    \begin{equation} \label{eqn:cran_optimization_ddf}
        C_{\mathrm{sum}} \triangleq \underset{\substack{p(u_1,u_2,x_1,x_2), \\ (\lambda_1, \lambda_2) \in \cK(P)}}{\max} \, \left(I(U_1;Y_1) + I(U_2;Y_2) - I(U_1;U_2) + \min\left\{\begin{matrix}
            0,\\
            C_1-I(U_1,U_2;X_1),\\
            C_2-I(U_1,U_2;X_2),\\
            C_1+C_2-I(U_1,U_2;X_1,X_2) - I(X_1;X_2),\\
        \end{matrix} \right\}\right),
    \end{equation}
\end{figure*}
where $\cK(P) = \{(\lambda_1,\lambda_2): 0\leq \lambda_1 \leq P, 0 \leq \lambda_2 \leq P\}$. For the encoding order adopted in Fig.~\ref{fig:cran_cp} and the corresponding achievable rates (see Remark~\ref{remark:cran_achievable_rates}), the maximum achievable sum-rate $C_{\mathrm{sum}}$ can be alternatively expressed as the solution to a simpler optimization problem, i.e.,
\begin{equation} \label{eqn:cran_optimization}
    C_{\mathrm{sum}} = \hspace*{-2.5em}\underset{\substack{p(u_1,u_2,x_1,x_2) \in \cD(C_1,C_2), \\ (\lambda_1, \lambda_2) \in \cK(P)}}{\max} \, \big(I(U_1;Y_1) + I(U_2;Y_2) - I(U_1;U_2)\big),
\end{equation}
where $\cD(C_1,C_2)$ is the set of joint distributions $p(u_1,u_2,x_1,x_2)$ that satisfy $C_1 > I(U_1,U_2; X_1)$ and $C_2 > I(U_1,U_2,X_1; X_2)$.
We use the genetic algorithm~\cite{Holland1992, Luke2013} as an efficient heuristic method to perform the optimization in (\ref{eqn:cran_optimization}). Fig.~\ref{fig:cran_achievable_rates} shows the plot of $C_{\mathrm{sum}}$ as a function of the power constraint $P$ under different fronthaul capacity constraints $C_1$ and $C_2$. As expected, larger sum-rates can be achieved by the proposed coding scheme when the fronthaul link capacities are increased.\footnote{Note that when $C_1 \geq 1$ and $C_2\geq 1$, the maximum achievable sum-rate of the downlink C-RAN coding scheme over the given Gaussian model is the same as that of the Marton coding scheme over the Gaussian broadcast channel model considered in Appendix~\ref{sec:marton_simulation}, in the special case when the transmitter is subject to a per-antenna power constraint and the precoding matrix is restricted to be a diagonal matrix.}

Next, the proposed coding scheme for the downlink C-RAN problem is simulated for a block length $n=1024$ and a sum-rate $R_{\mathrm{sum}} = 0.75$ using polar codes with successive cancellation decoding as the constituent point-to-point channel codes. Note that each of the asymmetric channel code and the Gelfand--Pinsker code can be implemented using a pair of point-to-point channel codes, and each of the lossless source codes can be implemented using a single point-to-point channel code. Hence, the downlink C-RAN coding scheme can be constructed starting from eight polar codes. Let $(R_{11},R_{12},R_{21},R_{22}, R_{31}, R_{32}, R_{41}, R_{42})$ be the rates of the constituent polar codes, where $(R_{11},R_{12})$ are the rates of the two polar codes needed to construct the asymmetric channel code, $(R_{21},R_{22})$ are the rates of the two polar codes needed to construct the Gelfand--Pinsker code, $(R_{31}, R_{32})$ are the rates of the two polar codes needed to construct the codes $\cC_{31}$ and $\cC_{32}$, and $(R_{41}, R_{42})$ are the rates of the two polar codes needed to construct the codes $\cC_{41}$ and $\cC_{42}$. The rates are chosen ``close'' to their theoretical limits, i.e., we set
\vspace*{-2.25em}
\begin{multicols}{2}
  \begin{equation*}
      \begin{aligned}
        R_{11} &= 1-H(U_1\cond Y_1)-\gamma_r,\\
        R_{21} &= 1-H(U_2\cond Y_2)-\gamma_r,\\
        R_{31} &= 1-H(X_1 \cond U_1,U_2),\\
        R_{41} &= 1-H(X_2 \cond U_1,U_2,X_1),
      \end{aligned}
  \end{equation*} \break
  \begin{equation} \label{eqn:cran_rates}
    \begin{aligned}
        R_{12} &= 1-H(U_1), \\
        R_{22} &= 1-H(U_2 \cond U_1), \\
        R_{32} &= 1-H(X_1)-\gamma_c,\\
        R_{42} &= 1-H(X_2)-\gamma_c,
    \end{aligned}
  \end{equation}
\end{multicols}
\vspace*{-0.5em}
\noindent where $\gamma_r > 0$ and $\gamma_c>0$ are ``back-off'' parameters from the theoretical limit that are used for the polar codes involved in error correction (i.e., not shaping). Thus, the sum-rate attained by the coding scheme is equal to 
\[
R_{11}-R_{12}+R_{21}-R_{22} \hspace*{-0.1em}=\hspace*{-0.1em} I(U_1;Y_1) + I(U_2;Y_2) - I(U_1;U_2) -2\gamma_r.
\]
As in Appendix~\ref{sec:marton_simulation}, in order to guarantee that the sum-rate is equal to $R_{\mathrm{sum}}$, the coding scheme targets the joint distribution $p^{*}(u_1,u_2,x_1,x_2)$ and the power levels $(\lambda_1^{*},\lambda_2^{*})$ that maximize $C_{\mathrm{sum}}$ under the optimization problem of (\ref{eqn:cran_optimization}), where the power constraint $P$ is set to be equal to the power level $P^{*}$ at which $C_{\mathrm{sum}} = R_{\mathrm{sum}} + 2\gamma_r$.\footnote{The power level $P^{*}$ can be estimated using the plot of Fig.~\ref{fig:cran_achievable_rates}.} Under the joint distribution $p^{*}(u_1,u_2,x_1,x_2)$ and the power levels $(\lambda_1^{*},\lambda_2^{*})$, the rates of the constituent polar codes are set according to~(\ref{eqn:cran_rates}). Note that $\gamma_c$ should be chosen so that $R_{31}-R_{32} \leq C_1$ and $R_{41}-R_{42} \leq C_2$. In our simulations over the downlink C-RAN channel model, we used $\gamma_r = 1/8$ and $\gamma_c = 5/32$. For more details about the simulation setup, our code is available on GitHub~\cite{Github}. The block error rate performance of the downlink C-RAN coding scheme is shown in Fig.~\ref{fig:cran_simulation} for different fronthaul capacity constraints $C_1$ and $C_2$. The results demonstrate the practicality of the proposed coding scheme over the downlink C-RAN channel model.

\begin{figure}[t]
	\centering
	\includegraphics[width=\columnwidth]{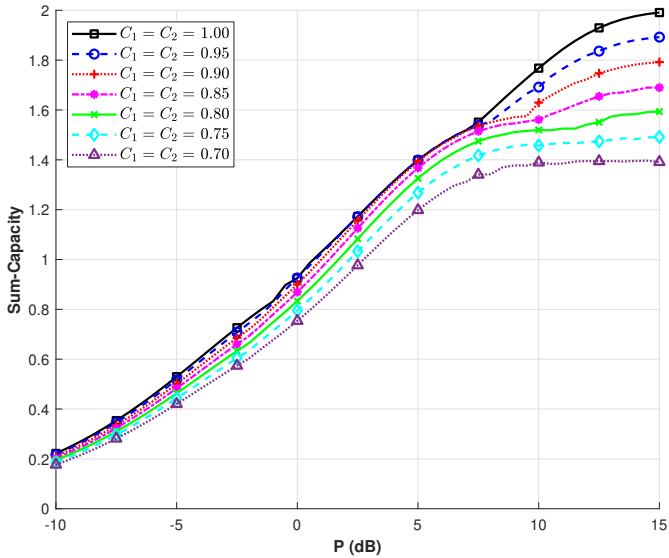}
	\caption{The maximum achievable sum-rate of the proposed coding scheme for the downlink C-RAN problem under different fronthaul capacity constraints.}
	\label{fig:cran_achievable_rates}
\end{figure}

\begin{figure}[t]
	\centering
	\hspace*{-1em}
	\includegraphics[width=\columnwidth]{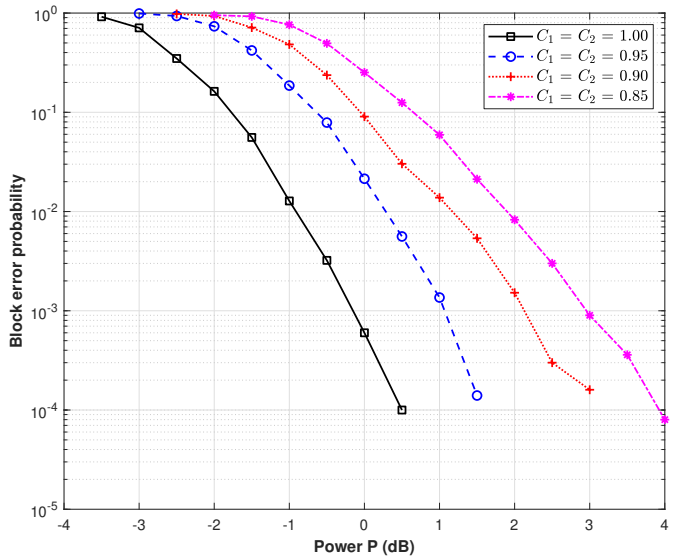}
	\caption{Simulation results of the downlink C-RAN coding scheme for a block length $n=1024$ and sum-rate $R_\mathrm{sum} = 0.75$ under different fronthaul capacity constraints.}
	\label{fig:cran_simulation}
\end{figure}

    \section{Coding for Uplink C-RAN} \label{sec:uplink-cran}
    In this appendix, we describe a coding scheme for the uplink C-RAN model starting from a multiple access channel code and a Berger--Tung code. We show that, if the constituent codes are rate-optimal, the coding scheme can achieve a corner point of the rate region achieved by the network compress-and-forward coding scheme with successive decoding~\cite{Zhou2016}.

\subsection{Problem Statement}
Consider the uplink of a C-RAN model with two users and two relays~\cite{Simeone2016}, as shown in Fig.~\ref{fig:uplink_cran_code}. Two users wish to communicate with a central processor (CP) through two relays that are connected to the CP through noiseless backhaul links of finite capacities $C_1$ and $C_2$. A memoryless channel $p(y_1,y_2 \cond x_1,x_2)$ is assumed between the users and the relays, with an input alphabet $\cX^2 = \{0,1\}^2$ and output alphabet $\cY_1 \times \cY_2$. An $(R_1,R_2,n)$ code for uplink C-RAN problem consists of encoders $g_j:[2^{nR_j}] \to \cX^n$ at the $j$th user for $j=1,2$ that map each message $M_j$ to a channel input sequence $X_j^n$, encoders $h_j: \cY_j^n \to [2^{nC_j}]$ at the $j$th relay for $j=1,2$ that map the received sequence $Y_j^n$ to an index $S_j$, and a decoder $\psi: [2^{nC_1}] \times [2^{nC_2}] \to [2^{nR_1}] \times [2^{nR_2}]$ at the CP that assign message estimates $(\hat{M}_1,\hat{M}_2)$ to the index pair $(S_1,S_2)$. The average probability of error of the code and achievability of a rate pair $(R_1, R_2)$ are defined as in Appendix~\ref{sec:cran}.

\begin{figure*}[t] {
		\centering
		\hspace*{10em}
		\def\svgscale{1.2}
		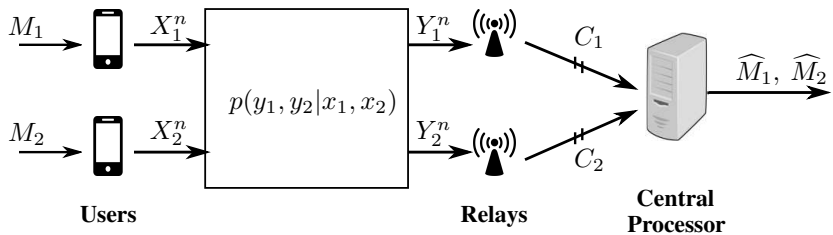
		\caption{Uplink C-RAN problem with two users and two relays.}
		\label{fig:uplink_cran_code}
	}
\end{figure*}

Coding schemes for the uplink C-RAN problem have been proposed in~\cite{Zhou2014, Sanderovich2009, Ganguly2021} based on the network compress-and-forward relaying scheme~\cite{Kramer2005}, the achievable rate region of which can be expressed as the rate tuples $(R_1,R_2)$ satisfying
\begin{equation} \label{eqn:uplink_cran_rate_region}
    \sum_{k\in \cK}R_k < \sum_{l\in \cL} \big(C_l - I(Y_l;\widehat{Y}_l \cond X_1,X_2)\big) + I(X_\cK;\widehat{Y}_{\cL^c} \cond X_{\cK^c}),
\end{equation}
for all $\cK \subseteq \{1,2\}$ and $\cL \subseteq \{1,2\}$ and some product distribution $p(x_1)p(x_2)p(\hat{y}_1 \cond y_1)p(\hat{y}_2 \cond y_2)$.

\subsection{Coding Scheme} \label{sec:cran_uplink_coding_scheme}
The coding scheme for the uplink C-RAN problem is constructed starting from a multiple access channel code and a Berger--Tung code. The coding scheme will target a given product distribution $p(x_1)p(x_2) p(\hat{y}_1 \cond y_1)p(\hat{y}_2 \cond y_2)$\footnote{For example, such a product distribution can be chosen to maximize the maximum achievable sum-rate $I(X_1,X_2;\widehat{Y}_1,\widehat{Y}_2)$.}. Therefore, the desired input-output joint distribution can be given by
\begin{equation*}
    \begin{aligned}
        &p(x_1,x_2,y_1,y_2,\hat{y}_1,\hat{y}_2) \\
        &= p(x_1)p(x_2) p(y_1,y_2 \cond x_1,x_2)p(\hat{y}_1 \cond y_1)p(\hat{y}_2 \cond y_2).
    \end{aligned}
\end{equation*}
The coding scheme for the uplink C-RAN problem uses the following Lego bricks.

\begin{lego}[\textbf{MAC} $\to$ \textbf{UL C-RAN}]
	an $(R_1,R_2,n)$ code $(f_1,f_2,\phi)$ for the multiple access channel $p(\hat{y}_1,\hat{y}_2 \cond x_1,x_2)$ which targets input distributions $p(x_1)p(x_2)$, such that, for $M_j \sim \Unif([2^{nR_j}])$, the channel input $X_j^n = f_j(M_j)$ satisfies
	\begin{equation} \label{eqn:asym_uplink_cran_1}
		\frac{1}{2}\sum_{x_j^n}\left| \P\{X_j^n=x_j^n\} - \prod_{i=1}^n p(x_{ji}) \right| \leq \delta_j^{\mathrm{MAC}},
	\end{equation}
	for $j = 1,2$ and some $\delta_1^{\mathrm{MAC}}$, $\delta_2^{\mathrm{MAC}} > 0$. Let $\eps$ be the average probability of error of the multiple access channel code when the channel input distributions are i.i.d. according to $p(x_1)p(x_2)$.
\end{lego}

\begin{lego}[\textbf{BT} $\to$ \textbf{UL C-RAN}]
	an $(R_3,R_4,n)$ Berger--Tung code $(f_3,f_4,\psi)$ for a $p(y_1,y_2)$-source that targets conditional distributions $p(\hat{y}_1 \cond y_1)p(\hat{y}_2 \cond y_2)$, such that, for $(Y_1^n,Y_2^n) \IID p(y_1,y_2)$, the sequences $(\widehat{Y}_1^n,\widehat{Y}_2^n) = \psi\big(g_3(Y_1^n), g_4(Y_2^n)\big)$ satisfy
	\begin{equation}\label{eqn:lossy_uplink_cran}
		\frac{1}{2}\hspace*{-0.1em}\sum_{y_j^n, \, \hat{y}_j^n}\hspace*{-0.1em} \left| \P\{Y_j^n\hspace*{-0.1em}=\hspace*{-0.1em}y_j^n,\widehat{Y}_j^n \hspace*{-0.1em}= \hspace*{-0.1em}\hat{y}_j^n\} \hspace*{-0.1em}-\hspace*{-0.1em} \prod_{i=1}^n  p(y_{ji})p(\hat{y}_{ji} \cond y_{ji})\right| \leq \delta_j^{\mathrm{BT}}
	\end{equation}
	for $j = 1,2$ and some $\delta_1^{\mathrm{BT}}$, $\delta_2^{\mathrm{BT}} > 0$. Furthermore, we assume that $R_3 < C_1$ and $R_4 < C_2$.
\end{lego}

\begin{remark}
    The multiple access channel code can be implemented using two asymmetric channel codes, as described in Appendix~\ref{sec:mac}, and the Berger--Tung code can be implemented using a lossy source code and a Wyner--Ziv code, as described in Appendix~\ref{sec:berger}.
\end{remark}

Fig.~\ref{fig:uplink_cran} shows the block diagram of the uplink C-RAN problem. The multiple access channel code is used to encode the user messages while shaping the channel input signals, whereas the Berger--Tung code is used to compress the received signals over the backhaul links. The central processor first decodes the compressed codewords using the Beger--Tung decoder, then decodes the user messages using the multiple access channel decoder. Following similar bounding techniques seen thus far, the average probability of error of the coding scheme can be bounded as
\[
\P\big\{ \{\widehat{M}_1 \neq M_1\} \cup \{\widehat{M}_2 \neq M_2\} \big\} \leq \eps + \delta_1^{\mathrm{MAC}} +\delta_2^{\mathrm{MAC}} + \delta_1^{\mathrm{BT}} + \delta_2^{\mathrm{BT}}.
\]
The user rate pair is given by $(R_1,R_2)$. Note that the conditions $R_3< C_1$ and $R_4< C_2$ are needed so that the compression rates of the coding scheme do not exceed the backhaul link capacities.

\begin{figure*}[t]
	\centering
	\hspace*{1em}
	\def\svgscale{1.25}
	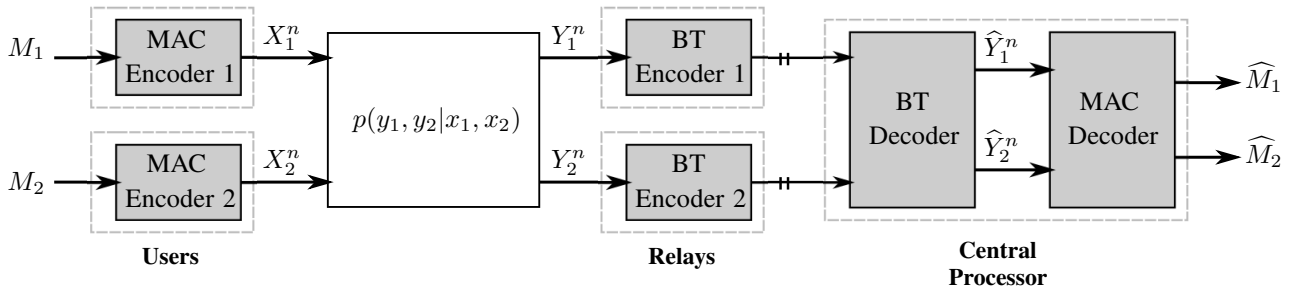
	\caption{Coding scheme for the uplink C-RAN problem using a multiple access channel code and a Berger--Tung code.}
	\label{fig:uplink_cran}
\end{figure*}

\begin{remark} \label{remark:uplink_cran_achievable_rates}
	For sufficiently large $n$, the rate pair $(R_1,R_2)$ can be made arbitrarily close to $\big(I(X_1;\widehat{Y}_1,\widehat{Y}_2),I(X_2;\widehat{Y}_1,\widehat{Y}_2 \cond X_1)\big)$ (e.g., using the multiple access channel coding scheme of Appendix~\ref{sec:mac_coding_scheme}). Similarly, for sufficiently large $n$, the compression rate pair $(R_3,R_4)$ can be chosen to be arbitrarily close to $\big(I(Y_1;\widehat{Y}_1), I(Y_2;\widehat{Y}_2) - I(\widehat{Y}_1;\widehat{Y}_2)\big)$ (e.g., using the Berger--Tung coding scheme of Appendix~\ref{sec:berger_coding_scheme}). It can be checked that these coding and compression rate pairs correspond to a corner point of the compress-and-forward achievable rate region given in~(\ref{eqn:uplink_cran_rate_region}). Note that the decoding at the central processor in Fig.~\ref{fig:uplink_cran} is done in the order $\widehat{Y}_1 - \widehat{Y}_2 - X_1 - X_2$. By exploiting other decoding orders, the coding scheme can achieve other corner points of the rate region~(\ref{eqn:uplink_cran_rate_region}).
\end{remark}

\begin{remark}
	A similar coding scheme can be constructed for the general $K$-user, $L$-relay uplink C-RAN problem using a code for a $K$-user multiple access channel and a Berger--Tung code for a distributed lossy compression problem with $L$ sources.
\end{remark}

\begin{remark}
    Note that the uplink C-RAN coding scheme can be implemented using point-to-point channel codes by constructing the multiple access channel and Berger--Tung codes using their constituent point-to-point channel codes, as described in Appendices~\ref{sec:mac} and~\ref{sec:berger} respectively. Nevertheless, note that our construction of a Berger--Tung code using point-to-point channel codes requires binary sources. Hence, the channel output alphabets $\cY_1$ and $\cY_2$ of the uplink C-RAN model should be binary in this case.
\end{remark}

\subsection{Simulation Results} \label{sec:uplink_cran_simulation}
By implementing the multiple access channel code and the Berger--Tung code using their constituent point-to-point channel codes, the coding scheme for uplink C-RAN can be simulated using point-to-point channel codes. We consider a two-user two-relay model with a binary-quantized Gaussian channel between the users and the relays. That is, the channel output at the relays can be expressed as
\begin{align*} \label{eqn:uplink_cran_model}
    \wtilde{Y}_1 &= X_1 + gX_2 + Z_1, \qquad Y_1 = q_1(\wtilde{Y}_1), \\
    \wtilde{Y}_2 &= X_2 + gX_1 + Z_2, \qquad Y_2 = q_2(\wtilde{Y}_2),
\end{align*}
where $g=0.9$, $(X_1, X_2) \in \{-\sqrt{P},\sqrt{P}\}^2$ are BPSK-modulated signals, $(Z_1, Z_2) \sim \cN({\bf 0},I)$ are independent Gaussian noise components, and $q_1(.)$ and $q_2(.)$ are binary quantizers. The use of binary quantizers is motivated by the case of severely capacity-limited backhaul links (e.g., when $C_1 \leq 1$ and $C_2\leq 1$). 

We first look at the maximum achievable sum-rate of the proposed coding scheme as a function of the power constraint $P$ and the backhaul capacity constraints $C_1$ and $C_2$. For the decoding order considered in Fig.~\ref{fig:uplink_cran} and the corresponding achievable rates (see Remark~\ref{remark:uplink_cran_achievable_rates}), the maximum achievable sum-rate can be expressed as
\begin{equation} \label{eqn:uplink_cran_optimization}
    C_{\mathrm{sum}} \triangleq \underset{\substack{p(x_1)p(x_2)p(\hat{y}_1 \cond y_1)p(\hat{y}_2 \cond y_2) \in \cD(C_1,C_2)}}{\max} \, I(X_1,X_2;\widehat{Y}_1,\widehat{Y}_2),
\end{equation}
where $\cD(C_1,C_2)$ is the set of all product distributions $p(x_1)p(x_2)p(\hat{y}_1 \cond y_1)p(\hat{y}_2 \cond y_2)$ satisfying $C_1>I(Y_1;\widehat{Y}_1)$ and $C_2>I(Y_2;\widehat{Y}_2)-I(\widehat{Y}_1;\widehat{Y}_2)$.
We use the genetic algorithm~\cite{Holland1992, Luke2013} as an efficient heuristic method to perform the optimization in (\ref{eqn:uplink_cran_optimization}). The quantizers $q_1(.)$ and $q_2(.)$ are optimized using the Lloyd--Max algorithm~\cite{Lloyd1982,Max1960}. That is, for any pair of input distributions $p(x_1)$ and $p(x_2)$, the distribution of $\wtilde{Y}_1$ (resp., $\wtilde{Y}_2$) is derived and used to compute the optimal partition points and quantization levels of $q_1(.)$ (resp., $q_2(.)$) using Lloyd--Max. Fig.~\ref{fig:uplink_cran_achievable_rates} shows the plot of $C_{\mathrm{sum}}$ as a function of the power constraint $P$ under different backhaul capacity constraints $C_1$ and $C_2$.

\begin{figure}[t]
	\centering
	\includegraphics[width=\columnwidth]{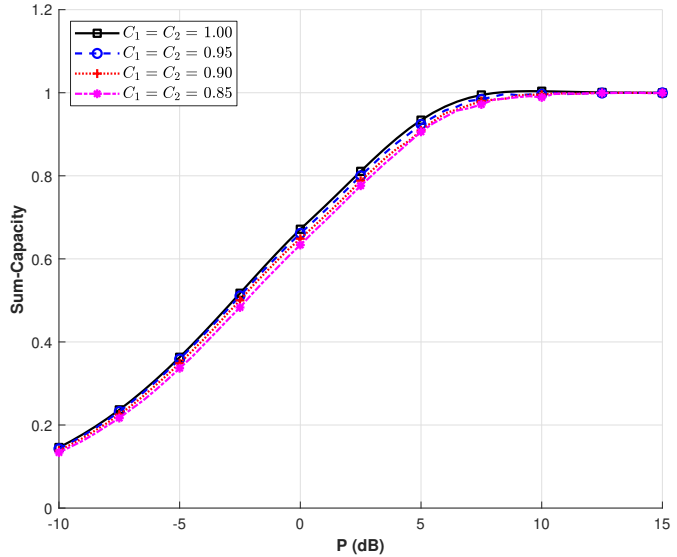}
	\caption{The maximum achievable sum-rate of the proposed coding scheme for the uplink C-RAN problem under different backhaul capacity constraints.}
	\label{fig:uplink_cran_achievable_rates}
\end{figure}

\begin{figure}[t]
	\centering
	\hspace*{-1em}
	\includegraphics[width=\columnwidth]{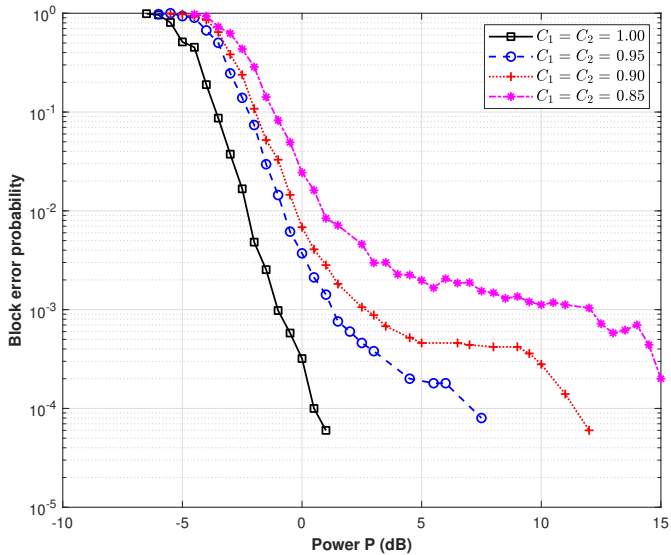}
	\caption{Simulation results of the uplink C-RAN coding scheme for a block length $n=1024$ and sum-rate $R_\mathrm{sum} = 0.25$ under different backhaul capacity constraints.}
	\label{fig:uplink_cran_simulation}
\end{figure}

Next, the proposed coding scheme for the uplink C-RAN problem is simulated for a block length $n=1024$ and a sum-rate $R_{\mathrm{sum}} = 0.25$ using polar codes with successive cancellation decoding as the constituent point-to-point channel codes. Note that the multiple access channel code can be implemented using two asymmetric channel codes, each of which can be implemented using a pair of point-to-point channel codes. Moreover, the Berger--Tung code can be implemented using a lossy source code and a Wyner--Ziv code, each of which can be implemented using a pair of point-to-point channel codes. Hence, the uplink C-RAN coding scheme can be constructed starting from eight polar codes. Let $(R_{11},R_{12},R_{21},R_{22}, R_{31}, R_{32}, R_{41}, R_{42})$ be the rates of the constituent polar codes, where $(R_{11},R_{12})$ and $(R_{21},R_{22})$ are the rates of the two polar codes needed to construct the first and second asymmetric channel code respectively, and $(R_{31},R_{32})$ and $(R_{41},R_{42})$ are the rates of the two polar codes needed to construct the lossy source code and Wyner--Ziv code respectively. The rates are chosen ``close'' to their theoretical limits, i.e., we set\footnote{These theoretical limits can be derived for the uplink C-RAN model as a consequence of Remarks~\ref{remark:lossy_asym_achievability}, \ref{remark:wyner_rates}, and~\ref{remark:asym_rates}.}
\begin{equation}\label{eqn:uplink_cran_rates}
    \begin{aligned}
        R_{11} &=1-H(X_1\cond \widehat{Y}_1,\widehat{Y}_2)-\gamma_r,\\
        R_{12} &= 1-H(X_1), \\
        R_{21} &=1-H(X_2\cond \widehat{Y}_1,\widehat{Y}_2,X_1)-\gamma_r,\\
        R_{22} &= 1-H(X_2), \\
        R_{31} &= 1-H(\widehat{Y}_1 \cond Y_1),\\
        R_{32} &= 1-H(\widehat{Y}_1)-\gamma_c,\\
        R_{41} &= 1-H(\widehat{Y}_2 \cond Y_2),\\
        R_{42} &= 1-H(\widehat{Y}_2 \cond \widehat{Y}_1)-\gamma_c,
    \end{aligned}
\end{equation}
\noindent where $\gamma_r > 0$ and $\gamma_c >0$ are ``back-off'' parameters from the theoretical limit that are used for the polar codes involved in error correction (i.e., not shaping). Thus, the sum-rate attained by the coding scheme is equal to 
\[
R_{11}-R_{12}+R_{21}-R_{22} = I(X_1,X_2;\widehat{Y}_1,\widehat{Y}_2) -2\gamma_r.
\]
Similar to the approach taken in Appendix~\ref{sec:marton_simulation}, the coding scheme targets the distributions $p(x_1)p(x_2)p(\hat{y}_1 \cond y_1)p(\hat{y}_2 \cond y_2)$ that maximize $C_{\mathrm{sum}}$ under the optimization problem of (\ref{eqn:uplink_cran_optimization}), where the power constraint $P$ is set to be equal to the power level $P^{*}$ at which $C_{\mathrm{sum}} = R_{\mathrm{sum}} + 2\gamma_r$.\footnote{As before, the power level $P^{*}$ can be estimated using the plot of Fig.~\ref{fig:uplink_cran_achievable_rates}.} This guarantees that the sum-rate of the coding scheme is equal to $R_{\mathrm{sum}}$. Under the target distributions, the rates of the constituent polar codes are set according to~(\ref{eqn:uplink_cran_rates}). Note that $\gamma_c$ should be chosen so that $R_{31}-R_{32} \leq C_1$ and $R_{41}-R_{42} \leq C_2$. In our simulations over the uplink C-RAN channel model, we used $\gamma_r = 1/8$ and $\gamma_c = 5/32$. For more details about the simulation setup, our code is available on GitHub~\cite{Github}. The block error rate performance of the uplink C-RAN coding scheme is shown in Fig.~\ref{fig:uplink_cran_simulation} for different backhaul capacity constraints $C_1$ and $C_2$. The results demonstrate the practicality of the proposed coding scheme over the uplink C-RAN channel model. Note that the observed error floor is due to binary quantization. 

\end{appendices}

\section*{Acknowledgement}
The authors would like to thank the anonymous reviewers for several valuable comments that helped improve the presentation of this paper. This work was supported in part by the Institute for Information \& Communication Technology Planning \& Evaluation (IITP) grant funded by the Korean government (MSIT) (No. 2018-0-01659, 5G Open Intelligence-Defined RAN (ID-RAN) Technique based on 5G New Radio), in part by the NSERC Discovery Grant No. RGPIN-2019-05448, and in part by the NSERC Collaborative Research and Development Grant CRDPJ 543676-19.

\bibliography{bibliography}
\bibliographystyle{IEEEtran}

\begin{IEEEbiographynophoto}{Nadim Ghaddar (S’16)}
received a B.E. degree in computer and communication engineering from the American University of Beirut, Lebanon, in 2014, an M.S. degree in communication systems from the {\'E}cole Polytechnique F{\'e}d{\'e}rale de Lausanne (EPFL), Switzerland, in 2016, and a Ph.D. degree in Communication Theory and Systems from the University of California San Diego, La Jolla, CA, USA, in 2022. He is currently a postdoctoral research fellow in the Department of Electrical and Computer Engineering at the University of Toronto. His research interests include modern coding theory, information theory, signal processing and wireless communication systems.
\end{IEEEbiographynophoto}

\begin{IEEEbiographynophoto}{Shouvik Ganguly (Member, IEEE)} 
received the B.Tech. degree in electrical engineering from IIT Kanpur, India in 2013 and the Ph.D. degree in electrical engineering from the University of California San Diego (UCSD) in 2020. From 2020—2022, he was with XCOM Labs, San Diego, CA, USA as a member of technical staff. In 2022, he joined Samsung Research America, Plano, TX, USA, where he is currently a staff research engineer. His research interests include network information theory and communication theory.
\end{IEEEbiographynophoto}

\begin{IEEEbiographynophoto}{Lele Wang}
is an Assistant Professor in the Department of Electrical and Computer Engineering at the University of British Columbia, Vancouver. Before joining UBC, she was an NSF Center for Science of Information postdoctoral fellow. From 2015 to 2017, she was a postdoctoral researcher at Stanford University and Tel Aviv University. She received her Ph.D. degree in Communication Theory and Systems at the University of California, San Diego, in 2015. She attended the Academic Talent Program and obtained a B.E. degree at Tsinghua University, China in 2009. Her research interests include information theory, coding theory, communication theory, statistical inference on graphs, high dimensional statistics, and generative models in machine learning. She is a recipient of the 2013 UCSD Shannon Memorial Fellowship, the 2013-2014 Qualcomm Innovation Fellowship, and the 2017 NSF Center for Science of Information Postdoctoral Fellowship. Her PhD thesis "Channel coding techniques for network communication" won the 2017 IEEE Information Theory Society Thomas M. Cover Dissertation Award. She serves as the Vice President of the Canadian Society of Information Theory 2022--2025, an associate editor for the IEEE Transactions on Communications 2023--2026, and a guest associate editor for the IEEE Journal on Selected Areas on Information Theory 2023.
\end{IEEEbiographynophoto}

\begin{IEEEbiographynophoto}{Young-Han Kim (S’99–M’06–SM’12–F’15)}
received his B.S. degree (Hons.) in electrical engineering from Seoul National University, Seoul, South Korea, in 1996, and M.S. degrees in electrical engineering and in statistics and Ph.D. degree in electrical engineering from Stanford University, Stanford, CA, USA, in 2001, 2006, and 2006, respectively. In 2006, he joined the University of California at San Diego, La Jolla, CA USA, where he is currently a Professor in the Department of Electrical and Computer Engineering. Since 2020, he has also been a founding CEO of Gauss Labs Inc., an industrial AI startup company. He has co-authored the book Network Information Theory (Cambridge University Press, 2011) and the monograph Fundamentals of Index Coding (Now Publishers, 2018). His current research interests include data science, machine learning, information theory, and their applications in manufacturing, microelectronics, communications, networking, cryptography, and bioinformatics. Prof. Kim was a recipient of the 2008 NSF Faculty Early Career Development Award, the 2009 US–Israel Binational Science Foundation Bergmann Memorial Award, the 2012 IEEE Information Theory Paper Award, and the 2015 IEEE Information Theory Society James L. Massey Research and Teaching Award for Young Scholars. He served as an Associate Editor of the IEEE Transactions on Information Theory and a Distinguished Lecturer for the IEEE Information Theory Society. He is a foreign member of the National Academy of Engineering of Korea.
\end{IEEEbiographynophoto}

\end{document}